\documentclass[10pt,reqno,final]{amsart}

\usepackage{amsfonts,latexsym,url,amsmath,amssymb,comment,graphicx,epstopdf,enumitem,float,mathtools,pgfplots,pgfplotstable,subcaption,xfrac}
\usepackage{algorithm,algpseudocode}

\usepackage[foot]{amsaddr}
\usepackage{url,centernot}
\usepackage{graphicx,color}
\usepackage[color,notref,notcite]{showkeys}
\usepackage{bm}
\usepackage{stmaryrd}
\usepackage[bookmarksopen,bookmarksdepth=2,breaklinks=true]{hyperref}

\newcommand{\cred}[1]{{\color{red}#1}}

\newcommand{\divergence}{\textnormal{div}}
\newcommand{\BK }{\boldsymbol{K}}  	
\newcommand{\BU}{\boldsymbol{U}}   	
\newcommand{\dof}{{X}}				
\newcommand{\dpdof}{\underline{\mathrm{p}}}		
\newcommand{\dc}{\underline{\mathrm{c}}}		
\newcommand{\tf}{\underline{\mathrm{w}}}		
\newcommand{\pc}{{\mathrm{c}}}		
\newcommand{\pf}{{\mathrm{w}}}		
\newcommand{\pp}{{\mathrm{p}}}		

\newcommand{\dbf}{a}				
\newcommand{\stab}{s}				
\newcommand{\ro}{\boldsymbol{r}}	
\newcommand{\RO}{\boldsymbol{R}}	
\newcommand{\du}{\underline{\mathrm{u}}}		
\newcommand{\pu}{\mathrm{u}}	
\newcommand{\dv}{\underline{\mathrm{v}}}
\newcommand{\pv}{\mathrm{v}}
\newcommand{\ad}{\mathbf{\mathcal{G}}}		
\newcommand{\fv}{\mathrm{U}}					
\newcommand{\RU}{\boldsymbol{\mathrm{U}}}					
\newcommand{\nv}{\boldsymbol{\mathrm{n}}}
\newcommand{\bLambda}{\boldsymbol{\Lambda}}
\newcommand{\BI}{\boldsymbol{I}}		
\newcommand{\bb}{\boldsymbol{b}}		
\newcommand{\bkappa}{\boldsymbol{\kappa}}	

\newcommand{\eye}{{\normalfont \textbf{I}}}
\newcommand{\hs}{{n + \sfrac{1}{2}}}
\newcommand{\dch}{\dc^\hs} 
\newcommand{\pch}{\pc^\hs} 
\newcommand{\br}{s}
\newcommand{\BD}{\boldsymbol{D}}
\newcommand{\dD}{\mathrm{\mathbf{D}}}
\newcommand{\R}{\mathbb{R}}

\newcommand{\N}{\mathbb{N}}

\newcommand{\Poly}{\mathbb{P}}

\newcommand{\scriptC}{\mathcal{C}}

\newcommand{\scriptM}{\mathcal{M}}
\newcommand{\scriptB}{\mathcal{B}}

\newcommand{\scriptT}{\mathcal{T}}
\newcommand{\scriptF}{\mathcal{F}}

\newcommand{\fspd}{$\textnormal{ft}^2/\textnormal{d}$ay}

\newcommand{\lip}[3]{\int_{#1} #2 #3 }
\newcommand{\lipd}[3]{\int_{#1} #2 \cdot #3 }

\theoremstyle{plain}
\newtheorem{theorem}{Theorem}
\newtheorem{lemma}[theorem]{Lemma}

\newtheorem{remark}[theorem]{Remark}
\newtheorem{definition}[theorem]{Definition}
\newtheorem{test}{Test}
\numberwithin{equation}{section}
\numberwithin{theorem}{section}
\numberwithin{test}{subsection}

\algnewcommand{\LineComment}[1]{\State \(\triangleright\) #1}

\graphicspath{ {img/} }

\begin{document}
\title[Arbitrary order scheme for porous, miscible displacement]{An arbitrary order scheme on generic meshes for miscible
displacements in porous media}

\author{Daniel Anderson}
\address{School of Mathematical Sciences, Monash University, Clayton, Victoria 3800, Australia.
\texttt{daniel.anderson@monash.edu}}
\author{J\'er\^ome Droniou}
\address{School of Mathematical Sciences, Monash University, Clayton, Victoria 3800, Australia.
\texttt{jerome.droniou@monash.edu}}
\date{\today}

\maketitle

\begin{abstract}
	We design, analyse and implement an arbitrary order scheme applicable to generic meshes for a coupled elliptic-parabolic PDE system describing miscible displacement in porous media. The discretisation is based on several adaptations of the Hybrid-High-Order (HHO) method due to Di Pietro et al. [{\it Computational Methods in Applied Mathematics}, 14(4), (2014)]. The equation governing the pressure is discretised using an adaptation of the HHO method for variable diffusion, while the discrete concentration equation is based on the HHO method for advection-diffusion-reaction problems combined with numerically stable flux reconstructions for the advective velocity that we have derived using the results of Cockburn et al.\ [{\it ESAIM: Mathematical Modelling and Numerical Analysis, 50(3), (2016)}]. We perform some rigorous analysis of the method to demonstrate its $L^2$ stability under the irregular data often presented by reservoir engineering problems and present several numerical tests to demonstrate the quality of the results that are produced by the proposed scheme.
\end{abstract}

{\small

%
}

\section{Introduction}

The single-phase flow of incompressible, miscible fluids in a porous medium, arising in the modelling of enhanced oil recovery is described by a coupled system of non-linear elliptic-parabolic equations on the pressure and concentration of invading solvent, sometimes referred to as the Peaceman model as derived by Peaceman in \cite{peaceman1962numerical} (not to be confused with the Peaceman model of wells). Given its complexity, the behaviour of this system can only be fully understood through numerical approximations. These approximations must account for the specificities of flows in porous media: discontinuous data, non-Cartesian (and possibly non-conforming) grids, etc.

In this paper, we design an arbitrary-order numerical scheme for the Peaceman model. The stability of
the approximation is established, and numerical results are provided. Based on these, we provide advice
on how to choose the various parameters of the scheme (spatial order of approximation, time-stepping method, etc.)
to optimise the accuracy of the result and the overall computational cost.

Existence of a weak solution to this miscible flow model has been first established in \cite{feng1995existence}, and
then extended in \cite{chen1999mathematical} to include gravity effects and various boundary conditions.
In practice, the molecular diffusion is much smaller than the dispersion effect (driven by the Darcy viscosity),
and often neglected in numerical simulation. In that case, the parabolic equation takes on a degenerate form,
which makes the analysis of the model even more complex. \cite{az04} establishes the existence of a solution
in the case of a vanishing molecular diffusion and regular source terms. Given the scale of the reservoir
and the well bores, it is customary in simulations to consider wells concentrated on measures (Dirac measures
in 2D, measures along lines in 3D). The corresponding mathematical analysis has been carried out in \cite{droniou2014miscible} for non-zero molecular diffusion and in \cite{DT16} for vanishing molecular diffusion.

A variety of numerical schemes have been considered, with or without convergence
analysis, for the Peaceman model. Finite-element based methods coupled with a modified method of characteristic
for the advective part of the flow, are applied in \cite{ewing1983simulation} and
analysed in \cite{ERW84}. Another method of characteristics, the Eulerian--Lagrangian Adjoint Method,
is coupled with $\mathbb{RT}_k$ mixed finite elements (for the pressure) and $\mathbb{Q}_k$ finite elements (for the
concentration) in  \cite{wang2000approximation,wang2008optimal}. A method combining $\mathbb{RT}_k$
finite elements and discontinuous Galerkin (dG) schemes is analysed in \cite{DGPeaceman}. Optimal error estimates for conforming
$\mathbb{P}_k$ finite elements on triangles are established in \cite{GalerkinPeaceman}, based
on maximal regularity properties of the continuous model.
Finite-element based methods are natural
and well known, but suffer from restrictions on the mesh geometries, which must be conforming and essentially
made of triangles or squares. Other schemes have been considered to deal with meshes with generic geometries,
as encountered in geophysical applications. In the last few
years, schemes have been developed to be applicable on such generic grids. They are mostly based
on finite volume techniques, which have the advantage of providing conservative approximations
of the Darcy fluxes \cite{droniou2014finite}, that can be used in the discretisation of the advective terms.
In \cite{droniou2007convergence},
the mixed finite volume (MFV) of \cite{droniou2006mixed} is adapted to the Peaceman model, its convergence analysed
and numerical tests are provided; as shown in \cite{DEGH09}, this method can be embedded
in a larger family, the hybrid mimetic method (HMM) family, that also contains the SUSHI scheme
of \cite{EYM10} and the mixed-hybrid mimetic finite difference methods of \cite{Brezzi.Lipnikov.ea:05}.
Discrete duality finite volume (DDFV) methods are considered in \cite{DDFV1,DDFV2}.
HMM and DDFV are finite volume schemes with first-order approximation properties, and rely for the
miscible displacement model on upwinding to stabilise
the advective terms; this raises the concern of an over-diffusion of the transition layer between
the invading solvent and the residing oil.

In this work, we develop an arbitrary-order numerical scheme for the Peaceman model, which is applicable
on generic grids. The scheme is an adaptation of the hybrid high-order (HHO) method, initially developed
for stationary diffusion PDEs in \cite{di2014arbitrary,di2015hybrid} and then extended to stationary
advection--diffusion--reaction models in \cite{droniou2015discontinuous}. The HHO method can be seen as
a higher-order extension of the HMM method, and is very close to virtual element methods \cite{Beirao-da-Veiga.Brezzi.ea:13}, to non-conforming mimetic finite difference methods \cite{Lipnikov-Manzini:2014} and to hybridizable dG methods \cite{cockburn2016bridging}. 
The initial degrees of freedom of the HHO method are scalar valued polynomials of arbitrary order $k$ on the cells and faces of the mesh. The cell degrees of freedom can however be eliminated by a local
static condensation procedure, and only the face degrees of freedom remain coupled, in a way that
is however highly parallelisable. The HHO scheme is built on a collection of high-order local reconstruction operators that mimic the quantities present in the weak formulation of the continuous equation. Our executive summary is as follows:
\begin{itemize}
\item There is a real advantage in going for a higher order method.
The choice $k=0$ leads to strong grid effects, that are mostly eliminated by taking
$k=1$. This choice $k=1$ seems to be optimal in terms of accuracy vs. computational
cost; the choices $k=2,3,\ldots$ increase the computational cost with only minor further improvements of the
accuracy.
\item A Crank-Nicolson or second order backward differentiation formula (BDF) time-stepping is sufficient to obtain good results; high-order BDF tend to become
unstable unless the time step is reduced a lot, and do not lead to perceptible improvement (even considering
higher order spatial approximation, that is $k\ge 2$).
\item The specific mesh geometry is mostly irrelevant to the quality of the numerical approximation,
which mostly seems to only depend on the number of faces of the mesh (which is expected, the face
unknowns being the main unknowns in the HHO method).
\end{itemize}

Let us conclude by describing the organisation of the paper. In the next section, we describe
the continuous miscible displacement model, both in strong and weak form. The scheme is described
in Section \ref{sec:scheme}, starting from the time stepping, designed in a classical way to decouple the pressure and concentration equations. As explained above, the HHO method is built on local polynomial
spaces and reconstruction operators; these are respectively described in Sections \ref{sec:discretisation.spaces} and \ref{sec:reconstruction.operators}. The discretised pressure equation
is then presented. In order to discretise the convection terms appearing in
the concentration equation, cell Darcy velocities and corresponding fluxes have
to be designed from the numerical solution of the pressure equation. The reconstruction
of this velocity and fluxes is described in Section \ref{sec:discrete.fluxes}.
We note that this reconstruction has to be carefully performed to preserve the
scheme stability; in particular, this implies discretising the pressure equation at
an order twice the order chosen for the concentration equation. The numerical approximation
of the latter is described in Section \ref{sec:concentration.equation}.
Existence, uniqueness and stability results for our complete scheme are stated in Theorem \ref{thm:stability_of_crank_nicolson} at the end of Section
\ref{sec:scheme}.
Numerical tests are provided in Section \ref{sec:tests}. We analyse in particular
the effect (in terms of cost as well as efficiency)
of varying the spatial degree of the method and of having to use distorted polygonal meshes.
This analysis is done on test-cases involving homogeneous or discontinuous permeability,
with various mesh geometries and by considering both the general eye-ball quality of the solution (compared
with the expected behaviour), as well as quantitative assessments based
on the variation, with respect to the polynomial degrees, of the amount of oil recovered after 10 years. The executive summary above is backed up by the extensive numerical results in
this section.
A short conclusion is given in Section \ref{sec:conclusion}, and the proof of the existence,
uniqueness and stability result is given in the first appendix, Section \ref{sec:appen}. These proofs
show in particular the importance of choosing, in the discretisation
of the concentration equation, Darcy velocity and fluxes adapted to the discretisation
of the pressure equation. A second appendix, Section \ref{appen:implementation}, describes
the practical implementation of the scheme, and provides a link to the code
we developed for the numerical tests.

\section{The continuous model}\label{sec:continuous}

We introduce the following notations that will be used to describe the model, and then present the aforementioned system.

\begin{center}
	\begin{tabular}{ l l }
		$ d \in \{2,3\}$ & the number of dimensions considered in the model, \\
		$\Omega \subset \R^d$ & a bounded Lipschitz domain representing the reservoir,  \\
		$(0, t_f) \subset \R$ & the time interval on which we consider the problem, \\
		$p : (0,t_f) \times \Omega \to \R$ & the pressure in the mixture, \\
		$\BU : (0,t_f) \times \Omega \to \R^d$ & the Darcy velocity of the fluid, \\
		$c : (0,t_f) \times \Omega \to \R$ & the concentration of the invading solvent in the reservoir, \\
		$\hat{c} : (0,t_f) \times \Omega \to \R$ & the concentration of solvent as it is injected, \\
		$\Phi : \Omega \to \R$ & the porosity of the medium, \\
		$\mu : [0,1] \to \R$ & the viscosity of the fluid mixture at a given concentration, \\
		$\BK : \Omega \to \R^{d\times d}$ & the absolute permeability tensor of the medium, \\
		$\BD : \Omega \times \R^d \to \R^{d\times d}$ & the diffusion-dispersion tensor of the medium, \\
		$q^+ : (0,t_f) \times \Omega \to \R$ & the source term corresponding to the injection well, \\
		$q^- : (0,t_f) \times \Omega \to \R$ & the source term corresponding to the production well.
	\end{tabular}
\end{center}
\medskip

Taking the effects of gravity to be negligible, the model reads:
\begin{subequations}\label{scheme:all.eqn}
\begin{align} 
&\begin{cases} \label{eqn:peaceman_pressure}
\divergence (\BU) = q^+ - q^- & \text{in } (0,t_f) \times \Omega , \\[1ex]
\BU = -\cfrac{\BK(x)}{\mu(c)}\nabla p	& \text{in } (0,t_f) \times \Omega ,
\end{cases}  \\[1ex]
&\Phi(x)\frac{\partial c}{\partial t} - \divergence(\BD(x,\BU)\nabla c - c\BU) + q^-c= q^+\hat{c} \qquad \text{in } (0,t_f) \times \Omega . \label{peaceman_concentration}
\end{align}

This system is comprised of two very natural parts. The pressure equation \eqref{eqn:peaceman_pressure} is an anisotropic diffusion equation with diffusivity $-\frac{\BK}{\mu(c)}$. For simplicity of notation, we may instead write
\begin{equation} \label{eqn:permeability-viscosity-ratio}
\bkappa = \cfrac{\BK(x)}{\mu(c)}.
\end{equation}
The concentration equation \eqref{peaceman_concentration} describes the convection of the fluid mixture via an advection-diffusion-reaction equation with diffusivity $\BD$ and advective velocity corresponding to the Darcy velocity of the fluid mixture $\BU$. We will now briefly summarise the models used for the data. Following \cite{peaceman1966improved}, Peaceman derived the diffusion-dispersion tensor $\BD$ to be
\begin{equation} \label{eqn:diffusion_dispersion_tensor}
\BD(x,\BU) = \Phi(x)\left(d_m \eye + |\BU|\left(d_l E(\BU) + d_t(\eye - E(\BU)) \right) \right),
\end{equation}
where $d_m$ is the molecular diffusion coefficient, $d_l$ and $d_t$ are the longitudinal and transverse dispersion coefficients, and $E(\BU)$ is an orthogonal projection in the direction of the Darcy velocity, given by the outer-product
\begin{equation}
E(\BU) = \frac{\BU \otimes \BU}{|\BU|^2} =  \left[\frac{\BU_i \BU_j}{|\BU|^2}\right]_{1\leq i,j \leq d} .
\end{equation}
Physical experiments reveal that the longitudinal dispersion $d_l$ is far stronger than the transverse dispersion $d_t$, and that the molecular dispersion $d_m$ is negligible in comparison \cred{\cite{WLELQ00}}.

The viscosity of the mixture is determined by the mixing rule given in \cite{koval1963method}:
\begin{equation} \label{eqn:viscosity}
\mu(c) = \mu(0)\left(1 + (M^\frac{1}{4} - 1)c\right)^{-4} \qquad c \in [0,1],
\end{equation}
where $\mu(0)$ is the viscosity of the oil and $M$ is the mobility ratio between the oil and the injected solvent, given by $M = \frac{\mu(0)}{\mu(1)}$. 

In order to maintain a balance of mass in the domain, the boundary of the reservoir $\partial \Omega$ is taken to be impermeable. Consequently, we include, denoting by $\nv$ the unit exterior normal to $\partial\Omega$, homogeneous no flow Neumann boundary conditions:
\begin{equation}
\begin{cases}
\BU\cdot \nv = 0 & \text{on } (0,t_f)\times \partial \Omega,  \\
D \nabla c\cdot \nv = 0 & \text{on } (0,t_f)\times \partial \Omega.
\end{cases}
\end{equation}
Additionally, in order to satisfy the no flow boundary conditions and maintain mass, we must further impose that our injection and production source terms are compatible:
\begin{equation} \label{eqn:source_compatability}
\int_\Omega q^+(t,x)dx = \int_\Omega q^-(t,x)dx \qquad \text{in } (0, t_f).
\end{equation}
We must also prescribe an initial condition:
\begin{equation}
c(0,x) = c_0(x).
\end{equation}
In practice we take the initial concentration in the well to be $0$ everywhere. It is also usual to take $\hat{c}$, the concentration at the injection well to be $1$.

Lastly, since the pressure is only defined up to an arbitrary constant, we normalise $p$ by the following condition:
\begin{equation} \label{eqn:pressure_normalization}
\int_\Omega p(\cdot,x)dx = 0 \qquad \text{in } (0,t_f).
\end{equation}
\end{subequations}


\subsection{Weak formulation}\label{sec:weak.formulation}

Noting that the geometry of a typical oil reservoir will contain many geological layers of varying porosity and permeability, it is very important to take into account the fact that these quantities can not be assumed to be smooth and continuous everywhere. The following assumptions given in \cite{droniou2007convergence} are reasonable and allow us to devise both analytically sound and physically acceptable solutions.

\begin{subequations} \label{eqn:regularity_assumptions}
	\begin{equation}
	\begin{gathered} \label{eqn:permeability_assumptions}
	\Phi \in L^\infty(\Omega) \text{ such that there exists $\Phi_* > 0$ satisfying} \\[-0.3em]
	\text{ $\Phi_* \leq \Phi \leq \Phi_*^{-1}$ almost everywhere in $\Omega$},
	\end{gathered}
	\end{equation}
	
	\begin{equation}
	\begin{gathered}
	q^+, q^- \in L^\infty (0, t_f; L^2(\Omega)) \text{ are non-negative and compatible, i.e.}  \\
	\int_{\Omega} q^+(\cdot, x)dx = \int_\Omega q^-(\cdot, x) dx \text{ for almost every $t\in(0, t_f)$},
	\end{gathered}
	\end{equation}
	
	\begin{equation}
	\begin{gathered}
	\BK : \Omega \to M_d(\R) \text{ is a bounded matrix-valued function} \\[-0.3em]
	\text{admitting symmetric, uniformly coercive values, that is, }\\[-0.3em] \exists \alpha_K > 0 \text{ s.t. } \BK(x)\xi \cdot \xi \geq \alpha_K |\xi|^2 
	\text{ for almost every $x\in\Omega$ and all}\\[-0.3em]
	\xi \in \R^d,\ \exists \Lambda_K > 0 \text{ s.t. }
	|\BK(x)| \leq \Lambda_K \text{ for almost every $x\in\Omega$,}
	\end{gathered}
	\end{equation}
	
	\begin{equation}
	\begin{gathered}
	\BD : \Omega \times \R^d \to M_d(\R) \text{ is given by \eqref{eqn:diffusion_dispersion_tensor} with	
	$d_l, d_m, d_t > 0$}.
	\end{gathered}
	\end{equation}
	
	\begin{equation}
	\begin{gathered}
	\mu \in \scriptC(\R) \text{ is positively bounded, that is } 0 \leq a \leq \mu(c) \leq b \text{ for} \\[-0.3em]
	\text{positive constants $a, b$ for all $c \in \R$. This is clearly satisfied by \eqref{eqn:viscosity}}, 
	\end{gathered}
	\end{equation}
	
	\begin{equation}\begin{aligned}
	&\hat{c} \in L^\infty((0,t_f)\times \Omega) \text{ satisfies } 0 \leq \hat{c} \leq 1 \text{ almost everywhere},\\
	&c_0 \in L^\infty(\Omega) \text { satisfies } 0 \leq c_0 \leq 1 \text{ almost everywhere}.
	\end{aligned}
\end{equation}
\end{subequations}

Under the regularity assumptions \eqref{eqn:regularity_assumptions}, the existence of a weak solution to \eqref{eqn:peaceman_pressure}--\eqref{peaceman_concentration} is established in \cite{feng1995existence,chen1999mathematical}.
The HHO scheme presented in Section \ref{sec:scheme} consists in discretising the following
equations, satisfied by this weak solution $(p, \BU, c)$:
\begin{equation} 
\label{eqn:weak_spaces}
\left\{\begin{split}
&p \in L^\infty(0,t_f;H^1(\Omega))\,,\;\BU \in L^\infty(0,t_f;L^2(\Omega))^d,  \\
&c \in \scriptC([0,t_f];L^2(\Omega)) \cap L^2(0,t_f;H^1(\Omega))\mbox{ with }\\
& \Phi \partial_t c \in L^2(0,t_f;(W^{1,4}(\Omega))')\mbox{ and }c(0) = c_0,
\end{split}\right.
\end{equation}
\begin{equation} \label{weak_pressure_equation}
\left\{
\begin{split}
&\text{For almost every } t \in (0,t_f), \text{ for all } \varphi \in H^1(\Omega), \\[0.5ex]
&- \int_\Omega \BU(t,\cdot)\cdot \nabla \varphi = \int_\Omega \left(q^+(t,\cdot) - q^-(t,\cdot) \right)\varphi,
\end{split}
\right.
\end{equation}
and
\begin{equation} \label{weak_concentration_equation}
\left\{
\begin{split}
&\text{For almost every } t \in (0,t_f), \text{ for all } \varphi \in W^{1,4}(\Omega), \\[0.5ex]
&\langle \Phi \partial_t c(t),\varphi\rangle_{(W^{1,4})',W^{1,4}} + \int_\Omega \BD(\cdot,\BU(t, \cdot))\nabla c(t, \cdot) \cdot \nabla \varphi \\
& \qquad- \int_\Omega c(t, \cdot) \BU(t, \cdot) \cdot \nabla \varphi + \int_\Omega q^-(t, \cdot)c(t, \cdot)\varphi = \int_\Omega q^+(t, \cdot)\hat{c}(t)\varphi.
\end{split}
\right.
\end{equation}

\section{Numerical scheme}\label{sec:scheme}

The scheme for the pressure equation consists in the standard HHO method for variable diffusion problems \cite{di2015hybrid}, taking into account the coupling of the pressure with the concentration equation. Adapting \cite{cockburn2016bridging}, we derive conservative discrete version of the Darcy velocity and its fluxes that are required to discretise the advective
term $\divergence(\BU c)$ in the concentration equation. The spatial terms in this equation are discretised by
using the HHO method for linear advection--diffusion--reaction \cite{droniou2015discontinuous}, incorporating the aforementioned discrete Darcy velocity and fluxes. The time stepping presented here is based on the Crank-Nicolson
method, due to its strong stability properties \cite{leveque2007finite} (we also tested BDF time steppings).

\subsection{Time stepping}

We define our time-stepping as follows, let $N \in \N$ be the number of time-steps to be taken and let
\begin{equation}
\Delta t = \frac{t_f}{N}, \qquad t^n = n\Delta t, \qquad n = 0,1,...,N.
\end{equation}

Let us denote the pressure, concentration and Darcy velocity at time-step $n$ by
$p^n = p(t^n, \cdot)$, $c^n = c(t^n, \cdot)$ and $\BU^n = \BU(t^n, \cdot) = -\bkappa(c^n, \cdot)\nabla p^n$.
The Crank-Nicolson time stepping of \eqref{peaceman_concentration} consists in writing
\begin{equation} \label{eqn:crank_nicolson_continuous_concentration}
\begin{split}
\Phi(x)\frac{c^{n+1} - c^n}{\Delta t} &- \divergence(D(x,\BU^\hs)\nabla c^\hs - c^\hs\BU^\hs) + q^-c^\hs \\&= q^+\hat{c}(t^\hs),
\end{split}
\end{equation}
where the intermediate time-stepped values are defined by
\begin{equation} \label{eqn:half_time_concentration}
\xi^\hs = \frac{\xi^{n+1} + \xi^n}{2} \qquad\mbox{($\xi=t$, $c$ or $\BU$)}.
\end{equation}
The problem data, eg.\ the well terms $q^+$ and $q^-$ are evaluated at $t^\hs$. 
The formulation \eqref{eqn:crank_nicolson_continuous_concentration} leads to the following
direct relation between $c^n$ and $c^\hs$:
\begin{equation}
\label{eqn:crank_nicolson_half_step_continuous_concentration}
\begin{split}
\Phi(x)\frac{2(c^\hs - c^n)}{\Delta t} &- \divergence(D(x,\BU^\hs)\nabla c^\hs - c^\hs\BU^\hs) \\& + q^-c^\hs = q^+\hat{c}(t^\hs).
\end{split}
\end{equation}
We note that this formulation is equivalent to performing a half time-step with an implicit Euler scheme to obtain $c^\hs$ and then linearly extrapolating $c^n$ with $c^\hs$ to obtain $c^{n+1}$ using \eqref{eqn:half_time_concentration}.

\subsection{Discretisation spaces}\label{sec:discretisation.spaces}

Let us briefly introduce the notion of a mesh and define the polynomial spaces central to the scheme.
\begin{definition} \label{definition_mesh}
	A mesh $\scriptM$ of $\Omega \subset \R^d$ is a collection of cells and faces $(\scriptT, \scriptF)$ where:
	\begin{enumerate}
		\item $\scriptT$ (The cells or control volumes) is a subdivision of $\Omega$ into a disjoint family of open, non-empty polygons (or polyhedra in higher dimensions). Formally,
		\begin{equation}
		\bigcup_{T\in\scriptT} \overline{T} = \overline{\Omega}.
		\end{equation}
		\item $\scriptF$ (The faces or edges) is a disjoint family of non-empty affine subsets of $\Omega$, with positive $(d-1)$ dimensional measure such that for each cell $T \in \scriptT$, there exists $\scriptF_T \subset \scriptF$ where
		\begin{equation}
		\bigcup_{F \in \scriptF_T}\overline{F} = \partial T,
		\end{equation}
		such that each face $F \in \scriptF$ borders exactly one or two cells.
	\end{enumerate}
\end{definition}
Additionally, we may employ the following notation
\begin{center}
	\begin{tabular}{ l l }
		$\scriptF_T \subset \scriptF$ & the faces that border the cell $T$, \\
		$\scriptF_b \subset \scriptF$ & the faces on the boundary of the domain $\Omega$,  \\
		$|T|$ or $|F|$ & the Lebesgue measure of the cell $T$ or face $F$, \\
		$\overline{x}_T$ & the centre of mass of the cell $T$,  \\
		$\nv_{TF}$ & a normal vector to the face $F$ facing outward from the cell $T$,  \\
		$h_T$ or $h_F$ & the diameter of a cell $T$ or a face $F$,  \\
		$h = \max_{T\in\scriptT} h_T$ & the maximum diameter of any cell.
	\end{tabular}
\end{center}

Take $\scriptM_h = (\scriptT_h, \scriptF_h)$ a generic polygonal/polytopal mesh of $\Omega$ as above.
The degrees of freedom of the scheme are scalar valued polynomials on the cells $T \in \scriptT_h$ and the faces $F \in \scriptF_h$. No continuity conditions between cells and faces, or
between cells and cells,
are imposed on the degrees of freedom. Selecting an integer $m \geq 0$, the following notation will help us describe the polynomial spaces; here, $K$ is a set of dimension $l$ (that is,
the affine space spanned by $K$ has dimension $l$).

\begin{center}
	\begin{tabular}{ l l }
		$\Poly^m(K)$	& the space of $l$-variate polynomials of degree $\leq m$\\
		$\nabla \Poly^m(K)$	&	the space of functions $\{ \nabla u : u \in \Poly^m(K) \}$,  \\
		$\pi^m_K : L^2(K) \to \Poly^m(K)$		&	the $L^2$ orthogonal projector onto $\Poly^m(K)$,
	\end{tabular}
\end{center}
where the $L^2$ orthogonal projector is given by:
\begin{equation} \label{eqn:projection_definition}
\pi^m_K(u) = v \in \Poly^m(K) \text{ is such that, for all } w \in \Poly^m(K),\int_K (u - v)w = 0.
\end{equation}
The spaces of degrees of freedom are then given as follows.

\begin{definition} Let $\scriptM_h$ be a mesh.
	For each cell $T \in \scriptT_h$, the space of local degrees of freedom on $T$ is defined by
	\begin{equation}
	\dof_T^m = \Poly^m(T) \times \left\{ \bigtimes_{F\in\scriptF_T} \Poly^m(F) \right\}.
	\end{equation}
	The space of global degrees of freedom on the mesh is defined as
	\begin{equation}
	\dof^m_h = \left\{ \bigtimes_{T\in\scriptT_h} \Poly^m(T) \right\} \times \left\{ \bigtimes_{F\in\scriptF_h} \Poly^m(F) \right\}.
	\end{equation}
\end{definition}

For a given $k\ge 0$, these spaces will be used to approximate both the pressure (with $m=2k$)
and the concentration (with $m=k$).
An arbitrary set of global degrees of freedom is denoted by
\begin{equation}
\du_h = ((\pu_T)_{T\in\scriptT_h}, (\pu_F)_{F\in\scriptF_h}) \in \dof^m_h,
\end{equation}
and, similarly, a set of  local degrees of freedom is
\begin{equation}
\du_{T} = (\pu_T, (\pu_F)_{F\in\scriptF_T}) \in \dof^m_T.
\end{equation}
For a given $\du_h\in \dof^m_h$, we define $\pu_h\in L^2(\Omega)$ as the piecewise polynomial
function given by $(\pu_h)_{|T}=\pu_T$ for all $T\in \scriptT_h$.
The space of degrees of freedom with zero average is then
\begin{equation}
\dof^m_{h,*} = \left\{ \du_h \in \dof^m_h\,:\, \int_{\Omega} \pu_h = 0 \right\}.
\end{equation}

\subsection{Reconstruction operators}\label{sec:reconstruction.operators}

We now introduce the local reconstruction operators which are central to the scheme. The cornerstone of the HHO method is a high-order local gradient reconstruction operator based on the cell and face polynomial degrees of freedom.

\begin{definition} \label{def:local_reconstruction}
	Let $\bLambda$ be a bounded, real, symmetric, coercive tensor-valued function on $\Omega$. Take a cell $T \in \scriptT_h$ of $\scriptM$. The local reconstruction operator $\ro_{T,\bLambda}^{m+1} : \dof_T^m \to \Poly^{m+1}(T)$ is defined such that, for a bundle of local degrees of freedom $\du_{T} = (\pu_T, (\pu_F)_{F\in\scriptF_T}) \in \dof_T^m$ and any test function $w \in \Poly^{m+1}(T)$,
	\begin{equation} \label{eqn:gradient_reconstruction}
	\lipd{T}{\bLambda\nabla\ro_{T,\bLambda}^{m+1} \du_{T}}{\nabla w} = \lipd{T}{\bLambda \nabla \pu_T}{\nabla w} + \sum_{F\in\scriptF_T} \lip{F}{(\pu_F - \pu_T)}{\nabla w \cdot (\bLambda \nv_{TF})},
	\end{equation}
	and
	\begin{equation} \label{eqn:local_reconstruction}
	\int_T \ro_{T,\bLambda}^{m+1} \du_{T} = \int_T \pu_T.
	\end{equation}
\end{definition}

In addition to the local reconstruction operator, we define the local high-order correction operator.
\begin{definition}
	The high-order correction operator $\RO_{T,\bLambda}^{m+1} : \dof_T^m \to \Poly^{m+1}(T)$ is defined such that
	\begin{equation} \label{high_order_potential}
	\RO_{T,\bLambda}^{m+1} \du_{T} = \pu_T + (\ro^{m+1}_{T,\bLambda} \du_{T} - \pi^m_T\ro_{T,\bLambda}^{m+1} \du_{T}).
	\end{equation}
	We remark that the bracketed term of \eqref{high_order_potential} is orthogonal to $\Poly^m(T)$ by construction, and hence this operator can be seen as adding a high-order orthogonal correction to the cell unknown $\pu_T\in \Poly^m(T)$.
\end{definition}

\subsection{The pressure equation}\label{sec:pressure.equation}

Let us fix an integer $k\ge 0$. The HHO scheme for \eqref{scheme:all.eqn} consists in discretising
the pressure in $\dof^{2k}_h$ and the concentration in $\dof^k_h$. The choice of
an order $2k$, instead of $k$, for the pressure is driven by stability considerations,
which are made clear in the proof of Theorem \ref{thm:stability_of_crank_nicolson} (see Remark \ref{rem:2k.flux}).

To write the scheme on the pressure, we introduce the following local bilinear forms:
$\dbf_{T,\bLambda} : \dof^{2k}_T \times \dof^{2k}_T \to \R$ and $\stab_{\bLambda,T} : \dof^{2k}_T \times \dof^{2k}_T \to \R$, defined by
\begin{align} \label{eqn:pressure_diffusion_local_bilinear_form}
\dbf_{T,\bLambda}(\du_{T}, \tf_{T}) &= \lipd{T}{\bLambda \nabla\ro_{T, \bLambda}^{2k+1} \du_{T}}{\nabla\ro_{T, \bLambda}^{2k+1} \tf_{T}} + s_{\bLambda,T}(\du_{T}, \tf_{T}),  \\
\stab_{\bLambda,T}(\du_{T}, \tf_{T}) &= \sum_{F\in\scriptF_T} \cfrac{{\bLambda_{TF}}}{h_F}\lip{F}{\pi^{2k}_F(\pu_F - \RO^{2k+1}_{T,\bLambda} \du_{T})}{\pi^{2k}_F(\pf_F - \RO^{2k+1}_{T,\bLambda} \tf_{T})}, \label{eqn:diffusion_stabilisation}
\end{align}
where 
\begin{equation}\label{def:LambdaF}
{\bLambda_{TF}} = \|\nv_{TF} \cdot {\bLambda_{|T}} \nv_{TF} \|_{L^\infty(F)}
\end{equation}
is a controlling factor for the size of $\bLambda$ (in cell $T$) across the face $F$. The first term of $\dbf_{T,\bLambda}$ can readily be recognised as a discrete analogue of the weak diffusive terms of the weak formulation. The function $\stab_{\bLambda,T}$ is the diffusive stabilisation term whose purpose is to enforce a least-squares penalty between face unknowns, and the projection of the high-order correction of the cell unknown. This is required to ensure that the cell and face unknowns are related, and to ensure that the global bilinear form defined in \eqref{eqn:global_pressure_diffusion_bilinear_form} below is symmetric positive definite on $\dof^{2k}_{h,*}$ (see Lemma \ref{lemma:diffusive_bilinear_form_is_positive}).

Finally we can define the global pressure bilinear form that mimics the weak pressure equation \eqref{weak_pressure_equation}. We denote by $\dbf_{h,\bLambda} : \dof^{2k}_h \times \dof^{2k}_h \to \R$, the global bilinear form such that
\begin{equation} \label{eqn:global_pressure_diffusion_bilinear_form}
\dbf_{h,\bLambda}(\du_h,\tf_h) = \sum_{T\in\scriptT_h} \dbf_{T,\bLambda}(\du_{T}, \tf_{T}).
\end{equation}
We also write the linear functional $l^{p,\hs}_h : \dof^{2k}_h \to \R$ such that
\begin{equation} \label{eqn:global_diffusion_linear_functional}
l^{p,\hs}_h(\tf_h) = \lip{\Omega}{(q^+(t^\hs,\cdot)-q^-(t^\hs,\cdot))}{\pf_h},
\end{equation}
which mimics the right-hand side of the weak pressure equation. 

Recalling that the pressure equation has diffusion tensor $\bLambda = \bkappa(c)$ given by \eqref{eqn:permeability-viscosity-ratio}, the discrete pressure $\dpdof_h^\hs \in \dof^{2k}_{h,*}$ on $[t^n,t^{n+1})$ is computed by solving
\begin{equation} \label{eqn:discrete_pressure_equation}
\dbf_{h,\bkappa^{\hs}}(\dpdof_h^\hs, \tf_h) = l^{p,\hs}_h(\tf_h)\qquad \forall \tf_h\in \dof^{2k}_{h,*}
\end{equation}
where $\bkappa^{\hs}=\bkappa(\tilde{\pc}^\hs)$ with
$\tilde{\pc}^{\hs}$ the following extrapolation of the concentration at time $t^\hs$:
\begin{equation}\label{eqn:concentration_extrapolation}
\tilde{\pc}^{\hs} = \frac{3}{2} \pc^n - \frac{1}{2} \pc^{n-1}
\end{equation}
(we take $\pc_T^{-1}=\pc_T^0=\pi^k_T c_0$ for all $T\in\scriptT_h$). Note that this choice of extrapolation decouples the
pressure equation from the concentration equation (see Algorithm \ref{algo1} below).
Of course, \eqref{eqn:discrete_pressure_equation} only defines the pressure up to an additive
constant, so we normalise by imposing
\begin{equation}\label{eqn:discrete_pressure_equation:norm}
\int_\Omega \pp_h^\hs=0.
\end{equation}

\subsection{The discrete fluxes}\label{sec:discrete.fluxes}

As can be readily seen from the concentration equation \eqref{peaceman_concentration}, the distribution of the concentration depends heavily on the Darcy velocity of the invading fluid as defined by the pressure equation \eqref{eqn:peaceman_pressure}. In the HHO framework, the discretisation of advective
terms is done by using both the velocity in each cell and its fluxes through the faces \cite{droniou2015discontinuous}.
This velocity and fluxes must be properly chosen to ensure the numerical stability of the
discretised concentration equation -- the fact that $\divergence (\BU) = q^+ - q^-$ at the continuous
level is what ensures stability estimates on the continuous concentration, and this must be
mimicked at the discrete level.

For simplicity of notation, we let here $\dpdof_h=\dpdof_h^\hs$ and $\bkappa=\bkappa^\hs$.
Define $\bkappa_{TF}$ from $\bkappa$ by \eqref{def:LambdaF}.
Following \cite[Section 3.1]{cockburn2016bridging},
the numerical flux $\fv_{TF}$, outward from cell $T$ through the face $F$, is given by:
\begin{equation} \label{eqn:discrete_fluxes}
\fv_{TF} = -\bkappa \nabla \ro^{2k+1}_{T,\bkappa} \dpdof_{T} \cdot \nv_{TF} + \frac{{\bkappa_{TF}}}{h_F}\br^{2k,\dagger}_{\partial T}\left(\pi^{2k}_F (\RO^{2k+1}_{T,\bkappa} \dpdof_{T} - \pp_F)\right),
\end{equation}
where, setting $\dof^{2k}_{\partial T}=\bigtimes_{F\in\scriptF_T} \Poly^{2k}(F)$
(identified with a subspace of $L^2(\partial T)$),
$\br^{2k,\dagger}_{\partial T}$ is the adjoint, for the $L^2(\partial T)$ inner product,
of the operator $\br^{2k}_{\partial T} : \dof^{2k}_{\partial T} \to \dof^{2k}_{\partial T}$ defined by:
for all $\tf_{\partial T}=(\pf_F)_{F\in\scriptF_T}\in \dof^{2k}_{\partial T}$,
\begin{equation}
\br^{2k}_{\partial T}(\tf_{\partial T}) =  \left(\pi^{2k}_{F}[\pf_F - \ro^{2k+1}_{T,\bkappa} (0,\tf_{\partial T}) + \pi^{2k}_T \ro^{2k+1}_{T,\bkappa} (0,\tf_{\partial T})]\right)_{F\in\scriptF_T}.
\end{equation}

The first term of \eqref{eqn:discrete_fluxes} can be seen as the ``naive'' discrete flux that we would obtain if we simply substituted the discrete pressure into the definition of the continuous fluxes. The second term can thus be thought of as a discrete conservative correction to the ``naive'' flux.

The discrete Darcy velocity in cell $T$ is then given by
\begin{equation} \label{eqn:darcy_velocity_volumetric_reconstruction}
\RU_T = -\bkappa \nabla r^{2k+1}_{T,\bkappa} \dpdof_{T}.
\end{equation}

Finally, it is important to note that the discrete fluxes and Darcy velocity satisfy the following local conservation condition from \cite{cockburn2016bridging}.

\begin{theorem} \label{thm:conservation_of_the_discrete_fluxes}
	Let $T \in \scriptT_h$ and let $\dpdof_h$ be the solution to \eqref{eqn:discrete_pressure_equation}. Then for any $\tf_{T} \in \dof^{2k}_T$, the discrete Darcy velocity and fluxes \eqref{eqn:discrete_fluxes} satisfy
	\begin{equation} \label{eqn:conservation_of_the_discrete_fluxes}
	\dbf_{T,\bkappa}(\dpdof_{T}, \tf_{T}) = -\lipd{T}{\RU_T}{\nabla \pf_T} + \sum_{F\in\scriptF_T}\lip{F}{\fv_{TF}}{(\pf_T - \pf_F)}.
	\end{equation}
\end{theorem}

\subsection{Concentration equation}\label{sec:concentration.equation}

The discrete concentration equation is formulated in terms of a stationary advection--diffusion--reaction equation as in \cite{droniou2015discontinuous},
with source and reaction terms incorporating the time-stepping. We recall that, for the concentration
equation, the relevant discrete space is $\dof^k_h$.

We consider $\RU=((\RU_T)_{T\in\scriptT_h},(U_{TF})_{T\in\scriptT_h, F\in\scriptF_T})$ the Darcy velocity and fluxes
reconstructed above at time $t^\hs$.
The following discrete advective derivative $\ad_{T,\RU}^k \dv_{T}$ is designed to be a high-order approximation to the continuous quantity $\BU \cdot \nabla v$ on the cell $T$.

\begin{definition} \label{def:advective_derivative}
	For all $T \in \scriptT_h$, the discrete advective derivative $\ad_{T,\RU}^k : \dof_T^k \to \Poly^k(T)$ is such
that, for any $\dv_{T} \in \dof^k_T$ and any test function $w \in \Poly^k(T)$,
	\begin{equation}
	\lip{T}{(\ad_{T,\RU}^k \dv_{T})}{w} = \lip{T}{(\RU_T \cdot \nabla \pv_T)}{w} + \sum_{F \in \scriptF_T} \lip{F}{\fv_{TF} (\pv_F - \pv_T)}{w}.
	\end{equation}
\end{definition}

Recalling Definition \eqref{eqn:diffusion_dispersion_tensor}, we define the discrete reconstructed diffusion tensor $\dD$ for each cell $T \in \scriptT_h$, for any $x \in T$ by
\begin{equation} \label{eqn:diffusion_tensor_reconstruction}
\dD(x) = \Phi(x)(d_m \eye + |\RU_T(x)|(d_l E(\RU_T(x)) + d_t (\eye - E(\RU_T(x)))).
\end{equation}

The local advection--reaction bilinear form  $\dbf_{T,R, \RU} : \dof^k_h \times \dof^k_h \to \R$ is then defined as
\begin{equation} \label{eqn:local_advection_bilinear_form}
\dbf_{T,R, \RU}(\du_{T}, \tf_{T}) = -\lip{T}{\pu_T}{(\ad_{T,\RU}^k \tf_{T})} + \lip{T}{R\pu_T}{\pf_T}  + \stab_{T,\RU}(\du_{T}, \tf_{T}),
\end{equation}
where $R$ are the reaction terms, encompassing the time-stepping, given by
\begin{equation}\label{eqn:concentration_reaction_terms}
R = \frac{2 \Phi}{\Delta t} + q^-(t^\hs,\cdot).
\end{equation}
Furthermore, $\stab_{T,\RU} : \dof^k_T \times \dof^k_T \to \R$ is the advective stabilisation term given by
\begin{equation} \label{eqn:advective_stabilisation_form}
\stab_{T,\RU}(\du_{T}, \tf_{T}) = \sum_{F \in \scriptF_T} \lip{F}{[\fv_{TF}]^-(\pu_F - \pu_T)}{(\pf_F - \pf_T)}
\end{equation}
where $[\fv_{TF}]^-=\max(0,-\fv_{TF})$.
The global advection-reaction bilinear form $\dbf_{h, R, \RU} : \dof^k_h \times \dof^k_h \to \R$ is defined such that
\begin{equation} \label{eqn:global_advection_bilinear_form}
\dbf_{h,R, \RU}(\du_h, \tf_h) =  \sum_{T \in \scriptT_h} \dbf_{T,R, \RU}(\du_{T}, \tf_{T}).
\end{equation}
Combining the diffusion and advection--reaction bilinear forms, the complete bilinear form for advection--diffusion--reaction is $\dbf_{h,\dD, R,\RU} : \dof^k_h \times \dof^k_h \to \R$ such that
\[
\dbf_{h,\dD, R, \RU}(\du_h, \tf_h) = \dbf_{h,\dD}(\du_h, \tf_h) + \dbf_{h,R, \RU}(\du_h, \tf_h),
\]
and the global linear functional $l^{c,\hs}_h : \dof^k_h \to \R$ is
\begin{equation} \label{eqn:concentration_rhs_linear_form}
l^{c,\hs}_h(\tf_h) = \int_\Omega \left(q^+(t^\hs,\cdot) \hat{c}(t^\hs,\cdot) + \frac{2\Phi}{\Delta t}\pc^n_h\right) \pf_h,
\end{equation}

The discrete concentration scheme then consists in seeking $\dc^\hs_h\in\dof^k_h$ such that
\begin{equation} \label{eqn:discrete_advection_diffusion_equation}
\dbf_{h,\dD, R,\RU}(\dc^\hs_h, \tf_h) = l^{c,\hs}_h(\tf_h),\qquad
\forall \tf_h \in \dof^k_h,
\end{equation}
then extrapolating to obtain $\dc^{n+1}_h$. We recall that $\dD$ and $\RU$ are computed
from $\dpdof^\hs_h$, and therefore depend on $n$.

\medskip

The full scheme is summarised in Algorithm \ref{algo1}.

\begin{algorithm}[!h]
\caption{Complete scheme for the pressure--concentration \label{algo1}}
\begin{algorithmic}[1]
\State \emph{Set $\pc_T^{-1}=\pc_T^0=\pi^k_T c_0$ for all $T\in\scriptT_h$}
\For{$n=0$ to $N-1$}
	\State Compute $\dpdof^\hs_h$ by \eqref{eqn:discrete_pressure_equation}--\eqref{eqn:discrete_pressure_equation:norm},
from $\pc_h^{n-1}$ and $\pc_h^n$.
	\State Compute $\RU$ by \eqref{eqn:discrete_fluxes} and \eqref{eqn:darcy_velocity_volumetric_reconstruction} with $\dpdof_h=\dpdof_h^\hs$.
	\State Compute $\dc_h^\hs$ by \eqref{eqn:discrete_advection_diffusion_equation} and set $\dc_h^{n+1}=2\dc_h^\hs-\dc_h^n$.
\EndFor
\end{algorithmic}
\end{algorithm}

At each iteration, $(\dpdof^\hs_h,\dc_h^\hs)$ are computed by solving two
decoupled linear equations, one corresponding to an HHO scheme for a pure diffusion
equation, the other one to an HHO scheme for a diffusion--advection--reaction equation.



Essential questions when designing a numerical scheme are the existence, uniqueness and stability
of its solution; here, stability is understood as the grid size and time steps go to zero. The following
theorem brings an answer to these questions. It states the existence and uniqueness of the
solution to the scheme, and provides a bound on the cell unknowns corresponding to the concentration.
We note that this bound is, of course, uniform with respect to the grid size, but also with respect to
other important parameters, in particular the molecular, longitudinal and transverse dispersion coefficients
$d_m$, $d_l$ and $d_t$. A stability with respect to $d_m$ is all the more essential since this coefficient
tends to be taken equal to $0$ in numerical tests \cite{wang2000approximation,droniou2007convergence}.

\begin{theorem}[Existence, uniqueness and estimates for the discrete solution] \label{thm:stability_of_crank_nicolson}
	Let $\scriptM = (\scriptT_h, \scriptF_h)$ be a mesh of $\Omega$ and take $N \geq 1$. If $(q^+, q^-, \BK, \BD, \Phi, \hat{c}, c_0)$ are data satisfying \eqref{eqn:regularity_assumptions}, then there exists a unique solution $(\dpdof_h, \dc_h)$ to the iterative scheme described in Algorithm \ref{algo1}.
Morever, for all $n=0,\ldots,N$,
	\begin{equation}\label{est:cN}
	\|\pc^{n}_h \|^2_{L^2(\Omega)} \leq {\frac{e^2}{\Phi_*^2}\left( \| \pc^0_h \|^2_{L^2(\Omega)}
+2t_f^2 \|q^+\|^2_{L^\infty(0,t_f;L^2(\Omega))}\right)}.
	\end{equation}
\end{theorem}

The proof of Theorem \ref{thm:stability_of_crank_nicolson} is provided in the appendix.

\section{Tests}\label{sec:tests}

In this section, we illustrate the results obtained from the HHO scheme for the simulation of miscible fluid displacement in an oil reservoir. Some of the following test cases first appeared in \cite{wang2000optimal} for the ELLAM-MFEM method and have been subsequently applied in \cite{droniou2007convergence} for the MFV scheme. In every simulation here, we use the spatial domain $\Omega = (0,1000)^2$ measured in ft$^2$ and consider the time period $[0, 3600]$ (approximately $10$ years) measured in days. The injection and production source terms are Dirac masses, approximated as usual by a piecewise constant function on the relevant mesh cell. The injection well is located at $(1000,1000)$ with an injection rate of $30$ ft$^2$/day. The production well is correspondingly located at $(0,0)$ with a production rate of $30$ ft$^2$/day. We always take the injected concentration $\hat{c} = 1$ with an initial condition given by $c_0(x) = 0$. The viscosity of the oil is given by $\mu(0) = 1.0$ cp and the mobility ratio is $M=41$ (see \eqref{eqn:viscosity}).
We  assume that molecular diffusion is negligible, setting $d_m = 0.0$ \fspd\ contrasting with the dispersion effects $d_l = 50.0$ \fspd\ and $d_t = 5.0$ \fspd.
The porosity of the medium is taken to be a constant $\Phi(x) = 0.1$. For each test we present surface and contour plots of the concentration $c$, the principle quantity of interest.

All tests were ran on a laptop with processor Intel i7-4710MQ 2.5Ghz, 6MB Cache
and 16GB of RAM at 1600Mhz.

\begin{remark}
	Taking $d_m > 0$ is required to meet the regularity assumptions \eqref{eqn:regularity_assumptions} but physically unrealistic due to the fact that the magnitude of the molecular diffusion present in miscible fluid flow is negligibly small compared to the dispersive effects. We will present  results taking $d_m = 0$ to demonstrate the suitability of the scheme to real-world parameters. We recall that the stability result of Theorem \ref{thm:stability_of_crank_nicolson} is
independent of $d_m$, and is therefore uniformly valid up to the limit $d_m\to 0$.
\end{remark}

Tables \ref{tab:mesh_parameters} and \ref{tab:mesh_parameters2} give the numbers of edges and  sizes (maximum ratio of area to perimeter across all cells) for each mesh used in the following tests. The four mesh varieties are shown in Figure \ref{fig:meshes}. The triangular, Cartesian and Kershaw meshes were first introduced in the FVCA5 benchmark \cite{herbin2008benchmark} as mesh families 1, 2 and 4.1 respectively. The hexagonal-dominant mesh was used in \cite{BDM10,di2015extension}.

\begin{figure}[h]
	\centering
	\begin{subfigure}{.24\textwidth}
		\centering
		\includegraphics[scale=0.15]{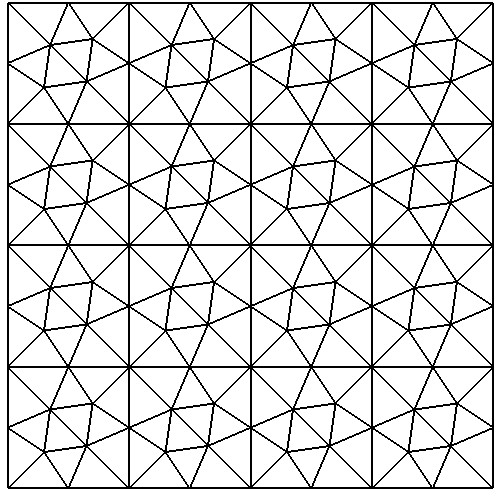}
		\caption{Triangular}
	\end{subfigure}%
	\begin{subfigure}{.24\textwidth}
		\centering
		\includegraphics[scale=0.15]{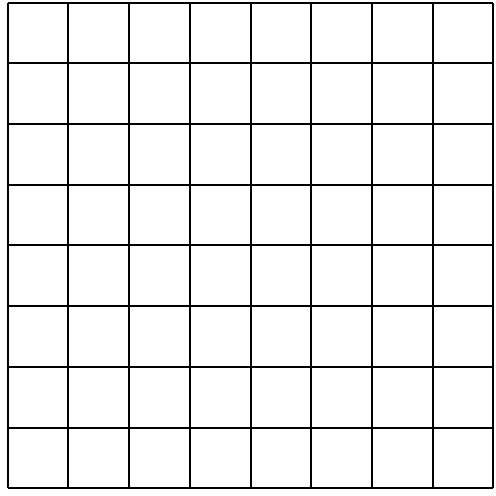}
		\caption{Cartesian}
	\end{subfigure}
	\begin{subfigure}{.24\textwidth}
		\centering
		\includegraphics[scale=0.155]{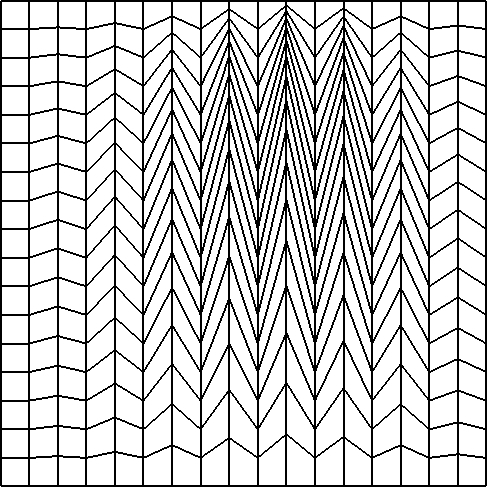}
		\caption{Kershaw} \label{fig:kershaw_mesh}
	\end{subfigure}%
	\begin{subfigure}{.24\textwidth}
		\centering
		\hspace{5px}\input{img/meshes/pi6_tiltedhexagonal_2.tex}
		\caption{Hexagonal} \label{fig:tilted_hexagonal_mesh}
	\end{subfigure}
	\caption{The four varieties of meshes. Figures obtained from \cite{di2014arbitrary}.}
	\label{fig:meshes}
\end{figure}

\begin{center}
\begin{table}[!h]
\begin{tabular}{| c | c | c | c | c |}
\hline
\multicolumn{5}{|c|}{Number of edges}\\\hline 
 & Triangular & Cartesian & Kershaw & Hexagonal  \\\hline
 Mesh 1 & 92 & 40 & 612 & 62 \\\hline
 Mesh 2 & 352 & 144 & 2380 & 220 \\\hline
 Mesh 3 & 1376 & 544 & 5304 & 824 \\\hline
 Mesh 4 & 5540 & 2112 & 9384 & 3184 \\\hline
 Mesh 5 & 21632 & 8320 & 14620 & 12512 \\\hline
\end{tabular}\\
\begin{tabular}{| c | c | c | c | c |}
\hline
\multicolumn{5}{|c|}{Mesh size}\\\hline 
 & Triangular & Cartesian & Kershaw & Hexagonal  \\\hline
 Mesh 1 & 31.8 & 62.5 & 16.2 & 70.6 \\\hline
 Mesh 2 & 15.9 & 31.2 & 8.96 & 36.7 \\\hline
 Mesh 3 & 7.95 & 15.6 & 6.12 & 18.5 \\\hline
 Mesh 4 & 3.98 & 7.81 & 4.64 & 9.26 \\\hline
 Mesh 5 & 1.99 & 3.91 & 3.73 & 4.63 \\\hline
\end{tabular}
\caption{Mesh parameters for the meshes used in homogeneous permeability tests.}
\label{tab:mesh_parameters}
\end{table}
\begin{table}
\begin{tabular}{| c | c | c |}
	\hline
	\multicolumn{3}{|c|}{Number of edges}\\\hline 
	& Triangular & Cartesian \\\hline
	Mesh 1 & 545 & 144  \\\hline
	Mesh 2 & 2140 & 840  \\\hline
	Mesh 3 & 4785 & 1860  \\\hline
	Mesh 4 & 8480 & 3280  \\\hline
	Mesh 5 & 13225 & 5100  \\\hline
\end{tabular}\\
\begin{tabular}{| c | c | c |}
	\hline
	\multicolumn{3}{|c|}{Mesh size}\\\hline 
	& Triangular & Cartesian  \\\hline
	Mesh 1 & 12.72 & 31.25  \\\hline
	Mesh 2 & 6.36 & 12.5  \\\hline
	Mesh 3 & 4.24 & 8.33  \\\hline
	Mesh 4 & 3.18 & 6.25  \\\hline
	Mesh 5 & 2.54 & 5.00  \\\hline
\end{tabular}
\caption{Mesh parameters for the meshes used in discontinuous permeability tests.}
\label{tab:mesh_parameters2}
\end{table}
\end{center}

\subsection{Numerical results} \label{sec:numres}

In these first four tests, we use a polynomial degree $k = 1$ for the spatial discretisation.
The time-step is $\Delta t = 18$ days ($N \approx 200$ time-steps over $10$ years).

\begin{test} \label{test:peaceman2}\rm
	We take a homogeneous permeability tensor $\BK = 80\BI$ uniformly over the domain. We experiment on a $32\times 32$ Cartesian mesh (Cartesian Mesh 4 in Table \ref{tab:mesh_parameters}.) The mobility within the solvent saturated regions caused by the large adverse mobility ratio $M$ and the lack of molecular diffusion should result in the front of the injected fluid progressing most rapidly along the diagonal between the injection and production wells. These effects are seen on the surface and contour plots in Figure \ref{fig:peaceman_test2} on the Cartesian mesh and imply that the flow is indeed strongest along the diagonal direction as expected. This effect is well studied in the literature and is referred to as the macroscopic fingering phenomenon \cite{ewing1983mathematics}. Notably, our results at $t = 3$ with $k \geq 1$ are far more realistic than those given in \cite{droniou2007convergence}, which suffer from fluid progressing much too rapidly along the boundary of the domain. The scheme in \cite{droniou2007convergence} corresponds to a variant of the HHO method with $k=0$; our own tests with $k=0$ reproduced similar results as in this reference (see, e.g., Figures \ref{fig:compare_ks_t3}\textsc{(a)} and \ref{fig:compare_ks}\textsc{(a)}).
\end{test}

\begin{figure}[h]
	\centering
	\begin{subfigure}{.5\textwidth}
		\centering
		\includegraphics[scale=0.3]{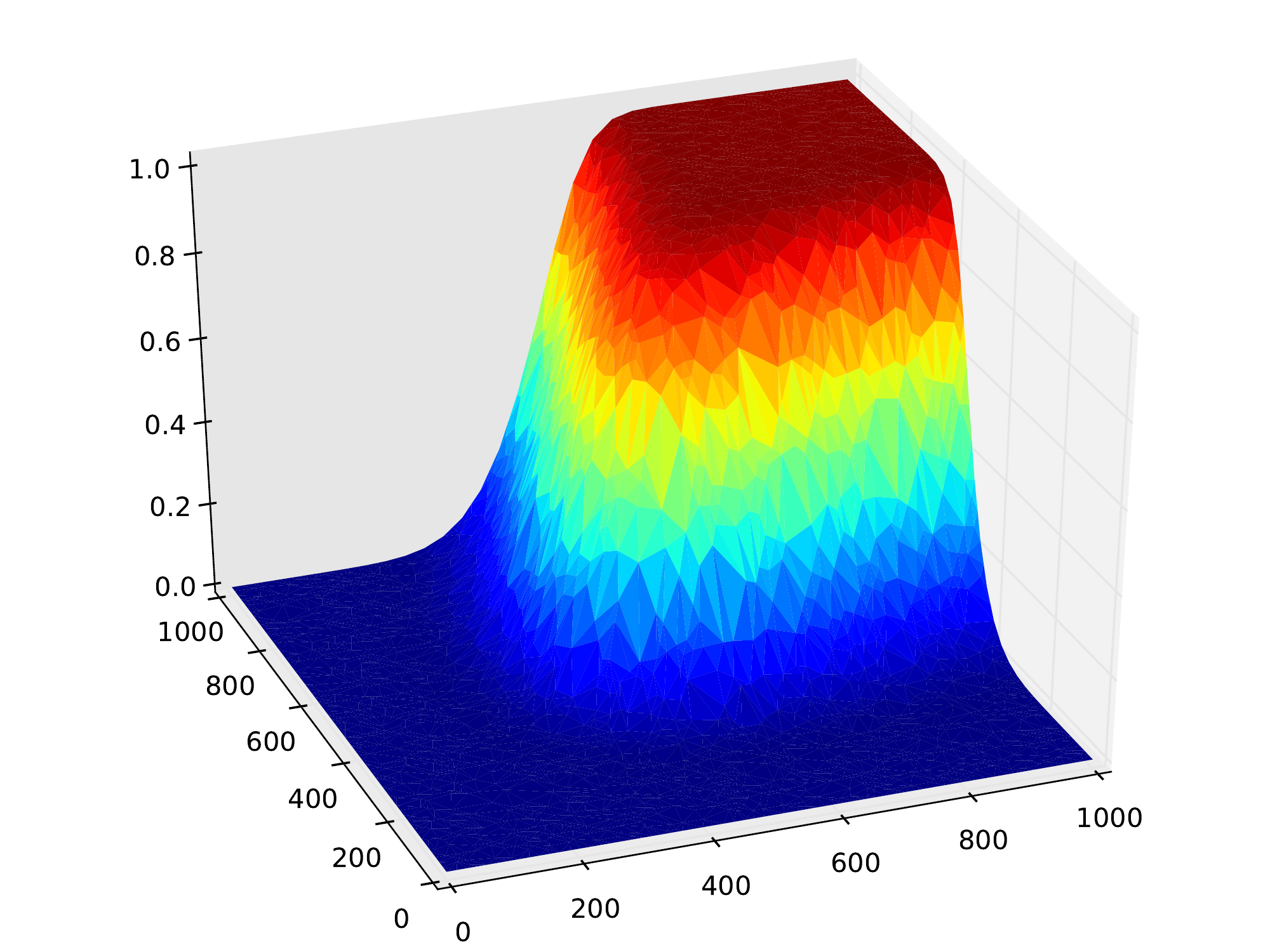}
		\caption{Surface plot at $t = 3$ years}
	\end{subfigure}%
	\begin{subfigure}{.4\textwidth}
		\centering
		\includegraphics[width=5.18cm,height=5cm]{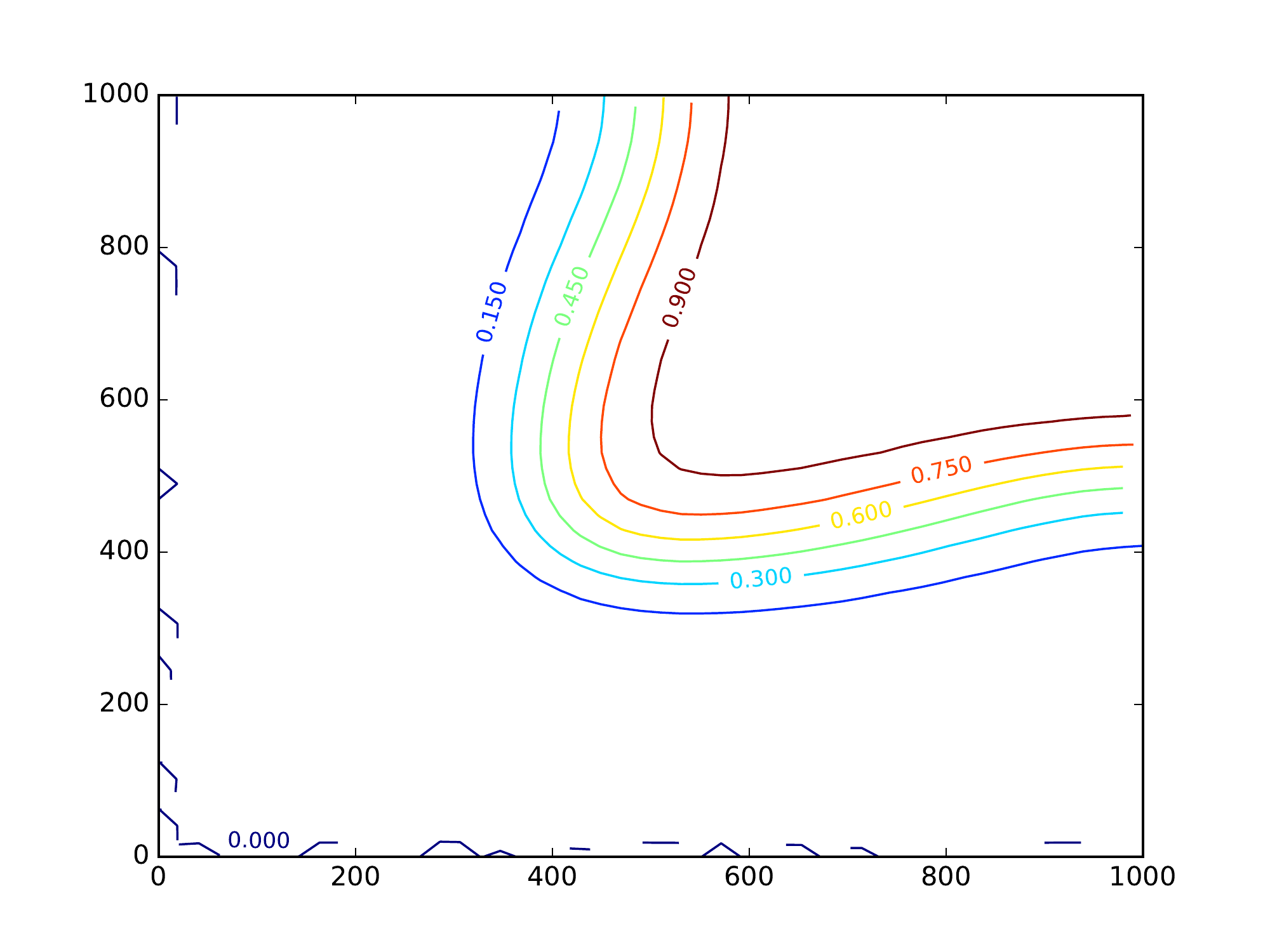}
		\caption{Contour plot at $t = 3$ years}
	\end{subfigure}
	\begin{subfigure}{.5\textwidth}
		\centering
		\includegraphics[scale=0.3]{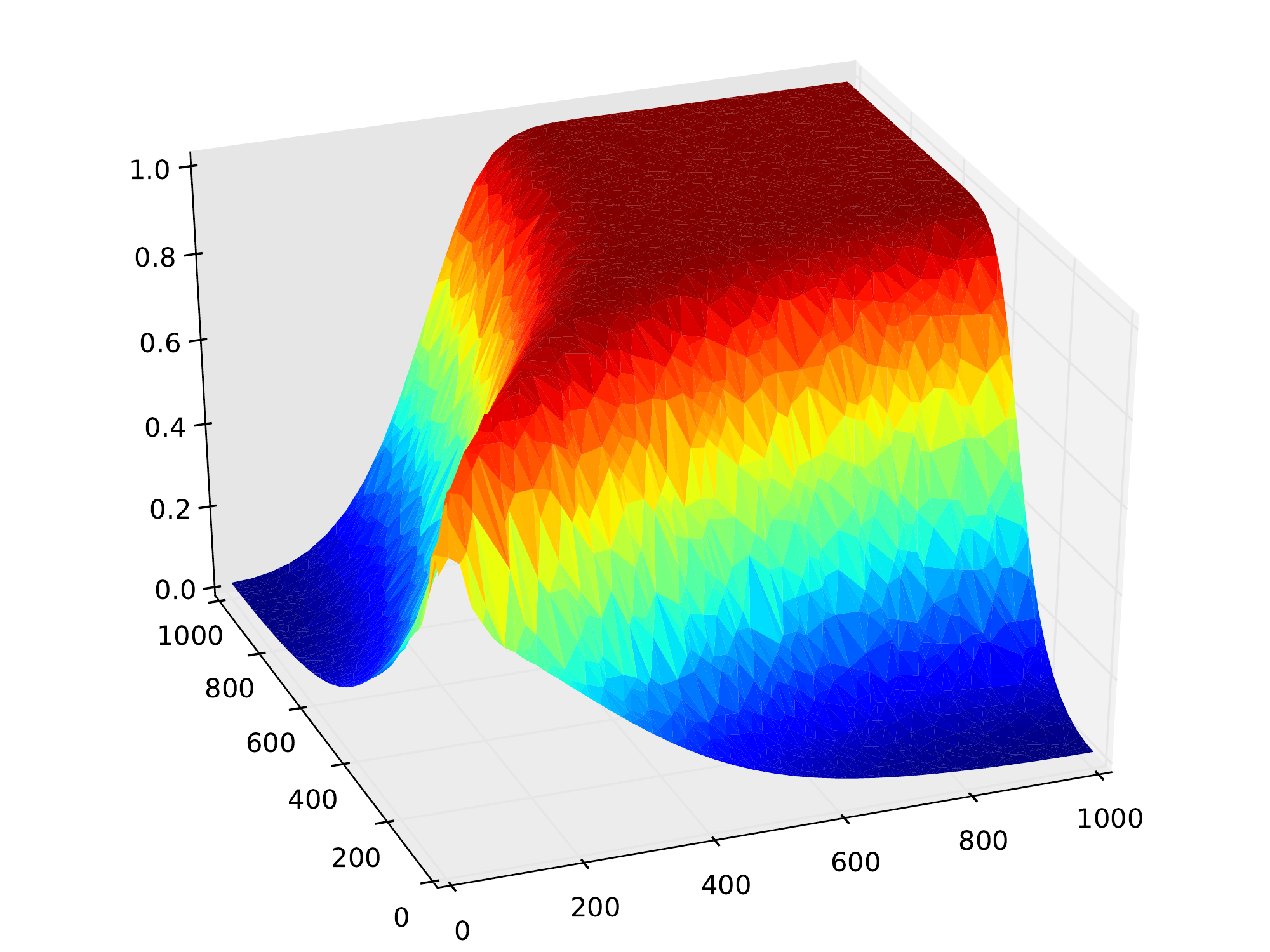}
		\caption{Surface plot at $t = 10$ years}
	\end{subfigure}%
	\begin{subfigure}{.4\textwidth}
		\centering
		\includegraphics[width=5.18cm,height=5cm]{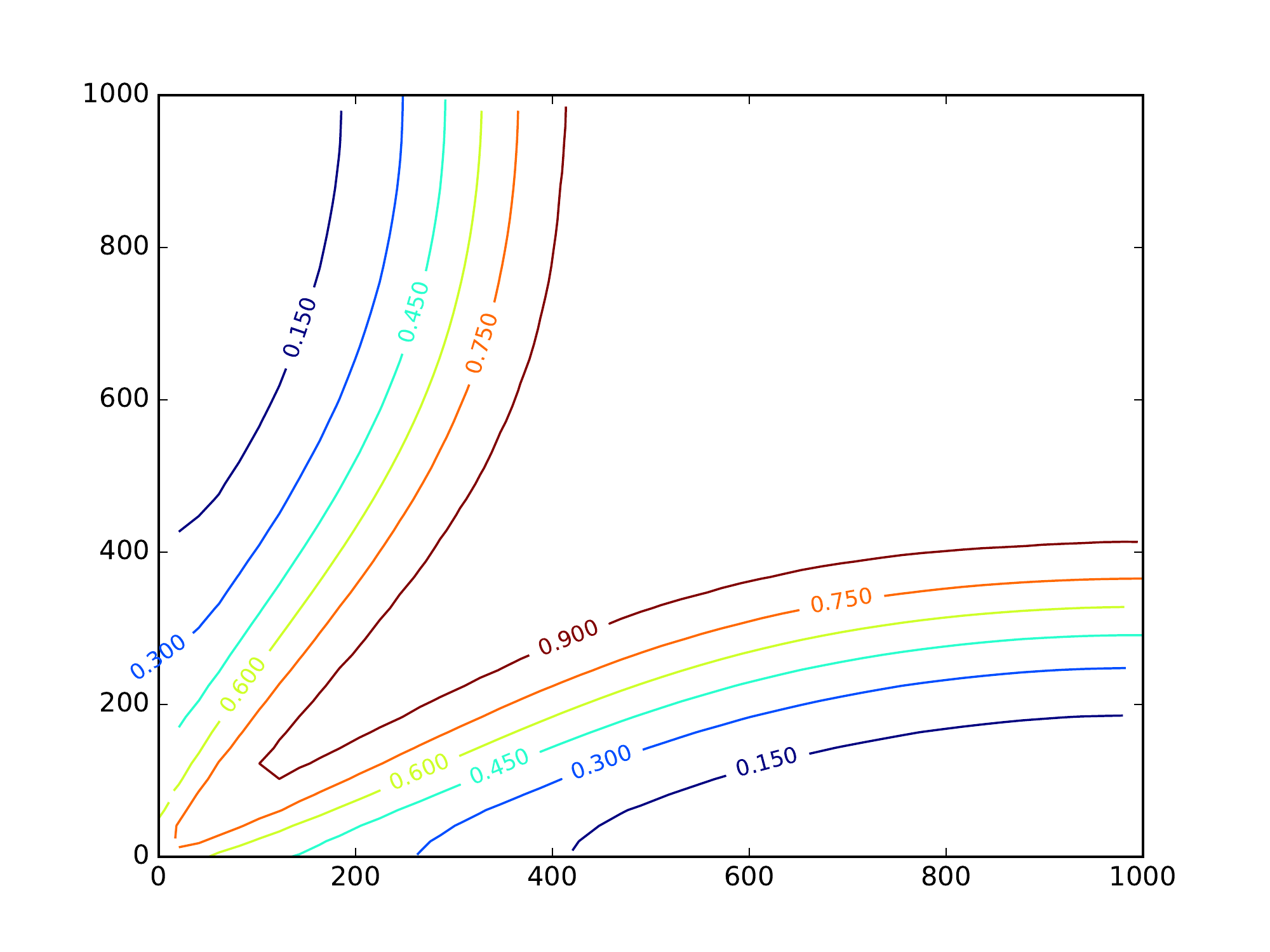}
		\caption{Contour plot at $t = 10$ years}
	\end{subfigure}
	\caption{Concentration of the invading solvent in Test \ref{test:peaceman2} with $k=1$ and $\Delta t = 18$ on Cartesian Mesh 2 with constant permeability.}
	\label{fig:peaceman_test2}
\end{figure}

\begin{test}\rm
	We retain the parameters of Test \ref{test:peaceman2} and use a hexagonal mesh in place of the Cartesian mesh (Hexagonal Mesh 4 in Table \ref{tab:mesh_parameters}). The fingering phenomenon is also observed in the results, which are shown in Figure \ref{fig:peaceman_test2hexa}, where we notice that there is a slight bias along the opposite diagonal at $t = 3$. This bias
is expected, given that the mesh is skewed in this direction (see Figure \ref{fig:meshes}), but it remains
rather small and does not seem to impact much the final result at $t=10$ years.
\end{test}

\begin{figure}[h]
	\centering
	\begin{subfigure}{.5\textwidth}
		\centering
		\includegraphics[scale=0.3]{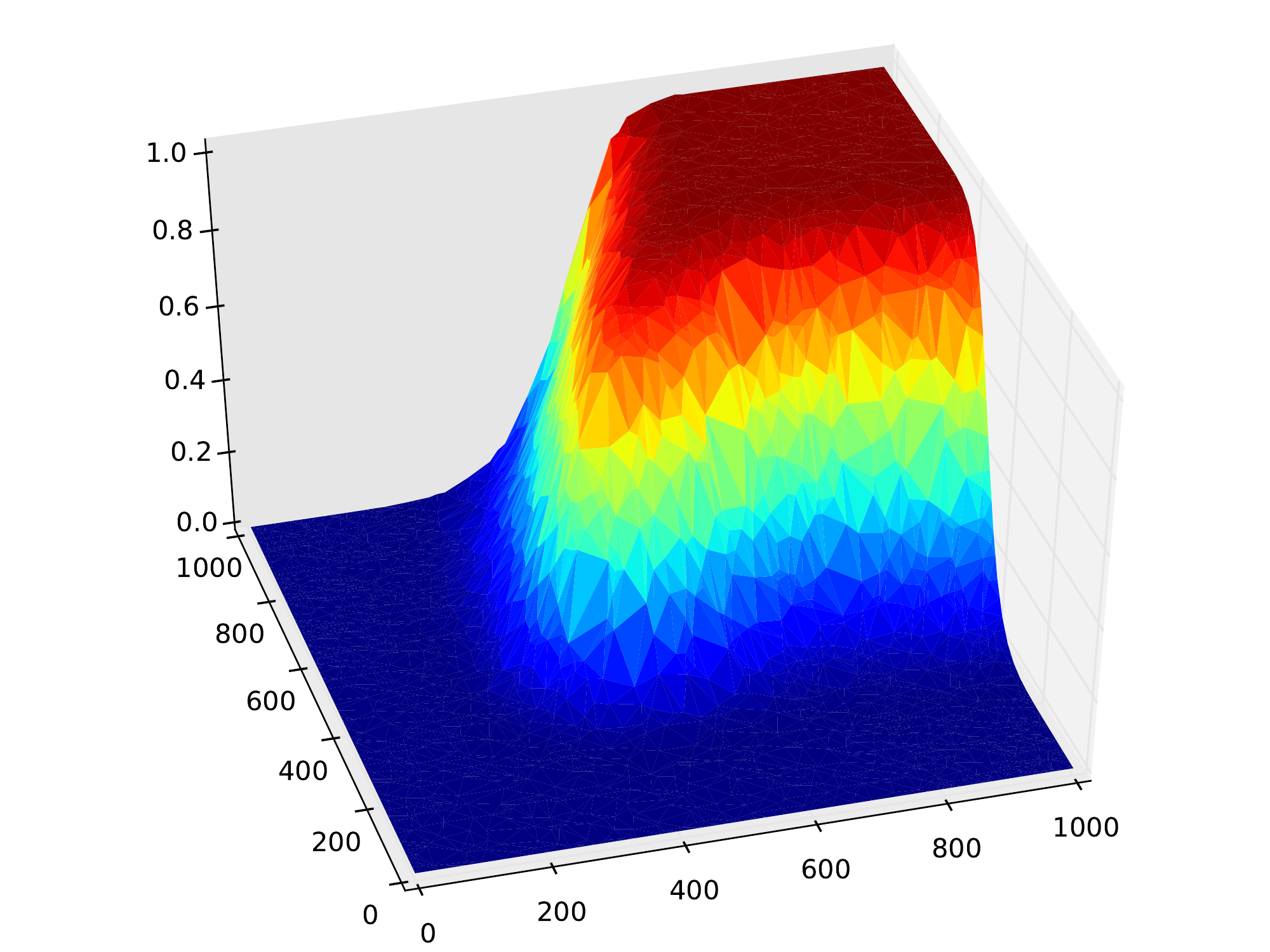}
		\caption{Surface plot at $t = 3$ years}
	\end{subfigure}%
	\begin{subfigure}{.4\textwidth}
		\centering
		\includegraphics[width=5.18cm,height=5cm]{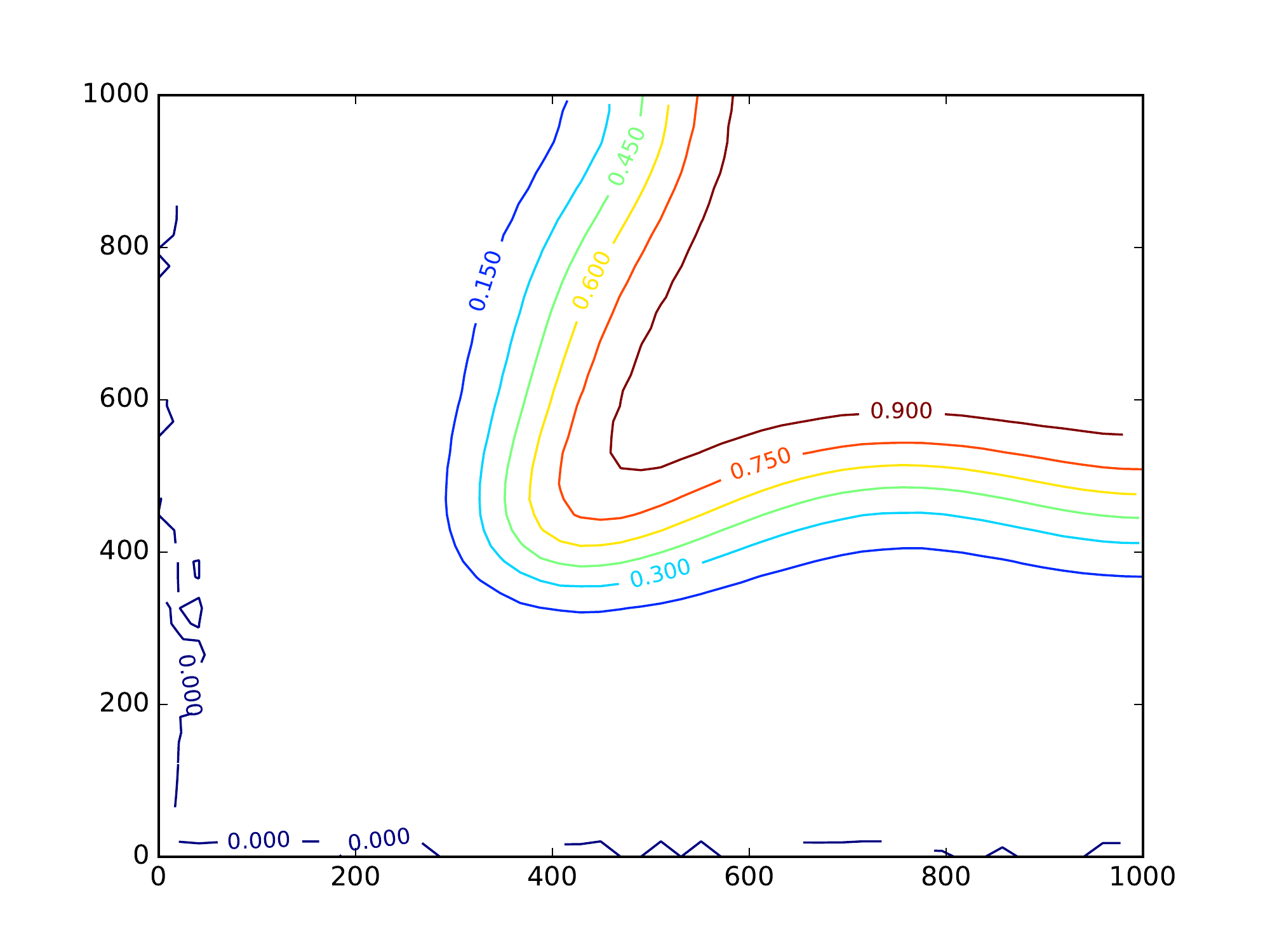}
		\caption{Contour plot at $t = 3$ years}
	\end{subfigure}
	\begin{subfigure}{.5\textwidth}
		\centering
		\includegraphics[scale=0.3]{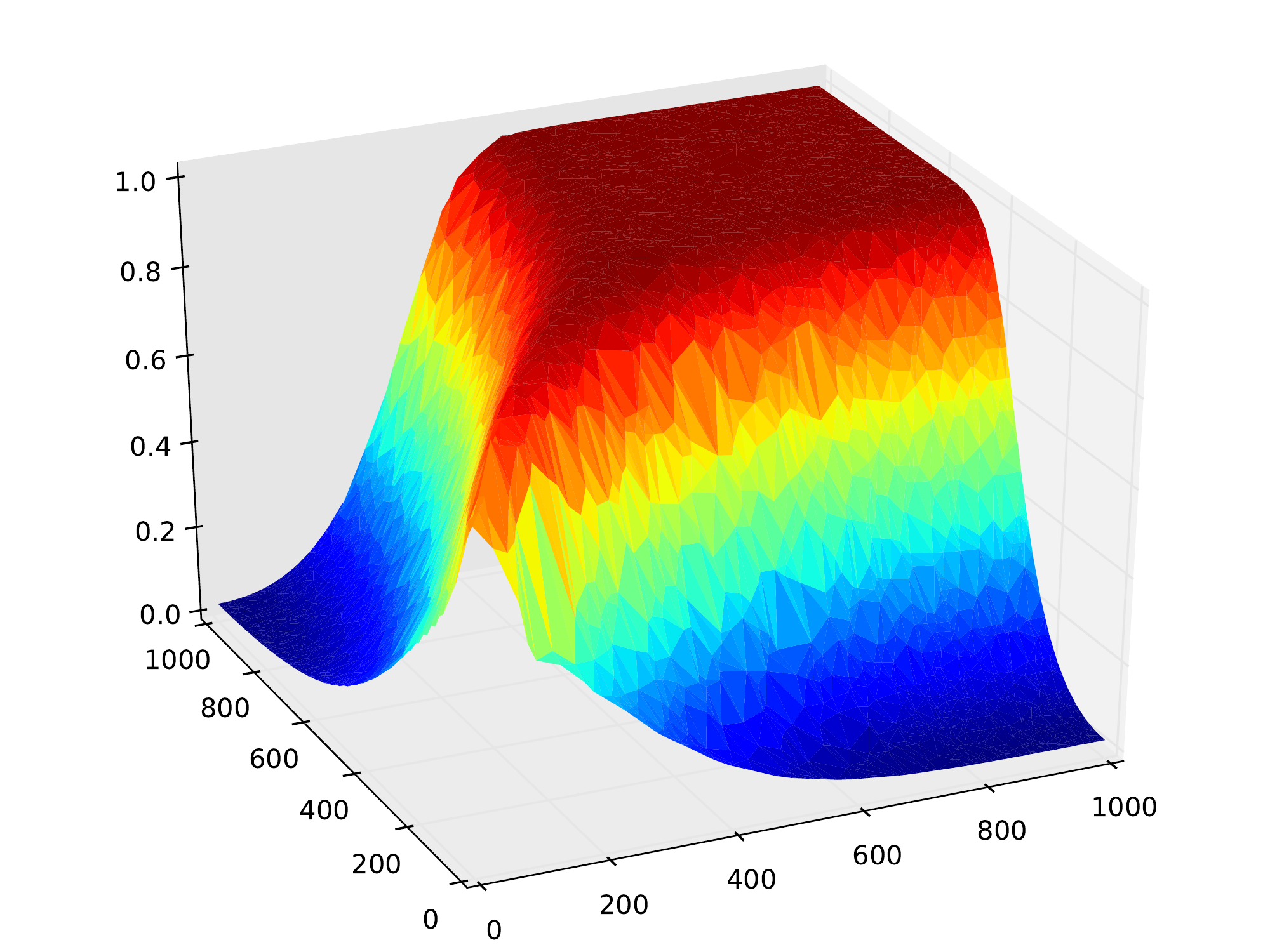}
		\caption{Surface plot at $t = 10$ years}
	\end{subfigure}%
	\begin{subfigure}{.4\textwidth}
		\centering
		\includegraphics[width=5.18cm,height=5cm]{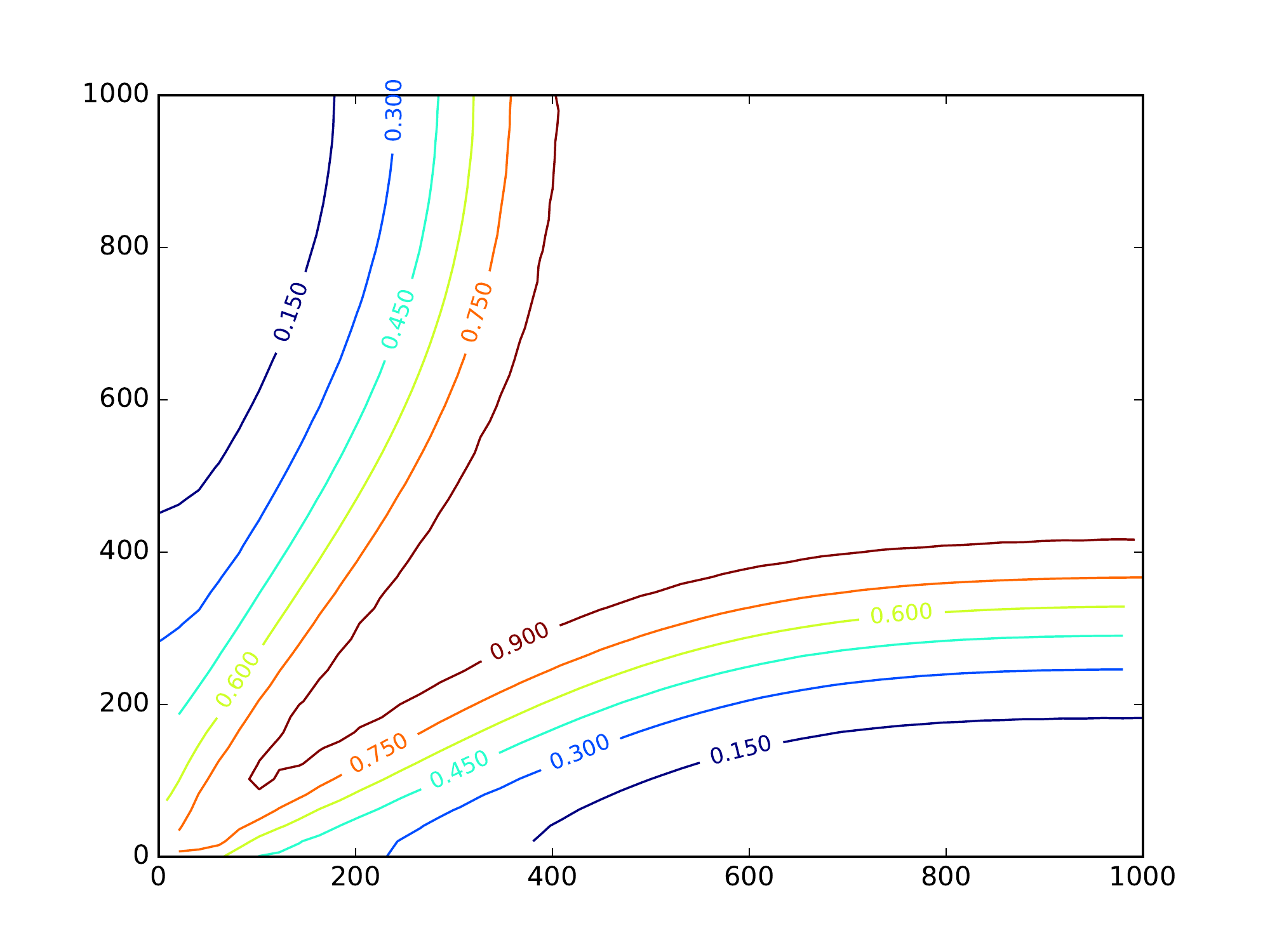}
		\caption{Contour plot at $t = 10$ years}
	\end{subfigure}
	\caption{Concentration of the invading solvent in Test \ref{test:peaceman2} with $k=1$ and $\Delta t = 18$ on Hexagonal Mesh 2 with constant permeability.}
	\label{fig:peaceman_test2hexa}
\end{figure}

\begin{test} \label{test:peaceman4}\rm
	For this test, we take a discontinuous permeability tensor $\BK = 80\BI$ except on the four subdomains $(200,400)\times(200,400)$, $(600,800)\times(200,400)$, $(200,400)\times(600,800)$ and $(600,800)\times(600,800)$ where instead $\BK = 20\BI$ (see Figure \ref{fig:heterogeneous}). We use a $40 \times 40$ Cartesian for this test in order to ensure that the regions of discontinuity are aligned with the edges (Cartesian Mesh 4 in Table \ref{tab:mesh_parameters2}). The results shown in Figure \ref{fig:peaceman_test4} are of great interest to us as they depict very different behaviour to those presented in \cite{droniou2007convergence}. Notably in \cite{droniou2007convergence}, by $t = 10$ years, the invading fluid has yet to subsume the two low permeability regions along the main diagonal. However, the HHO scheme depicts both blocks almost entirely saturated by $t = 10$ years. This implies the presence of a significantly higher amount of dispersion in the solution compared to that produced by the MFV scheme, which suggests that low order schemes may underestimate the amount of diffusion described by the model. Finally, we note that the region saturated with solvent is larger when the permeability is inhomogeneous compared to Test \ref{test:peaceman2} where it was not. This is another common phenomenon that has been well observed \cite{droniou2007convergence,wang2000approximation}.
\end{test}

\begin{figure}[h]
\centering
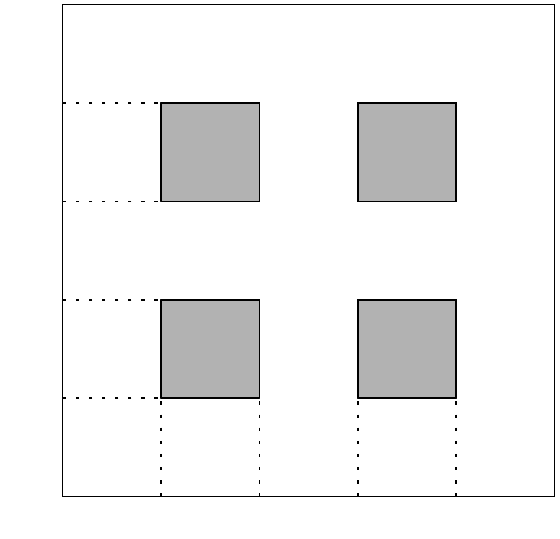
\caption{Permeability tensor for Tests \ref{test:peaceman4} and \ref{test:peaceman5}. \label{fig:heterogeneous}}
\end{figure}

\begin{figure}[h]
	\centering
	\begin{subfigure}{.45\textwidth}
		\centering
		\includegraphics[scale=0.3]{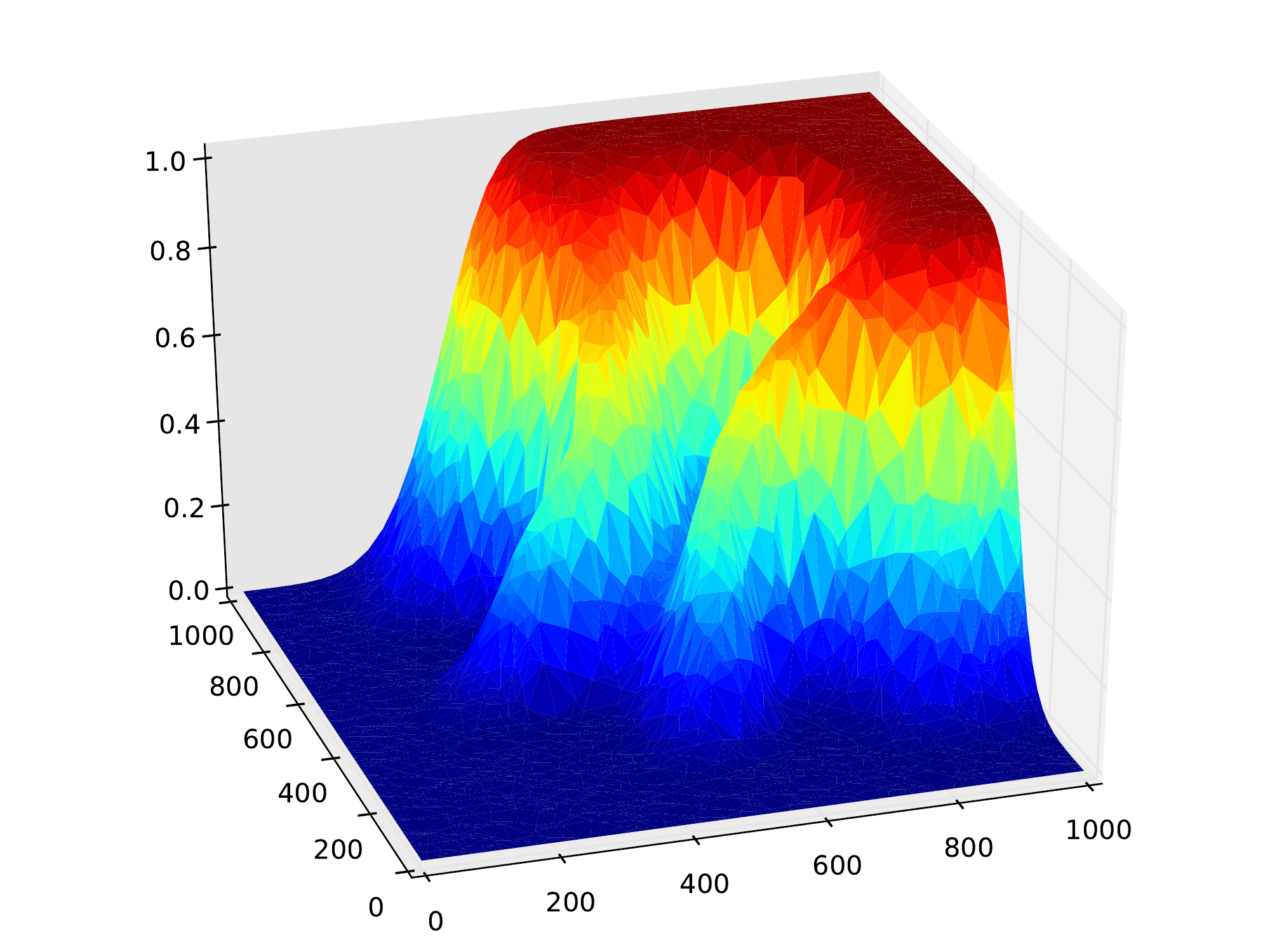}
		\caption{Surface plot at $t = 3$ years}
	\end{subfigure}%
	\begin{subfigure}{.45\textwidth}
		\centering
		\includegraphics[scale=0.3]{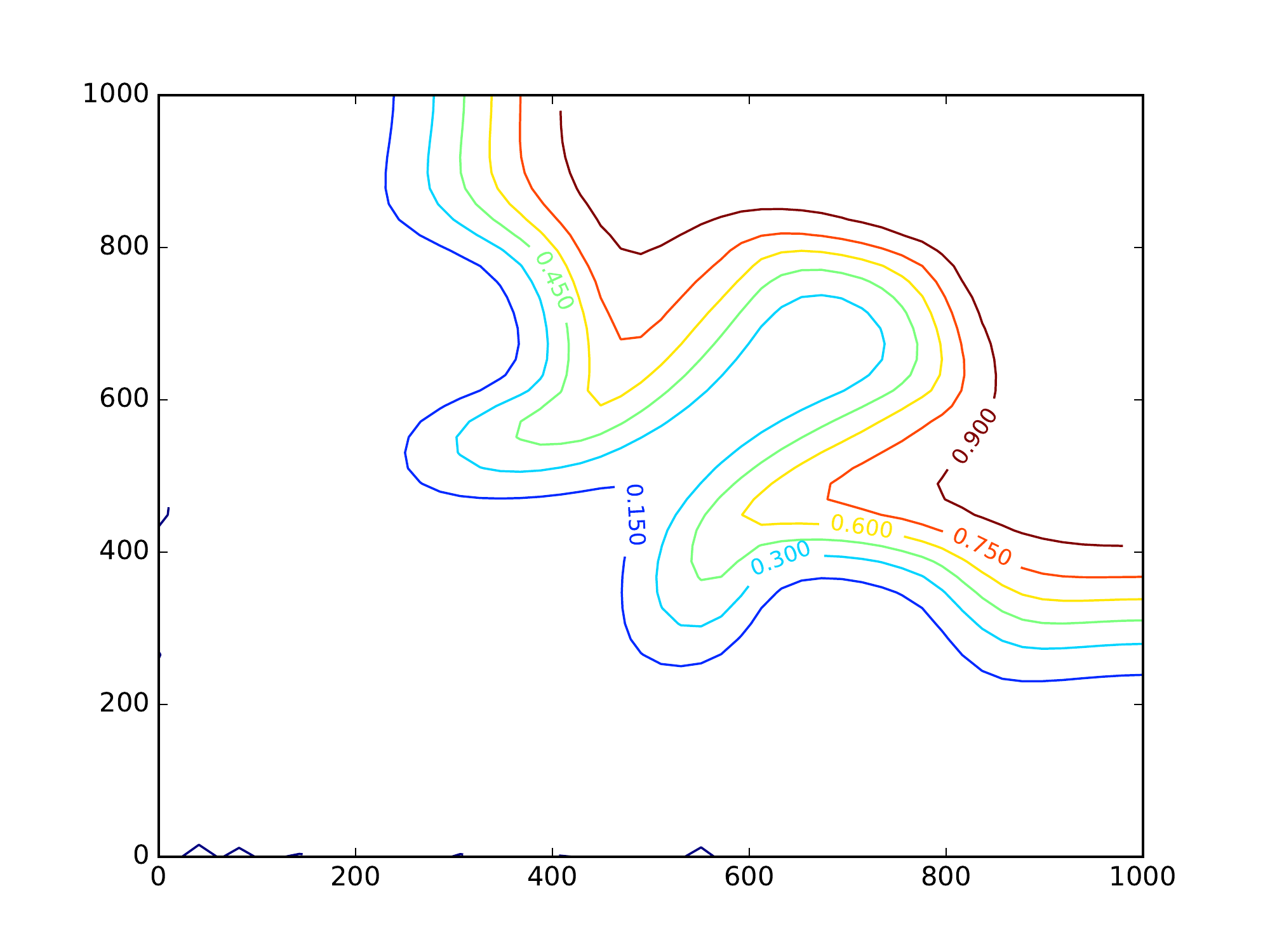}
		\caption{Contour plot at $t = 3$ years}
	\end{subfigure}
	\begin{subfigure}{.45\textwidth}
		\centering
		\includegraphics[scale=0.3]{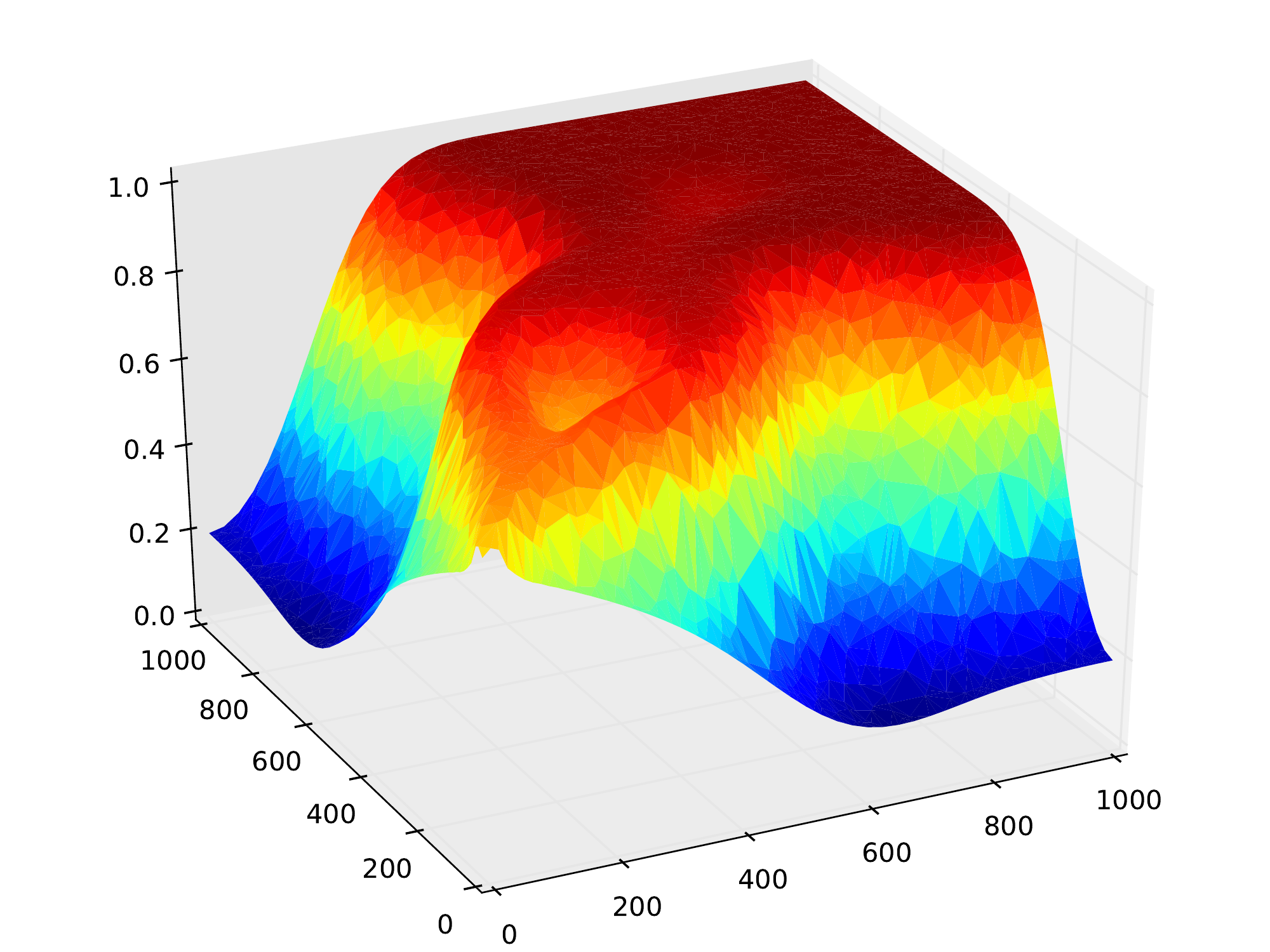}
		\caption{Surface plot at $t = 10$ years}
	\end{subfigure}%
	\begin{subfigure}{.45\textwidth}
		\centering
		\includegraphics[scale=0.3]{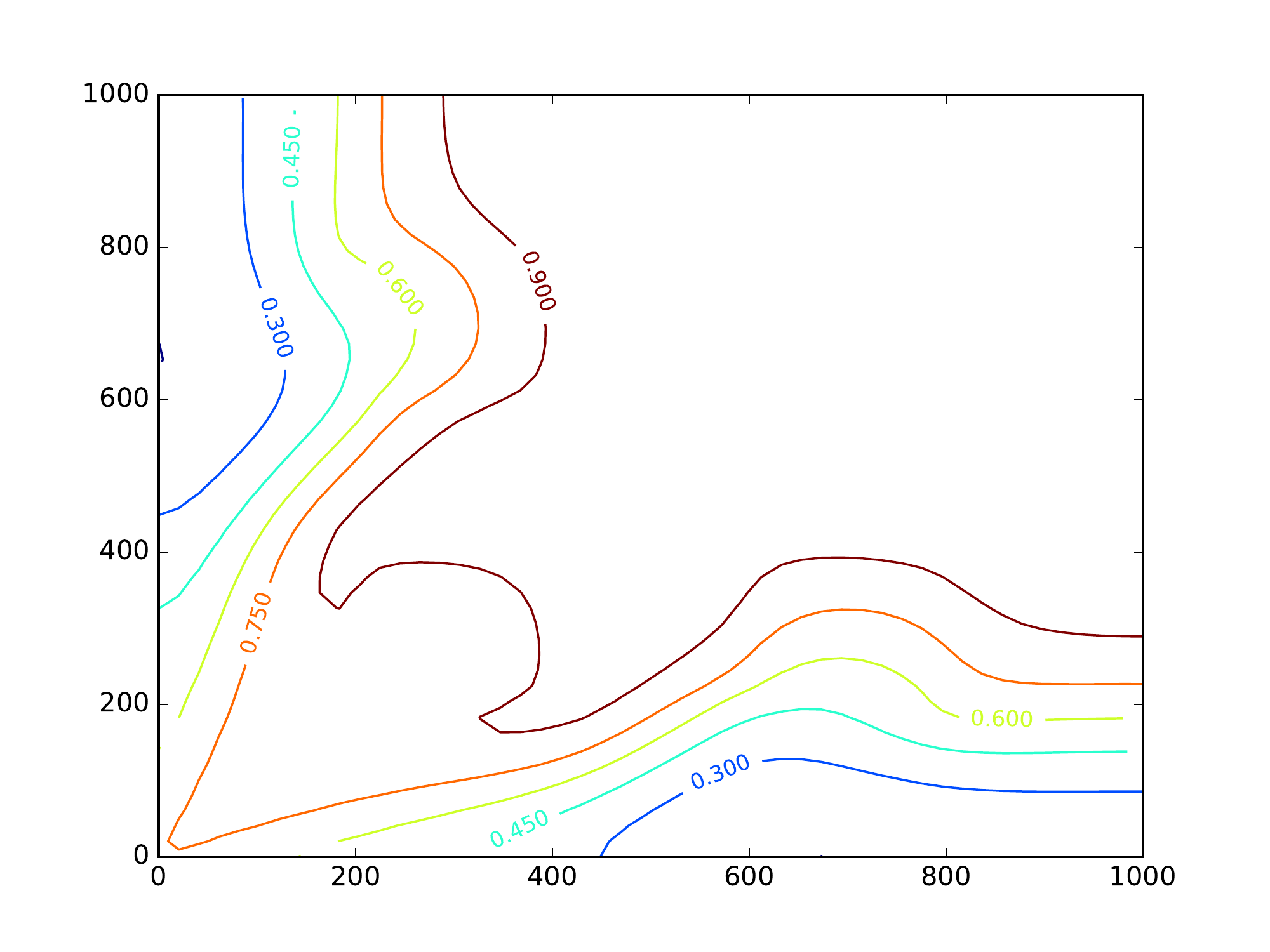}
		\caption{Contour plot at $t = 10$ years}
	\end{subfigure}
	\caption{Concentration of the invading solvent in Test \ref{test:peaceman4} with $k=1$ and $\Delta t = 18$ on a $40 \times 40$ Cartesian mesh with a discontinuous permeability.}
	\label{fig:peaceman_test4}
\end{figure}

\begin{test}\label{test:peaceman5}\rm
	We retain the parameters of Test \ref{test:peaceman4} and use a triangular mesh in place of the Cartesian mesh. The mesh is a $10 \times 10$ grid of the triangular pattern depicted in Figure \ref{fig:meshes} (Triangular Mesh 2 in Table \ref{tab:mesh_parameters2}), which ensures its alignment with the permeability discontinuities. The results can be seen in Figure \ref{fig:peaceman_test4_tri}, where we observe the same general distribution as in the Cartesian case, although the solvent seems slightly more dispersed. For such hybrid methods as the HHO method, 
it is well known that the main unknowns are the edge-based unknowns (see in particular Section \ref{sec:comp.cost});
it is therefore expected that, for a comparable mesh size, a mesh with fewer edges will perform slightly worse
than a mesh with more edges. The Cartesian mesh used in Test \ref{test:peaceman4} has 1600 cells, 3280 edges
and a size of $6.25$; the triangular mesh used here has 1400 cells, 2060 edges and
a size of $6.36$.
\end{test}

\begin{figure}[h]
\centering
	\begin{subfigure}{.45\textwidth}
  \centering
		\includegraphics[scale=0.3]{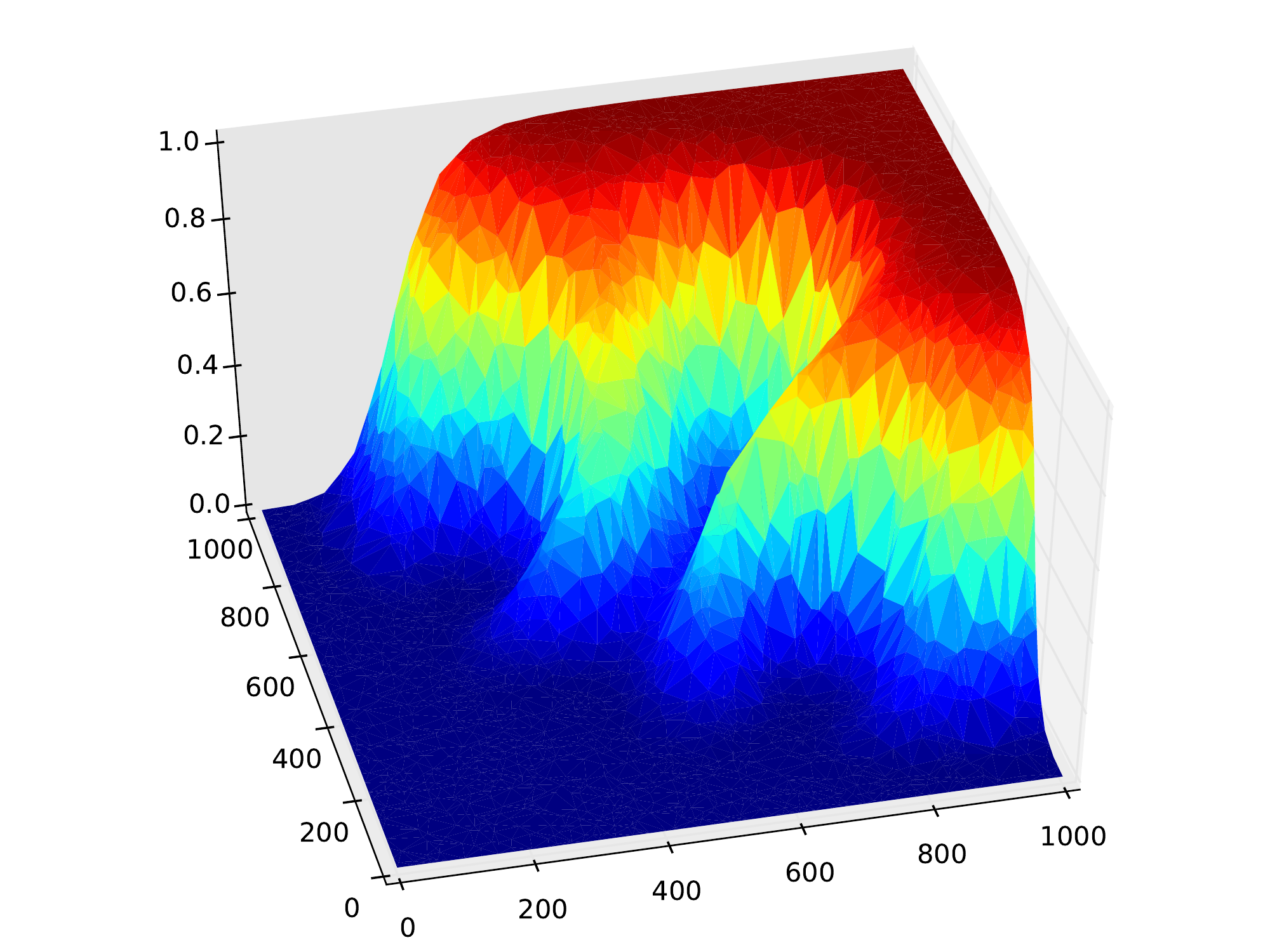}
		\caption{Surface plot at $t = 3$ years}
\end{subfigure}%
	\begin{subfigure}{.45\textwidth}
  \centering
		\includegraphics[scale=0.3]{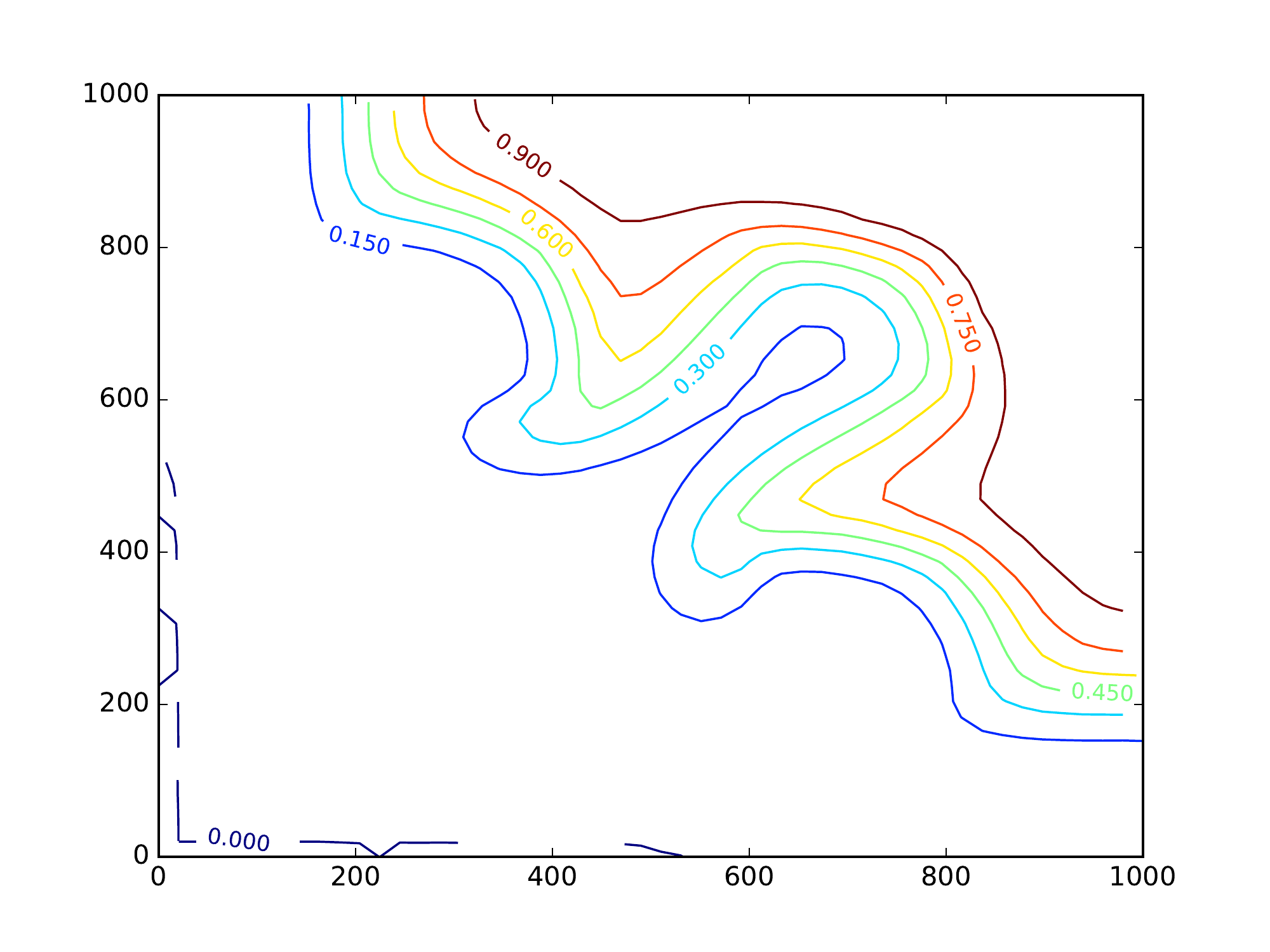}
		\caption{Contour plot at $t = 3$ years}
\end{subfigure}
	\begin{subfigure}{.45\textwidth}
  \centering
		\includegraphics[scale=0.3]{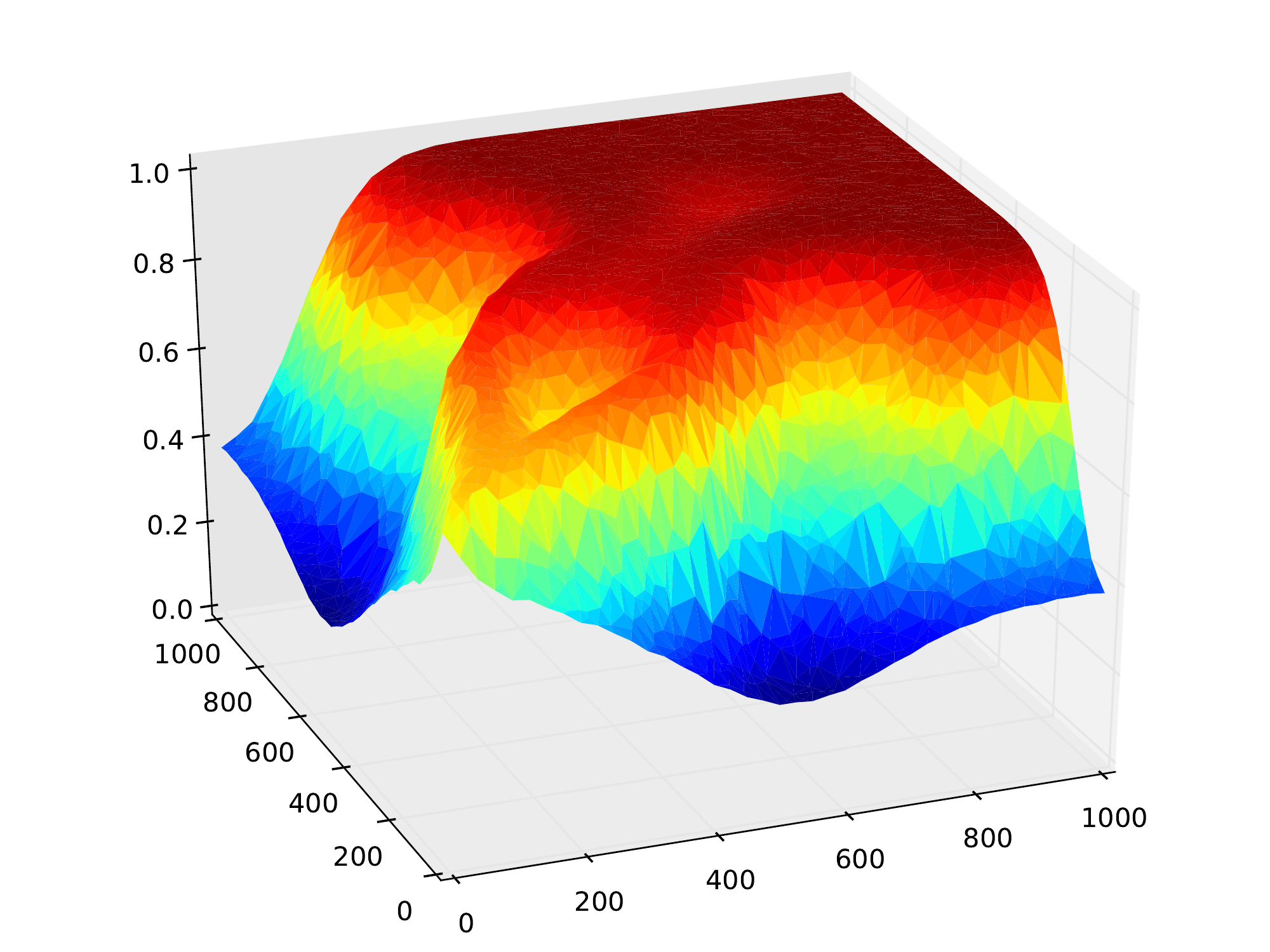}
		\caption{Surface plot at $t = 10$ years}
\end{subfigure}%
	\begin{subfigure}{.45\textwidth}
  \centering
		\includegraphics[scale=0.3]{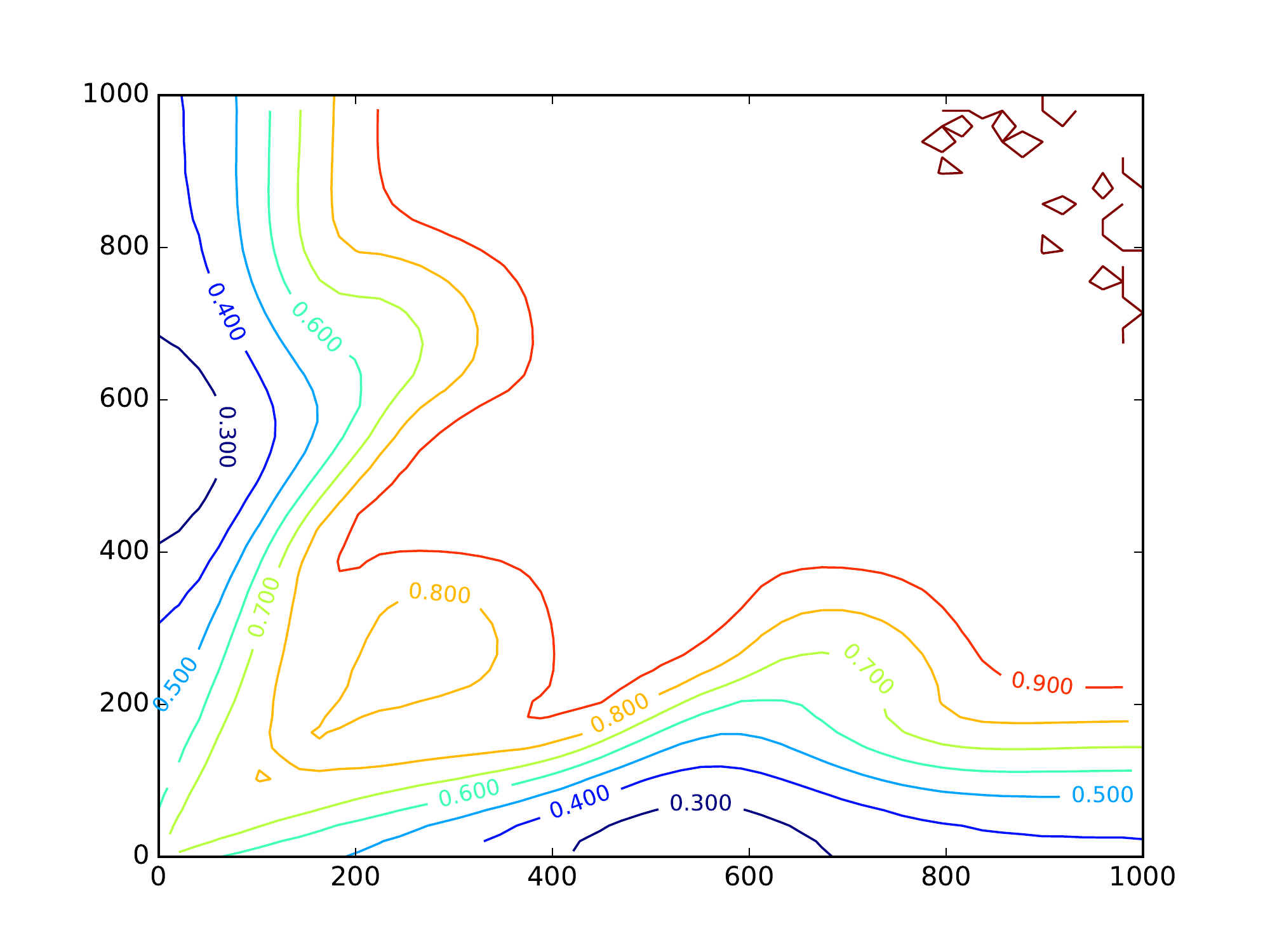}
		\caption{Contour plot at $t = 10$ years}
\end{subfigure}
	\caption{Concentration of the invading solvent in Test \ref{test:peaceman4} with $k=1$ and $\Delta t = 18$ on a Triangular mesh with a discontinuous permeability.}
	\label{fig:peaceman_test4_tri}
\end{figure}

\subsection{Comparison of higher order parameters} \label{sec:numres.comp}

Although the scheme is arbitrary order in space and can easily be extended to arbitrary order in time (by substituting the Crank-Nicolson time-stepping for a higher-order scheme such as backward differentiation), we argue that the $k=1$ scheme provides the best balance of accuracy and computational efficiency. In order to demonstrate that this is the case, we present the following tests.

\begin{test}\rm
To contrast the different solutions, we compare the principal quantity of interest to us, which is the total volume of oil recovered after ten years, as measured by the integral
\begin{equation*}
\int_{\Omega} \Phi(x) \pc_h^N(x).
\end{equation*}
For the $k = 1$ scheme as presented in Test \ref{test:peaceman2}, the total of volume of oil recovered can be measured to be $65.798\%$ of the total volume of the reservoir.
For the high-order tests, we range $k$ from $0$ to $3$, and replace the Crank-Nicolson time-stepping scheme with a high order backward differentiation formula of order $4$ in order to minimise the contribution of the temporal error. Additionally, in order to mitigate the majority of the extrapolation error produced by the pressure estimate when extrapolating $\pc_h$ (Equation \eqref{eqn:concentration_extrapolation}), we take a reduced time-step of $\Delta t = 7.2$ days (approximately $N = 500$ steps.)

Figure \ref{fig:oil_recovery_convergence} depicts the total recovery volume on each family of meshes in Table \ref{tab:mesh_parameters} with various polynomial degrees $k$. The $k = 0$ scheme is shown to perform quite poorly, producing results that are well out-of-line with the rest of the schemes. All of the other schemes however quickly converge to a similar estimate as the mesh size is refined, which is consistently within $1\%$ of the estimate produced by Test \ref{test:peaceman2} with only $k=1$, a Crank-Nicolson time-stepping, and $\Delta t=18$ days. We note that the scheme behaves well even on distorted meshes (the Kershaw mesh) until $k = 3$, where the linear system becomes too difficult to solve (Figure \ref{fig:oil_recovery_convergence} (c)). This solvability issue of HHO with high degrees
(remember that for $k=3$, the pressure equation is approximated with an order $2k=6$) on severely distorted
meshes has already been noticed even for the Poisson problem, and might be a consequence of rounding errors \cite{Dprivate}. A way to mitigate this poor conditioning of the system matrix on skewed meshes is
to change the local basis functions by applying a Gram--Schmidt orthonormalization process,
see \cite{BBCDPT12}. We however did not explore this option here as, in our experiments, the quality of the results do not significantly improve when using orders higher than $k = 1$ or $2$, and the computational cost increases drastically.

Finally, we also compare the high order schemes on the discontinuous permeability tensor of Test \ref{test:peaceman4}, using otherwise the same parameters. We use the mesh families, described in Table \ref{tab:mesh_parameters2}, whose edges are aligned with the discontinuities depicted in Figure \ref{fig:heterogeneous}. The results in Figure \ref{fig:oil_recovery_convergence_discontinuous} show that we obtain similar convergence pat\-terns to those of the homogeneous permeability tests. The $k=0$ scheme still produces results that are well out of line with the rest, while all $k \geq 1$ schemes convergence to a similar value as the mesh is refined.
\end{test}

\begin{figure}[h]
	\centering
	\begin{subfigure}{.49\textwidth}
		\centering
		\includegraphics[scale=0.3]{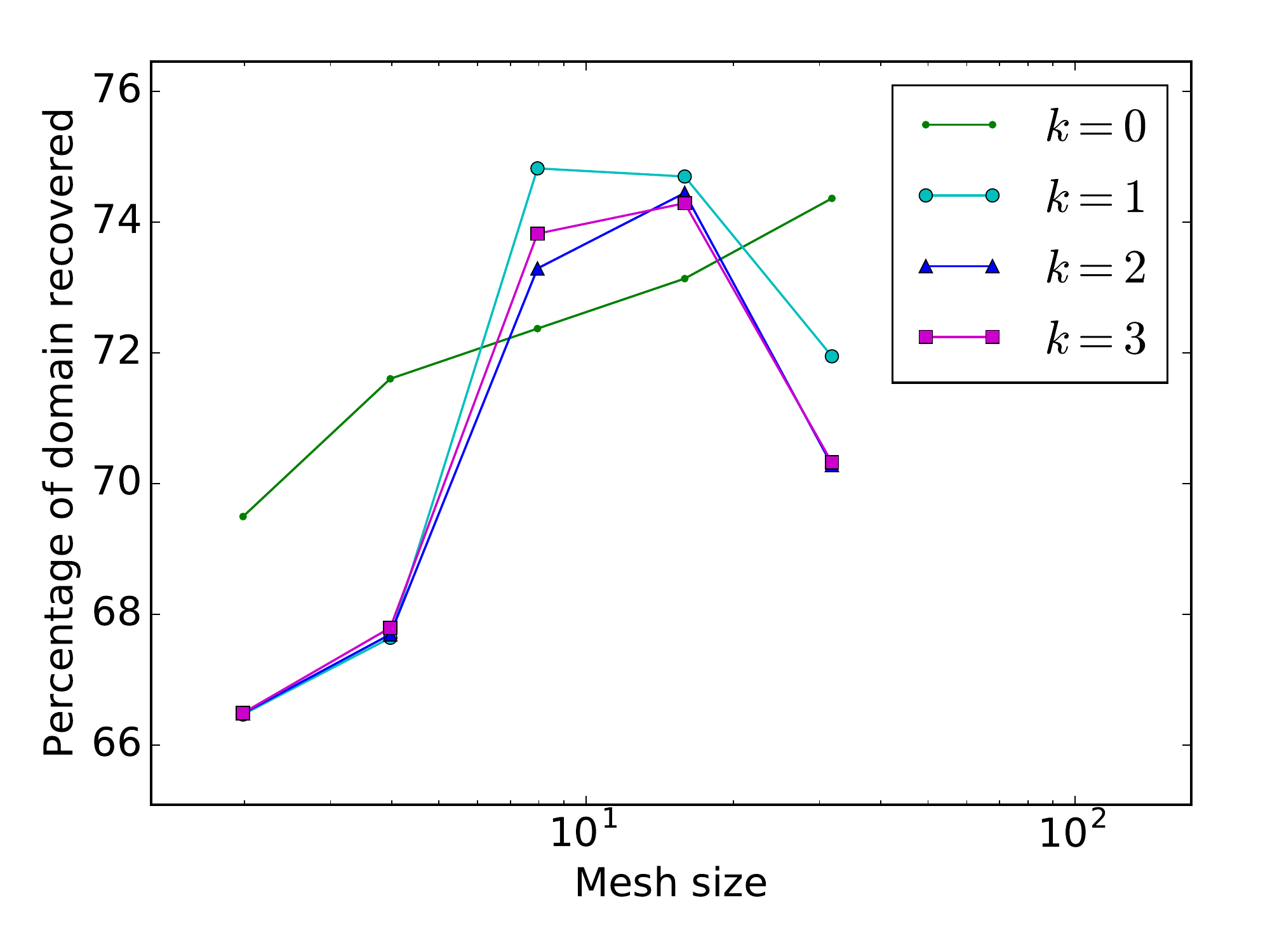}
		\caption{Recovery on triangular meshes}
	\end{subfigure}
	\begin{subfigure}{.49\textwidth}
		\centering
		\includegraphics[scale=0.3]{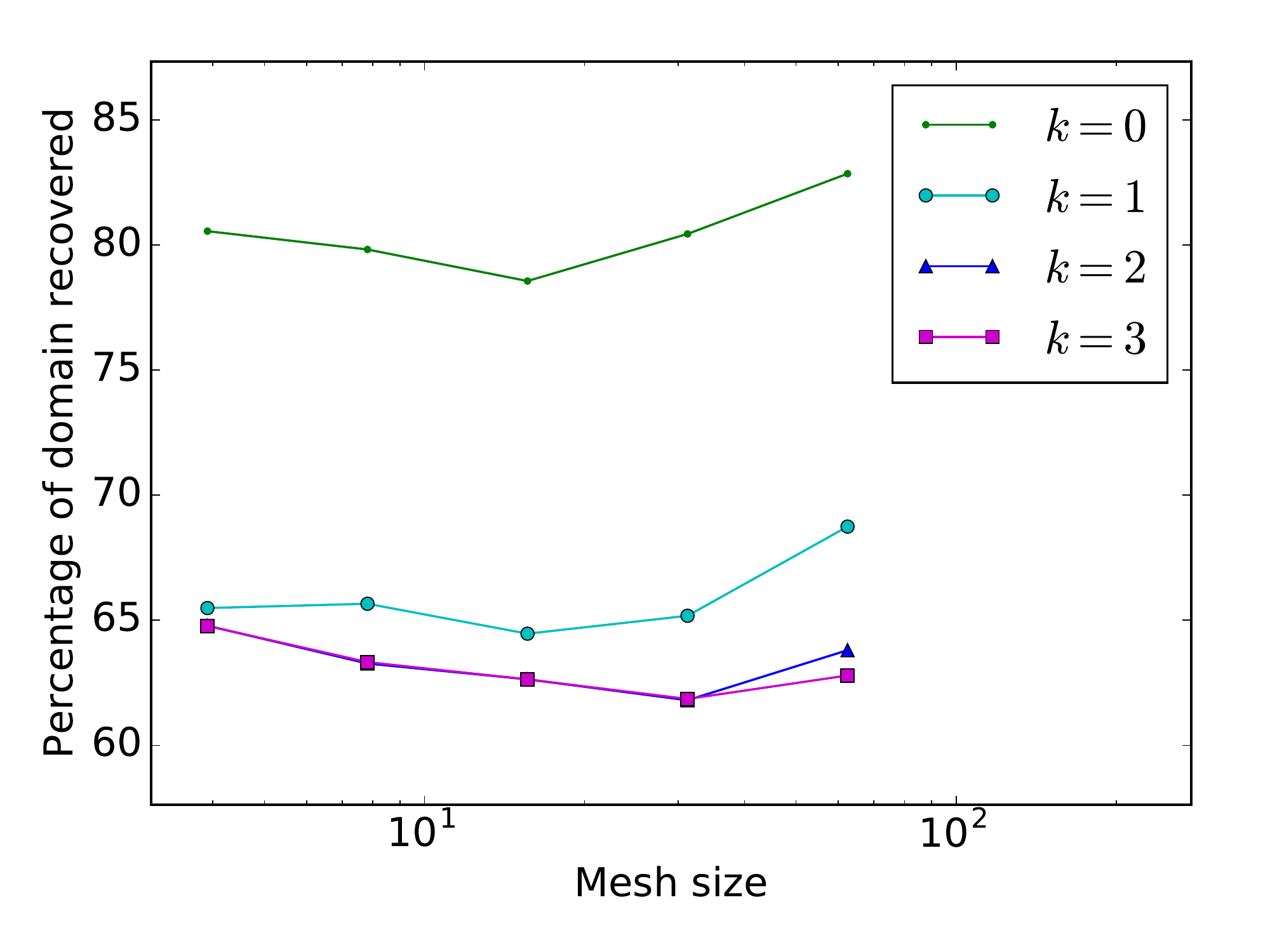}
		\caption{Recovery on Cartesian meshes}
	\end{subfigure}
	\begin{subfigure}{.49\textwidth}
		\centering
		\includegraphics[scale=0.3]{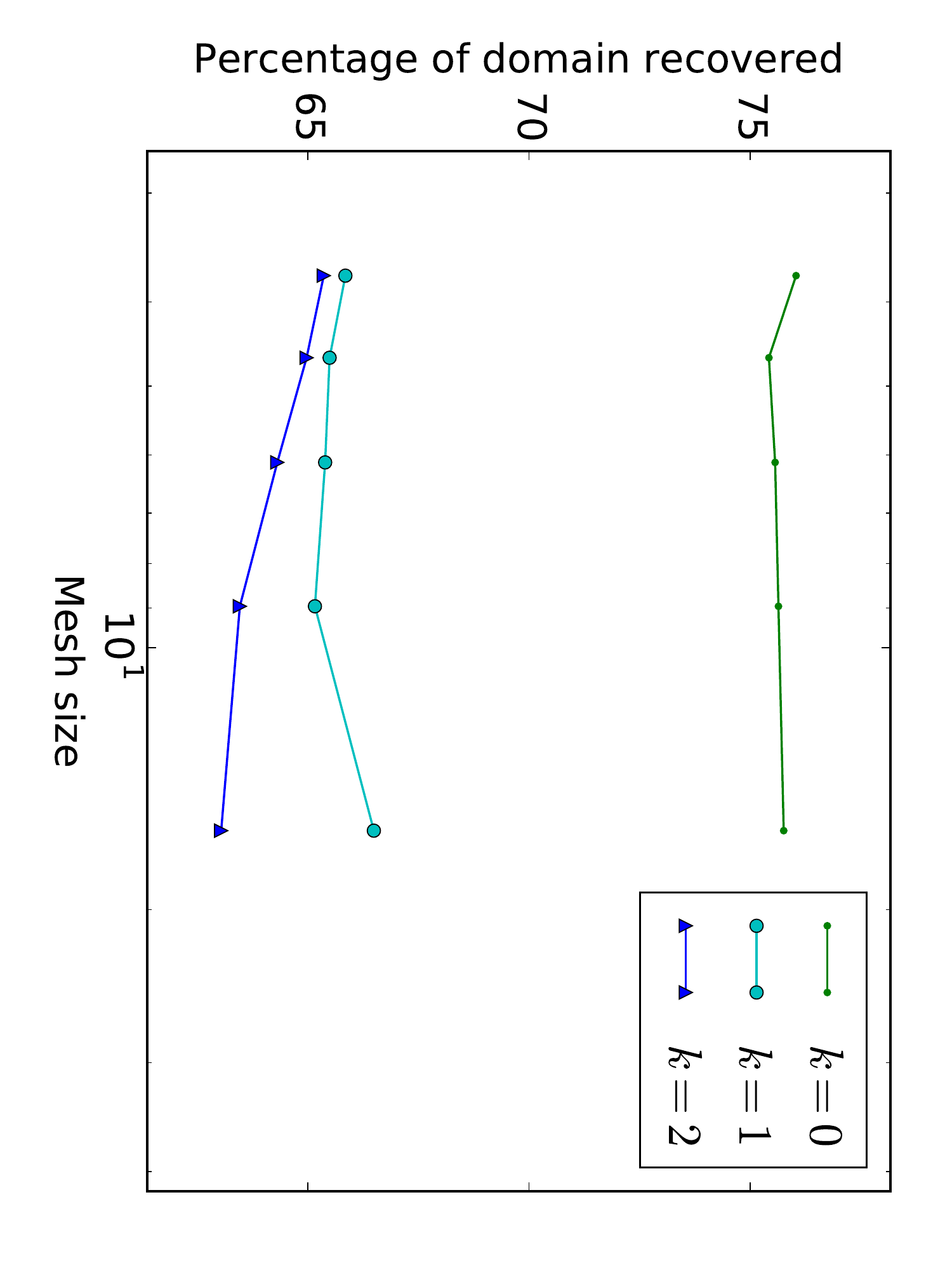}
		\caption{Recovery on Kershaw meshes}
	\end{subfigure}
	\begin{subfigure}{.49\textwidth}
		\centering
		\includegraphics[scale=0.3]{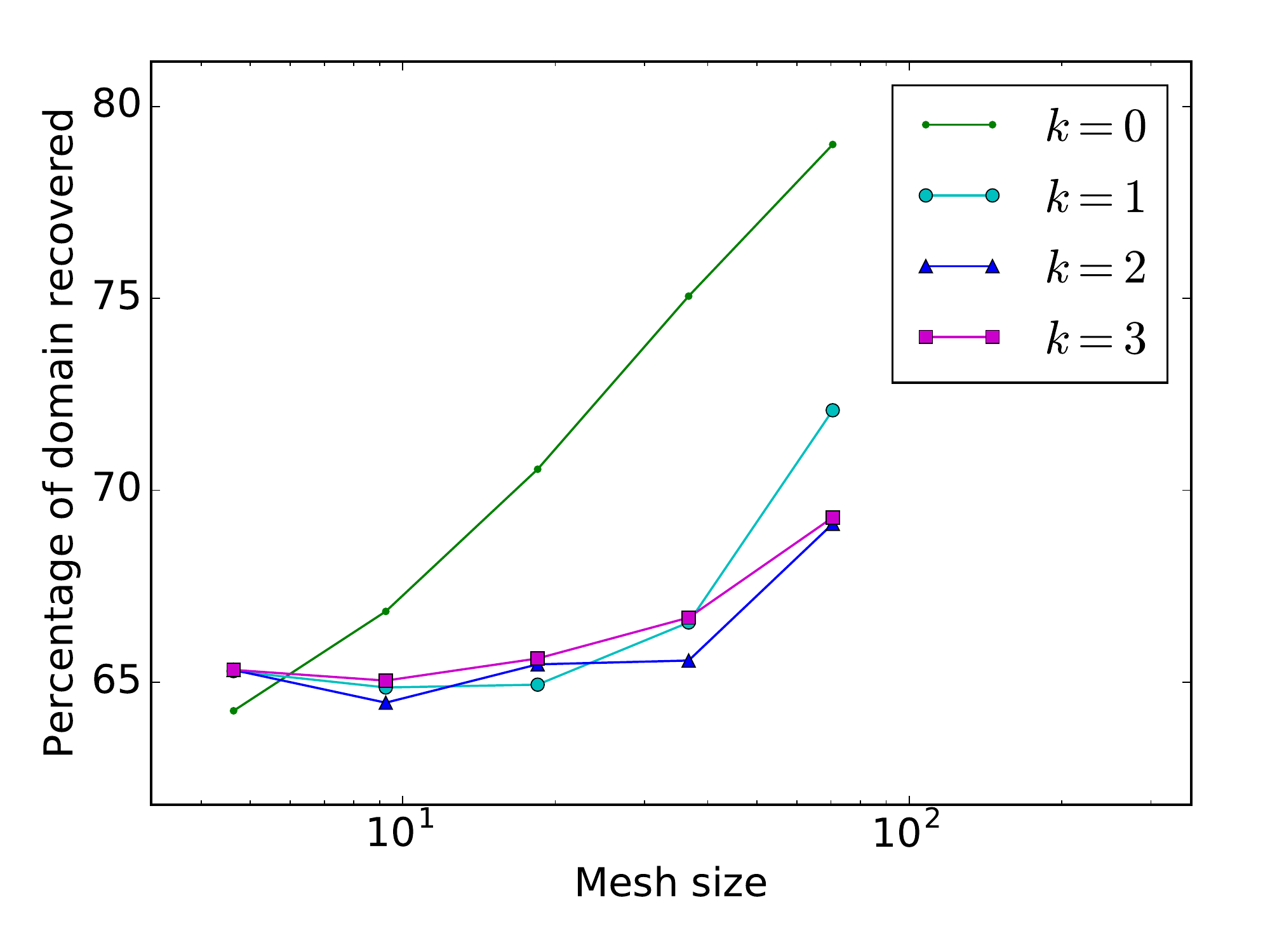}
		\caption{Recovery on hexagonal meshes}
	\end{subfigure}
	\caption{Total percentage of oil recovered from the reservoir with a uniform permeability,
after $10$ years and for various polynomial degrees and mesh sizes. To minimise the temporal error contribution, a high order backward difference time-stepping scheme with $\Delta t = 7.2$ is
used.}
	\label{fig:oil_recovery_convergence}
\end{figure}

\begin{figure}[h]
	\centering
	\begin{subfigure}{.49\textwidth}
		\centering
		\includegraphics[scale=0.3]{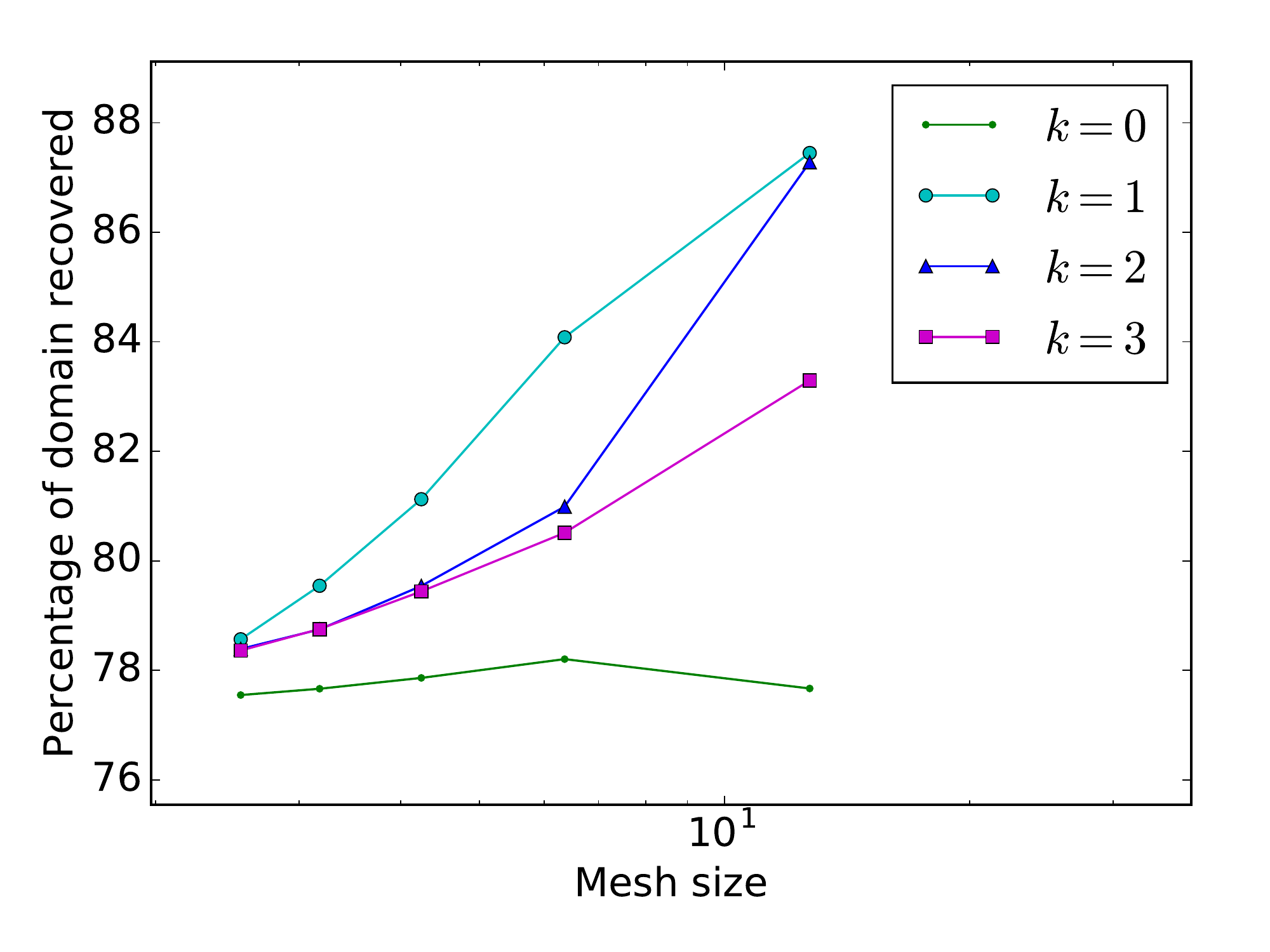}
		\caption{Recovery on triangular meshes}
	\end{subfigure}
	\begin{subfigure}{.49\textwidth}
		\centering
		\includegraphics[scale=0.3]{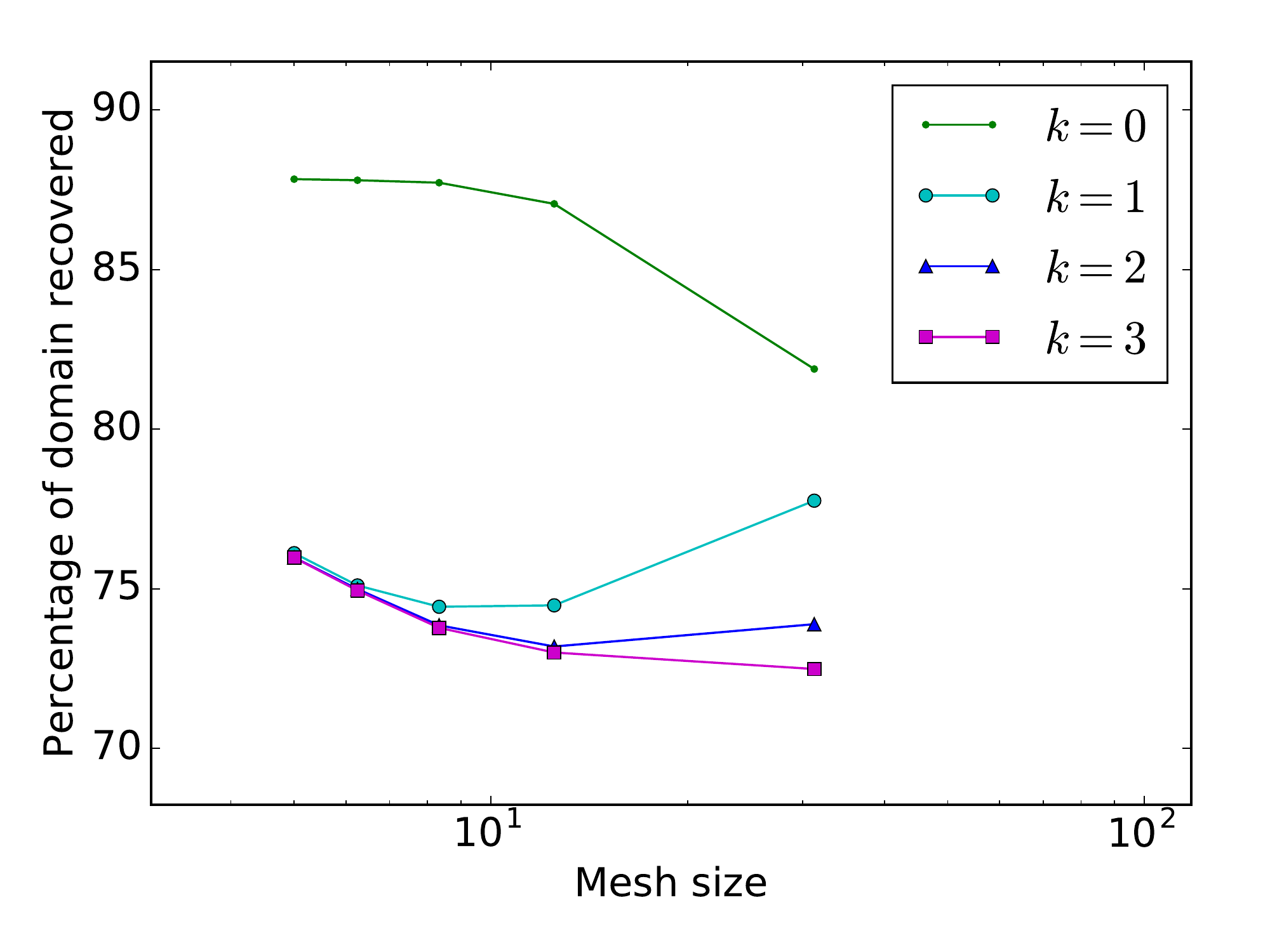}
		\caption{Recovery on Cartesian meshes}
	\end{subfigure}
	\caption{Total percentage of oil recovered from the reservoir with a discontinuous permeability tensor, after $10$ years for various polynomial degrees and mesh sizes.
 To minimise the temporal error contribution, a high order backward difference time-stepping scheme with $\Delta t = 7.2$ is
used}
	\label{fig:oil_recovery_convergence_discontinuous}
\end{figure}

\begin{test}\rm
The quality of the numerical approximations for various values of $k$ can also be observed visually. Depicted in Figures \ref{fig:compare_ks_t3} and \ref{fig:compare_ks} are the contour plots for the solution to Test \ref{test:peaceman2} using $k=0,1,2,3$ at time $t=3$ and $t=10$ respectively. The low-order solution using $k=0$ suffers from obvious grid effects in which the fluid mixture is progressing too rapidly along and clinging to the boundary of the domain. A similar effect is present in the results of the MFV scheme of \cite{droniou2007convergence}, suggesting that this artefact is a result of the low order of the scheme. Moving to a higher order scheme, even just $k = 1$, remedies this effect and shows the solvent mixture progressing in a physically realistic pattern.

The results of Figures \ref{fig:oil_recovery_convergence}, \ref{fig:oil_recovery_convergence_discontinuous}, \ref{fig:compare_ks_t3} and \ref{fig:compare_ks}
also show that there is little advantage in selecting a spatial order $k\ge 2$, since the results for
these higher order are qualitatively and quantitatively similar to those obtained with $k=1$.
\end{test}

\begin{figure}[h]
	\centering
	\begin{subfigure}{.49\textwidth}
		\centering
		\includegraphics[scale=0.3]{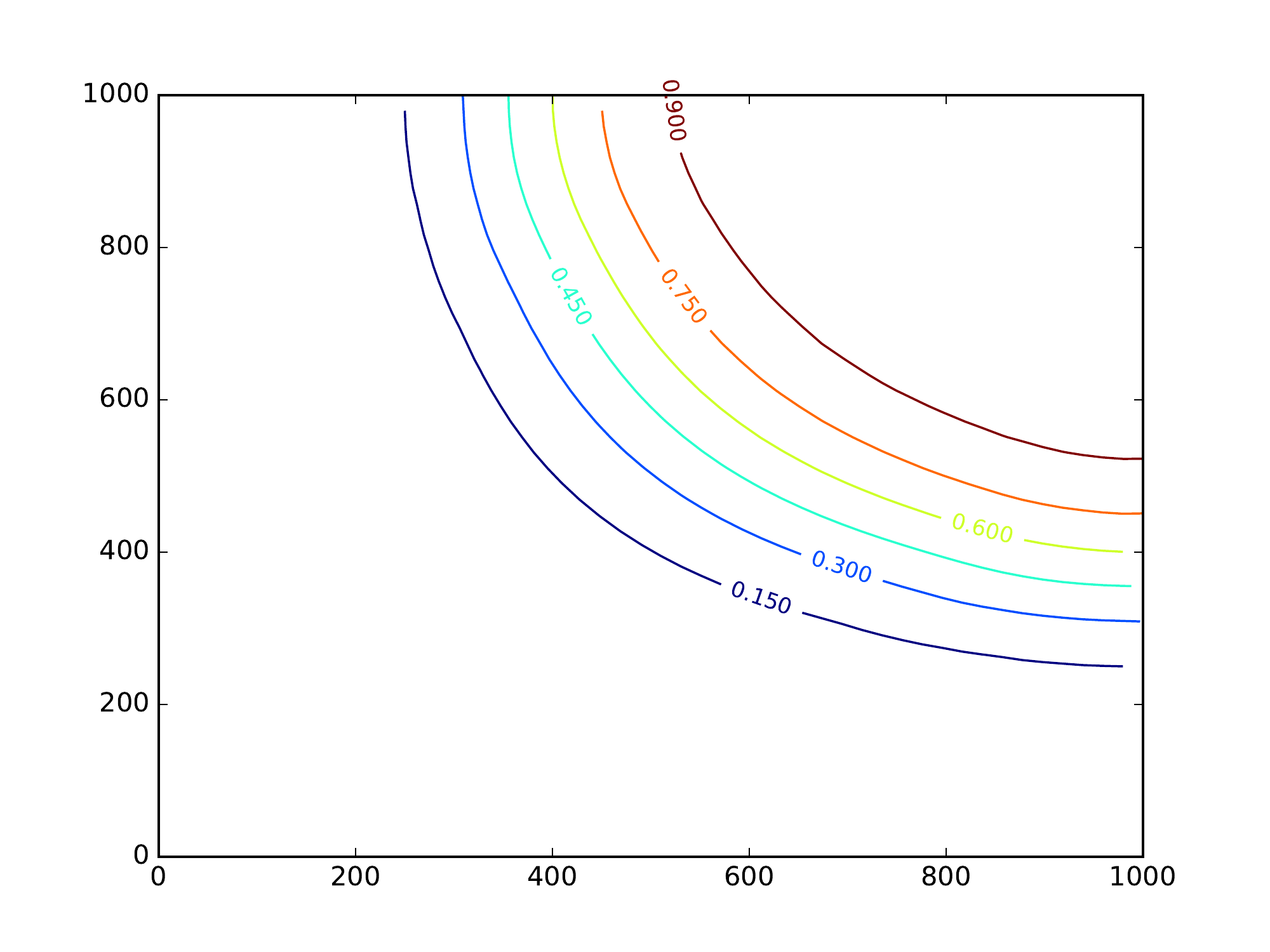}
		\caption{$k = 0$}
	\end{subfigure}%
	\begin{subfigure}{.49\textwidth}
		\centering
		\includegraphics[scale=0.3]{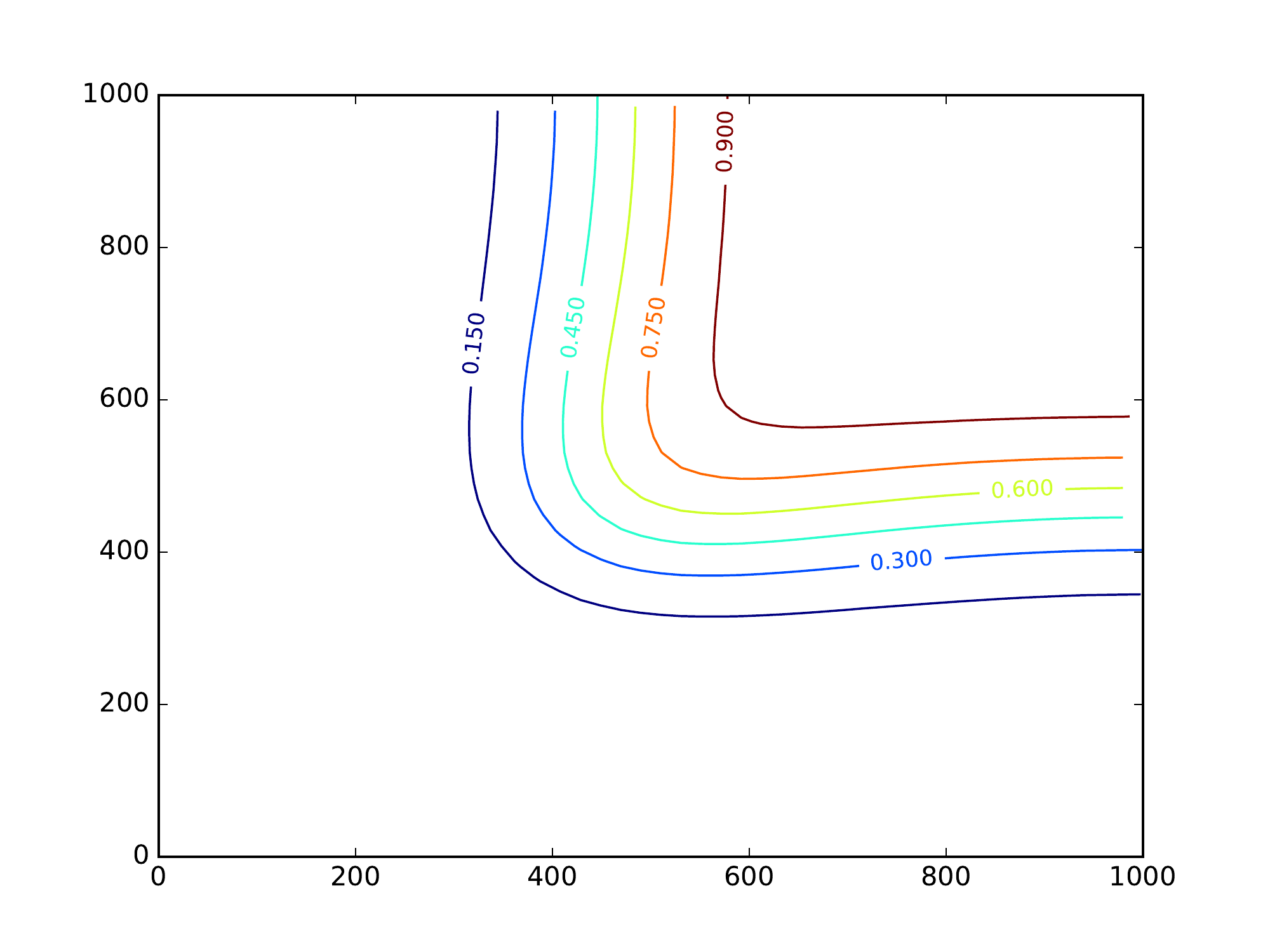}
		\caption{$k = 1$}
	\end{subfigure}
	\begin{subfigure}{.49\textwidth}
		\centering
		\includegraphics[scale=0.3]{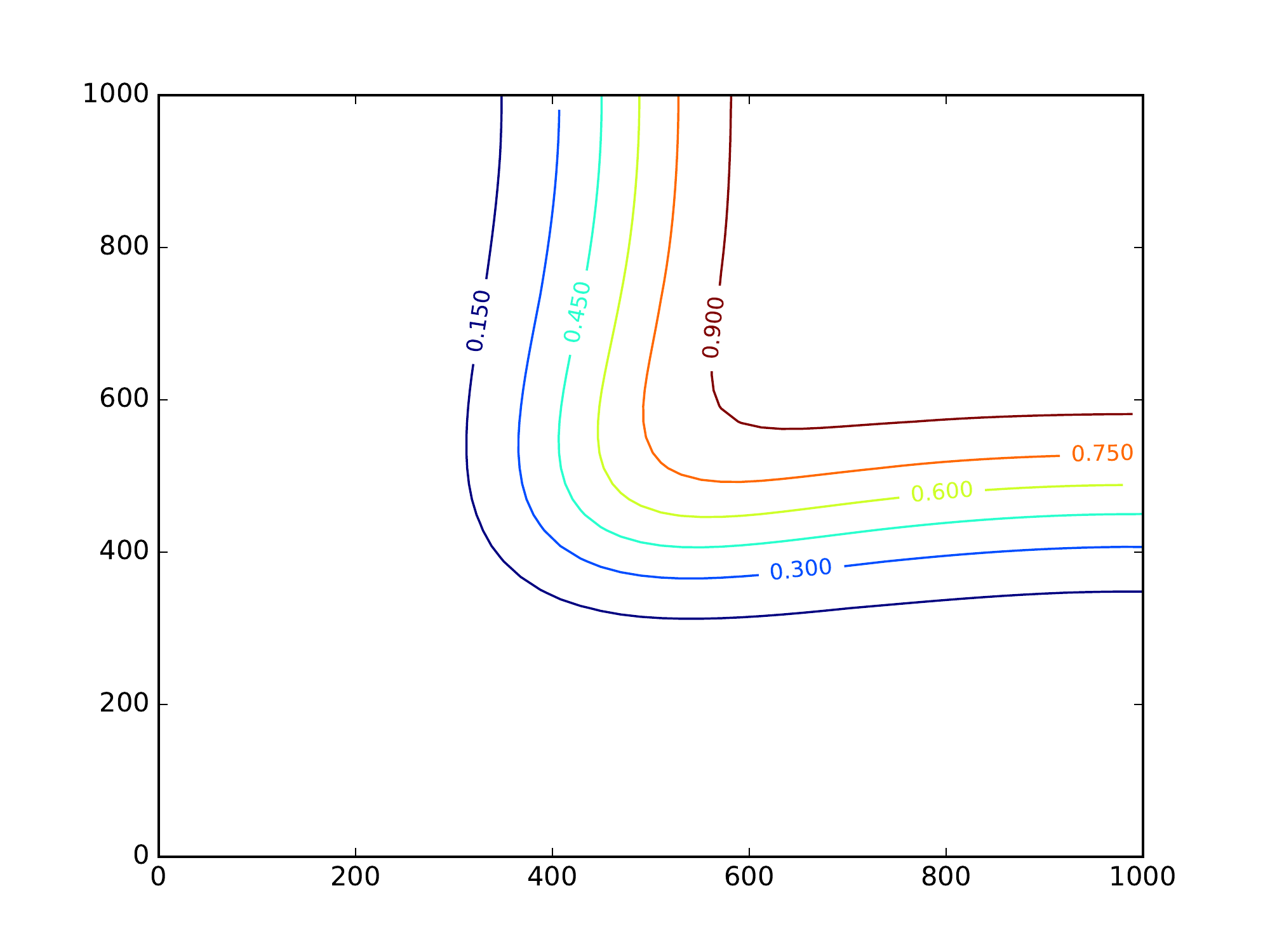}
		\caption{$k=2$}
	\end{subfigure}%
	\begin{subfigure}{.49\textwidth}
		\centering
		\includegraphics[scale=0.3]{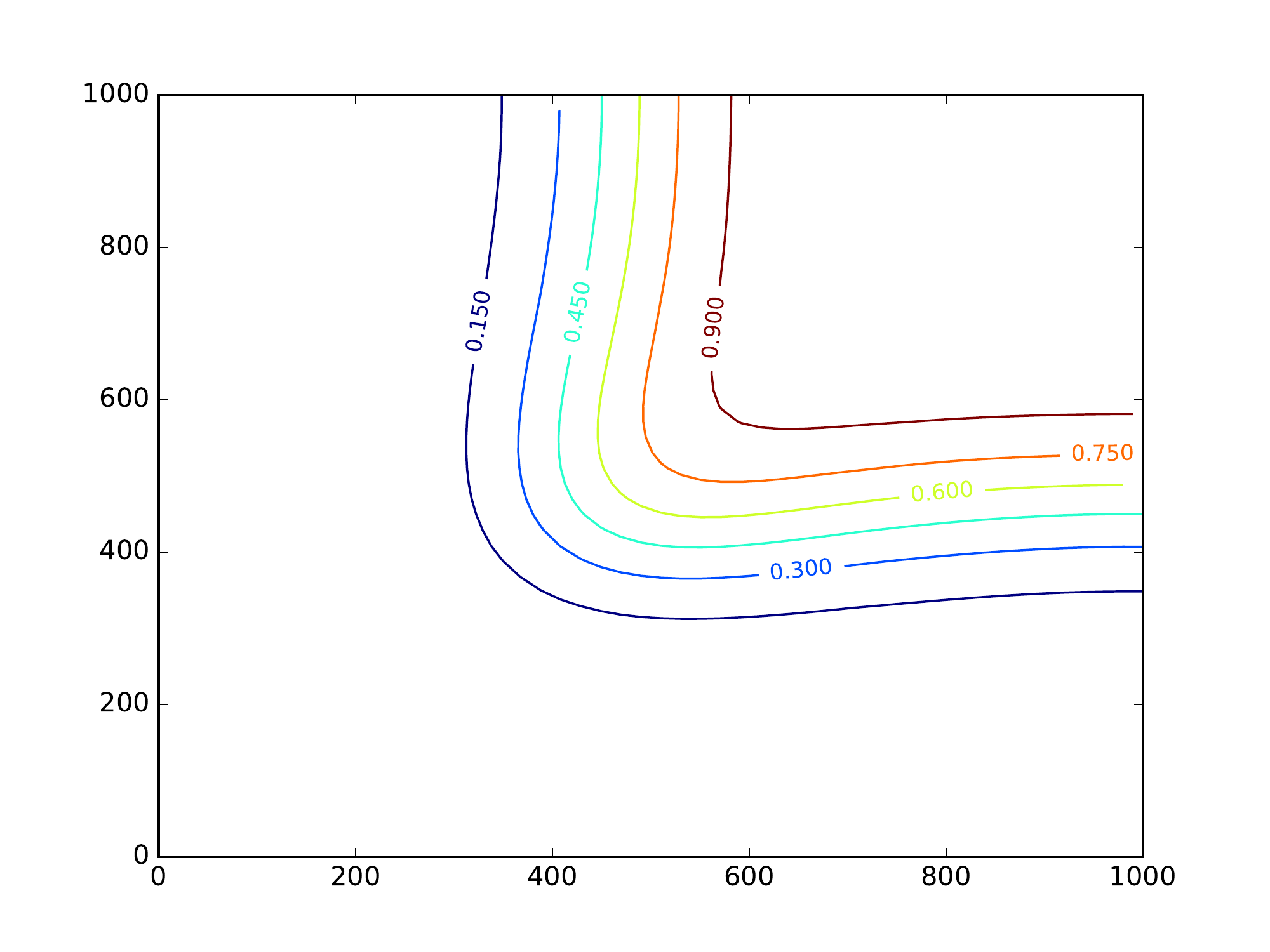}
		\caption{$k=3$}
	\end{subfigure}
	\caption{Comparison of the quality of the numerical approximation for $k = 0,1,2,3$ on a Cartesian mesh at time $t = 3$ years. For $k = 0$, the fluid mixture clings to and overly rapidly progresses along the reservoir boundary. Moreover, the expected fingering effect is not visible. These issues are remedied by already selecting $k=1$, and further increases of the order does
not noticeably impact the solution (at least visually).}
	\label{fig:compare_ks_t3}
\end{figure}

\begin{figure}[h]
\centering
\begin{subfigure}{.49\textwidth}
  \centering
  \includegraphics[scale=0.3]{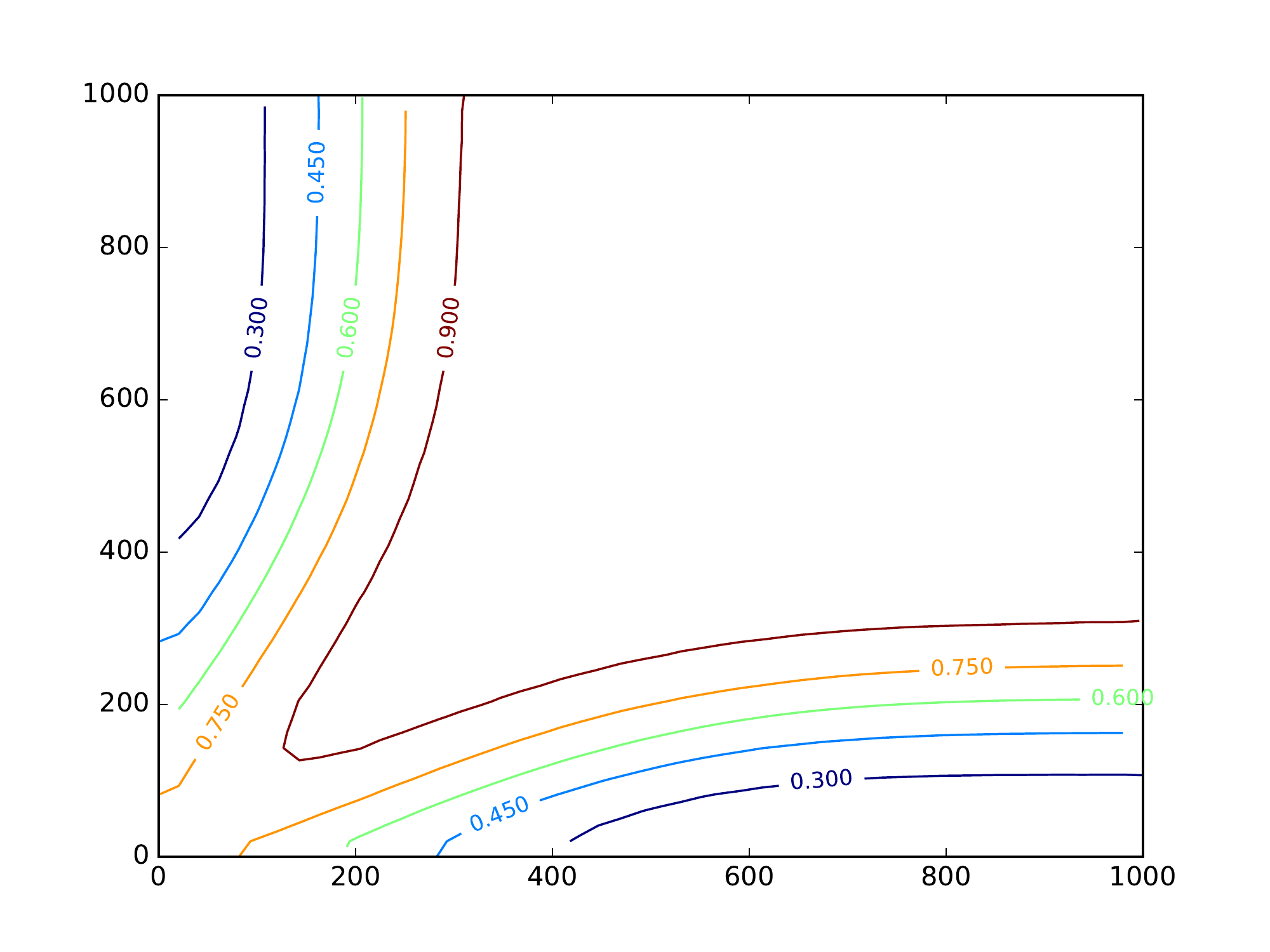}
  \caption{$k = 0$}
\end{subfigure}%
\begin{subfigure}{.49\textwidth}
  \centering
  \includegraphics[scale=0.3]{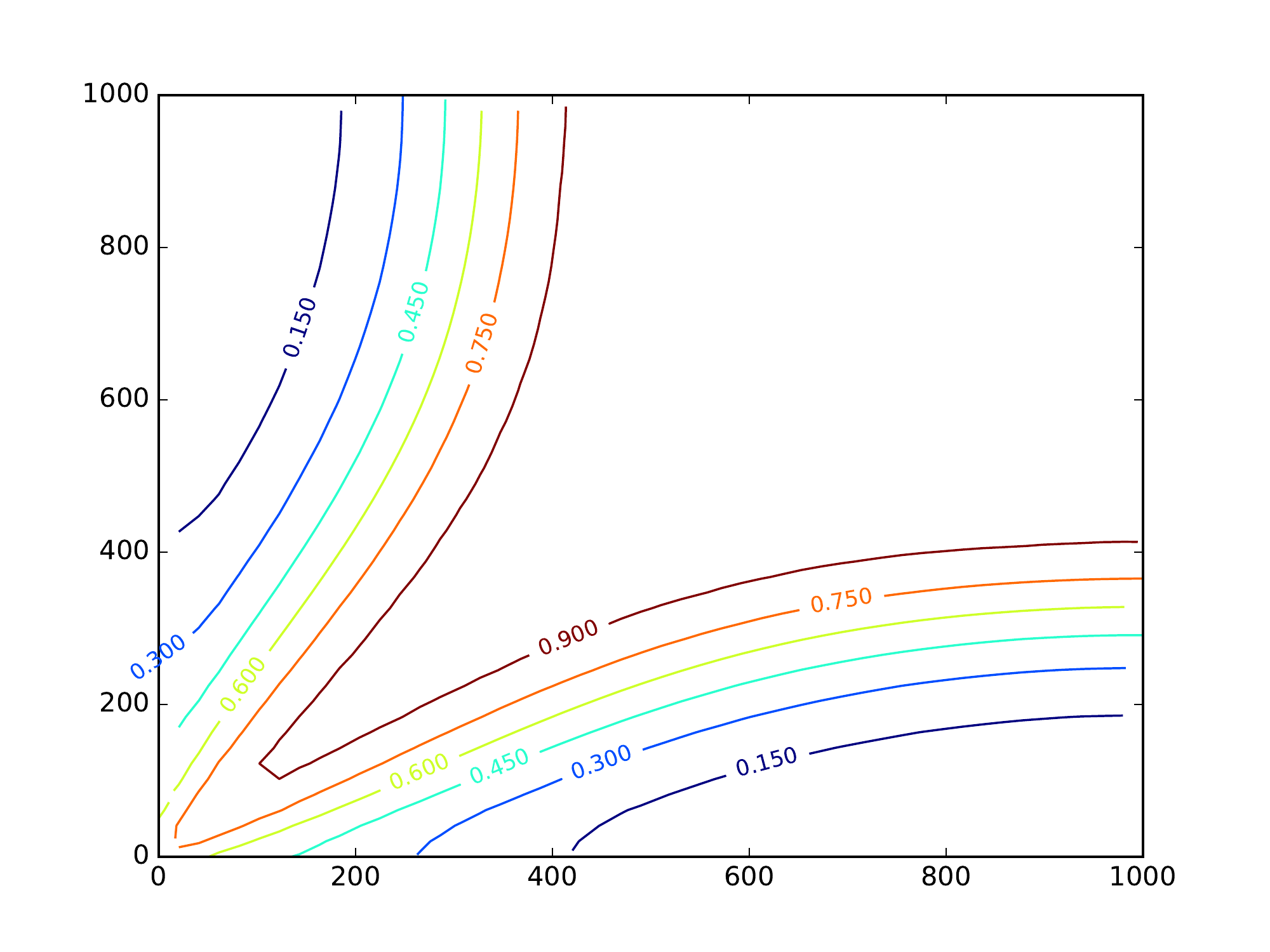}
  \caption{$k = 1$}
\end{subfigure}
\begin{subfigure}{.49\textwidth}
  \centering
  \includegraphics[scale=0.3]{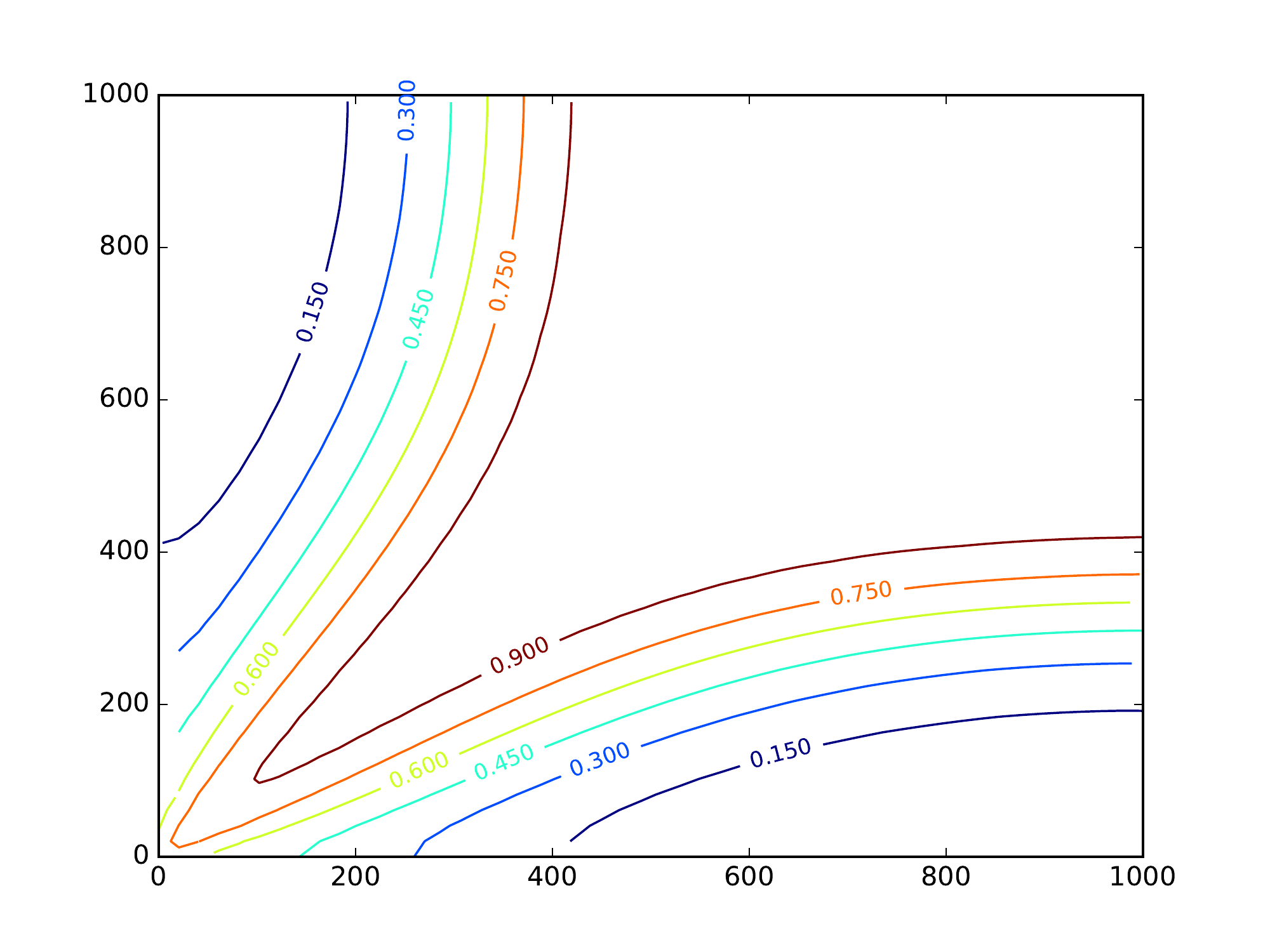}
  \caption{$k=2$}
\end{subfigure}%
\begin{subfigure}{.49\textwidth}
  \centering
  \includegraphics[scale=0.3]{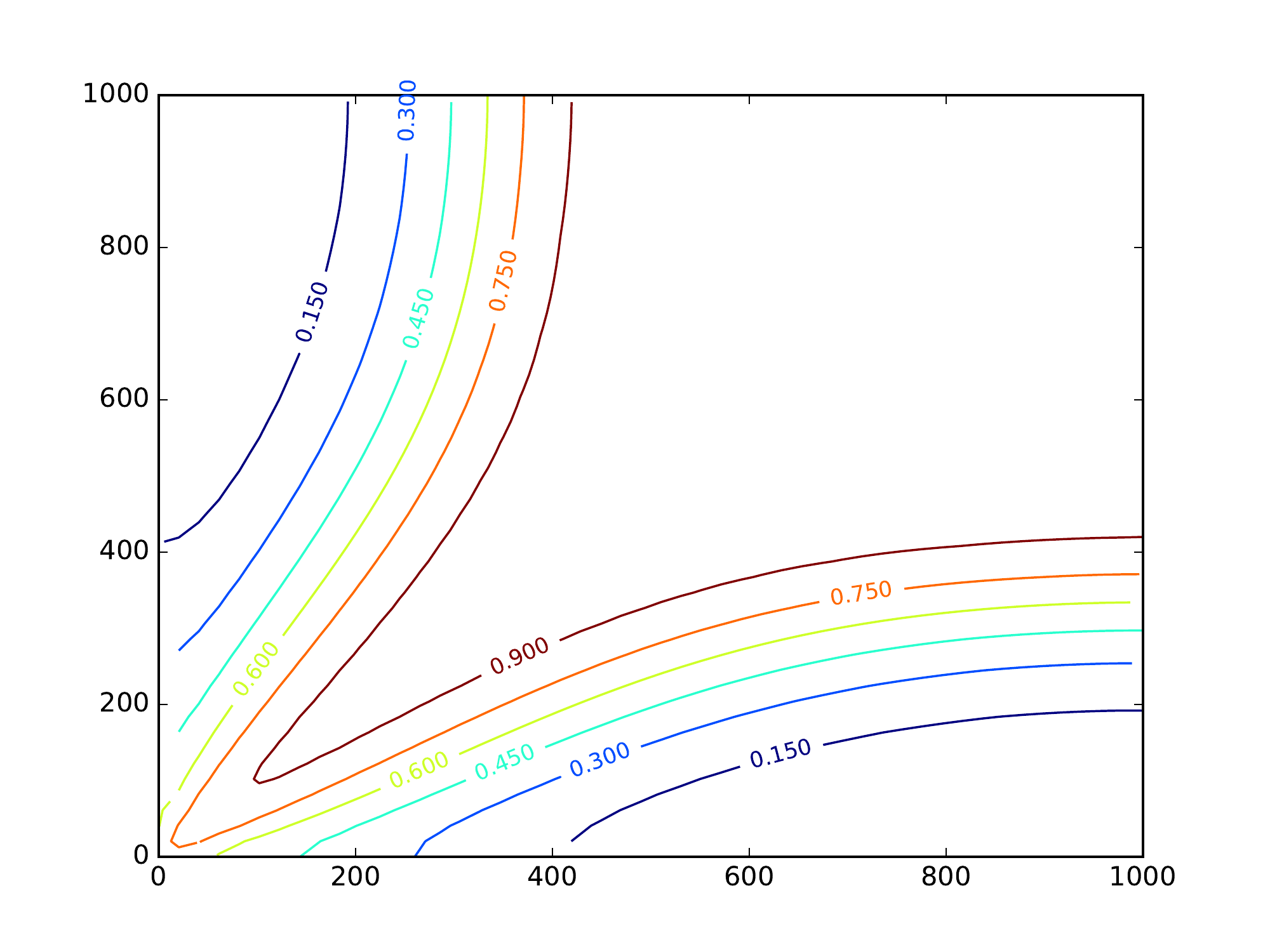}
  \caption{$k=3$}
\end{subfigure}
\caption{Comparison of the quality of the numerical approximation for $k = 0,1,2,3$ on a Cartesian mesh at time $t = 10$ years.}
\label{fig:compare_ks}
\end{figure}

\subsection{Computational cost}\label{sec:comp.cost}
High-order versions of the HHO scheme have been shown to produce very reliable results. This increased
accuracy however obviously comes with a higher computational cost than low-order schemes. The HHO
method has been designed to allow for a static condensation of the cell unknowns: by local Gaussian elimination, the systems \eqref{eqn:discrete_pressure_equation} and \eqref{eqn:discrete_advection_diffusion_equation}
can be expressed in terms of the face unknowns, resulting in
systems on the face unknowns only and with the same sparsity structure as the original equations.
On a given mesh $\scriptM_h = (\scriptT_h, \scriptF_h)$ with polynomials of degree $m$, the
global system is therefore of size
\[
\binom{m + d - 1}{d - 1}|\scriptF_h|,
\]
and has an $\mathcal O(m^{d-1})$ growth with respect to the polynomial degree.
Recall that the pressure is solved at a degree $m=2k$.
The cost of the time discretisations can be considered independent since it is clearly linear in the number of time-steps $N$. 

Figure \ref{fig:peaceman_computational_time} illustrates the relative costs of the methods for various degrees $k$. Depicted are the average times taken per step using $N = 100$ time-steps for the data given in Test \ref{test:peaceman2} on the mesh families in Table \ref{tab:mesh_parameters}. We emphasise that our implementation is not optimised for high performance, and that the tests are performed on a personal computer. These measurements are not intended to give an absolute estimate of the cost, only a comparison of running times of the various order schemes (this comparison is valid since all tests were done on the same computer). It is expected that, even if the times vary from one computer to the other, the relative positions of the curves corresponding to various $k$ will be similar to those in Figure \ref{fig:peaceman_computational_time}.

As predicted, the running times begin to grow very rapidly for fine meshes with high-order $k$. Since the number of degrees of freedom of the scheme is directly tied to the number of faces in the mesh, the execution time can be seen to be larger for meshes with a greater number of faces (see Table \ref{tab:mesh_parameters}).
Combined with the qualitative and quantitative results in Sections \ref{sec:numres} and \ref{sec:numres.comp},
these relative running times further supports our argument that the $k = 1$ scheme may be the best balance of accuracy and speed.

\begin{figure}[h]
	\centering
	\begin{subfigure}{.49\textwidth}
		\centering
		\includegraphics[scale=0.3]{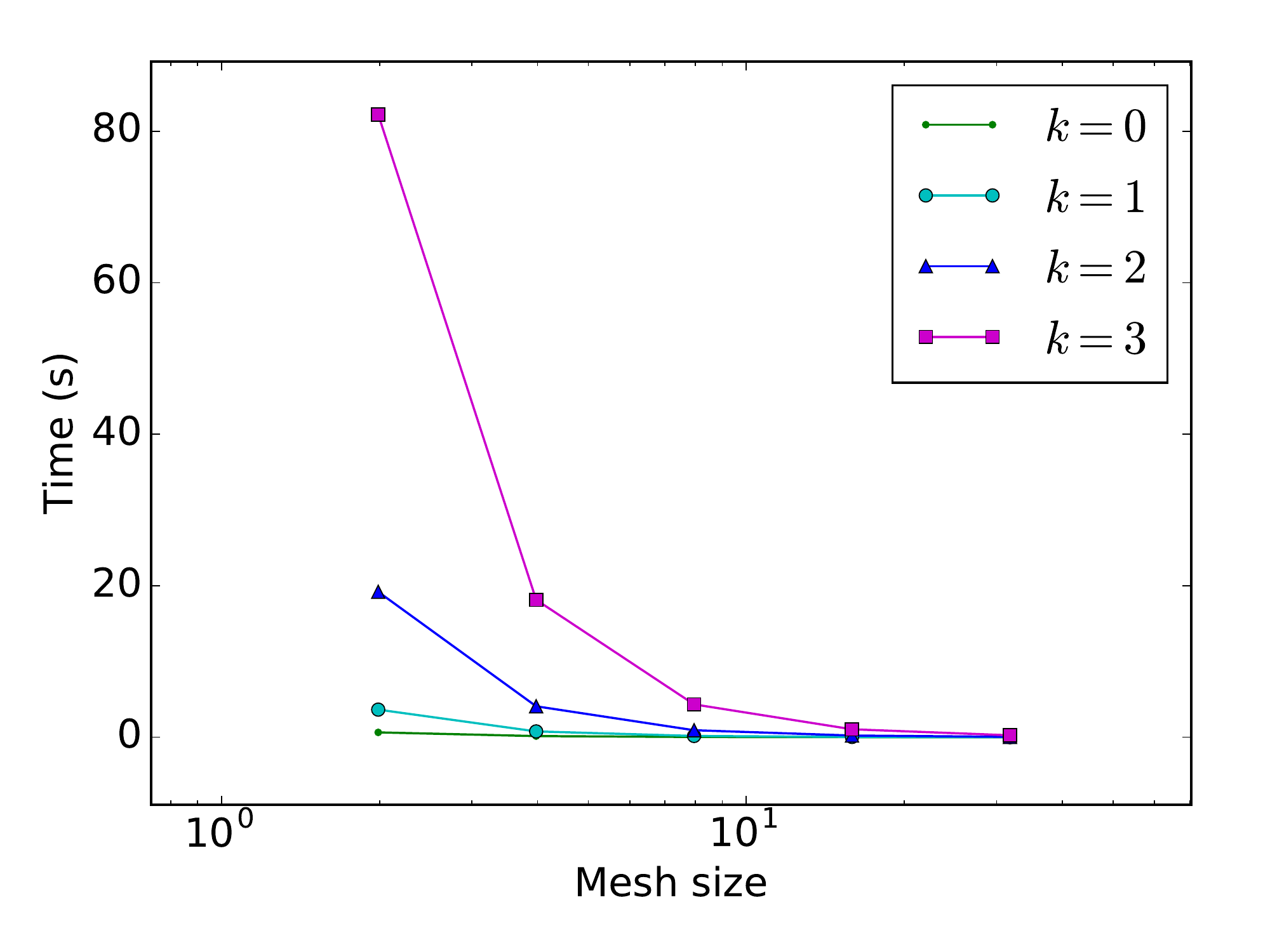}
		\caption{Execution time on triangular meshes}
	\end{subfigure}
	\begin{subfigure}{.49\textwidth}
		\centering
		\includegraphics[scale=0.3]{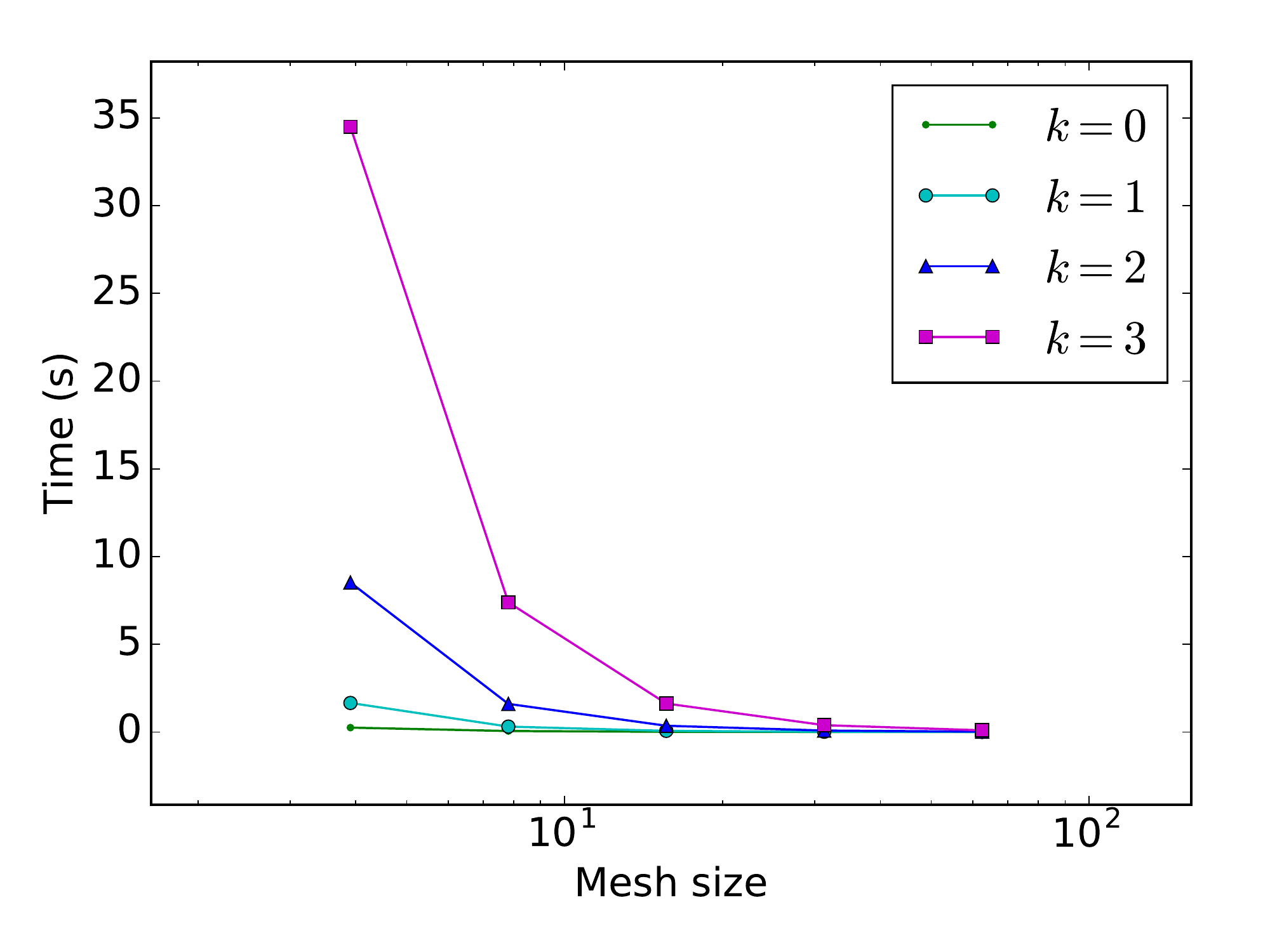}
		\caption{Execution time on Cartesian meshes}
	\end{subfigure}
	\begin{subfigure}{.49\textwidth}
		\centering
		\includegraphics[scale=0.3]{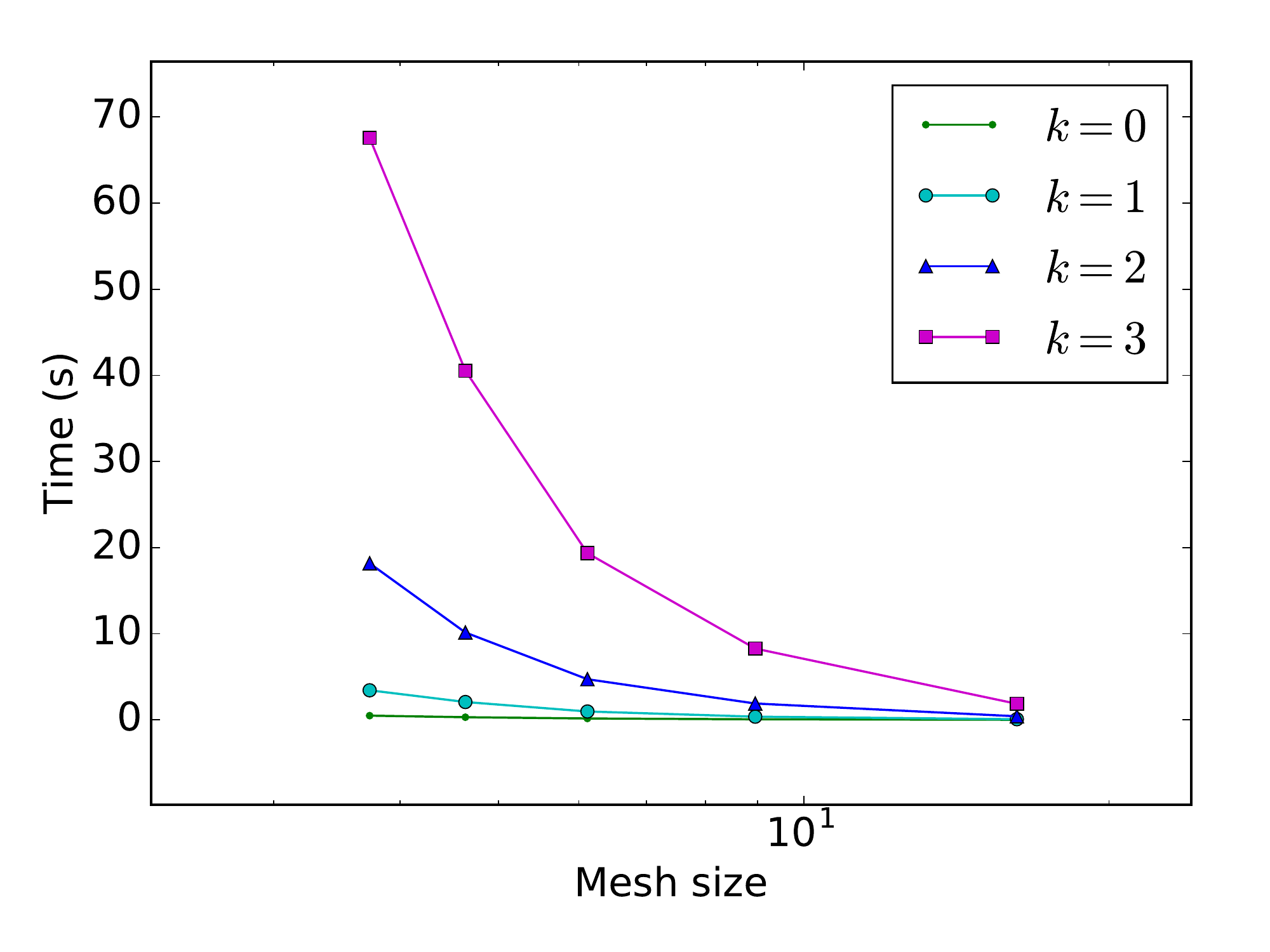}
		\caption{Execution time on Kershaw meshes}
	\end{subfigure}
	\begin{subfigure}{.49\textwidth}
		\centering
		\includegraphics[scale=0.3]{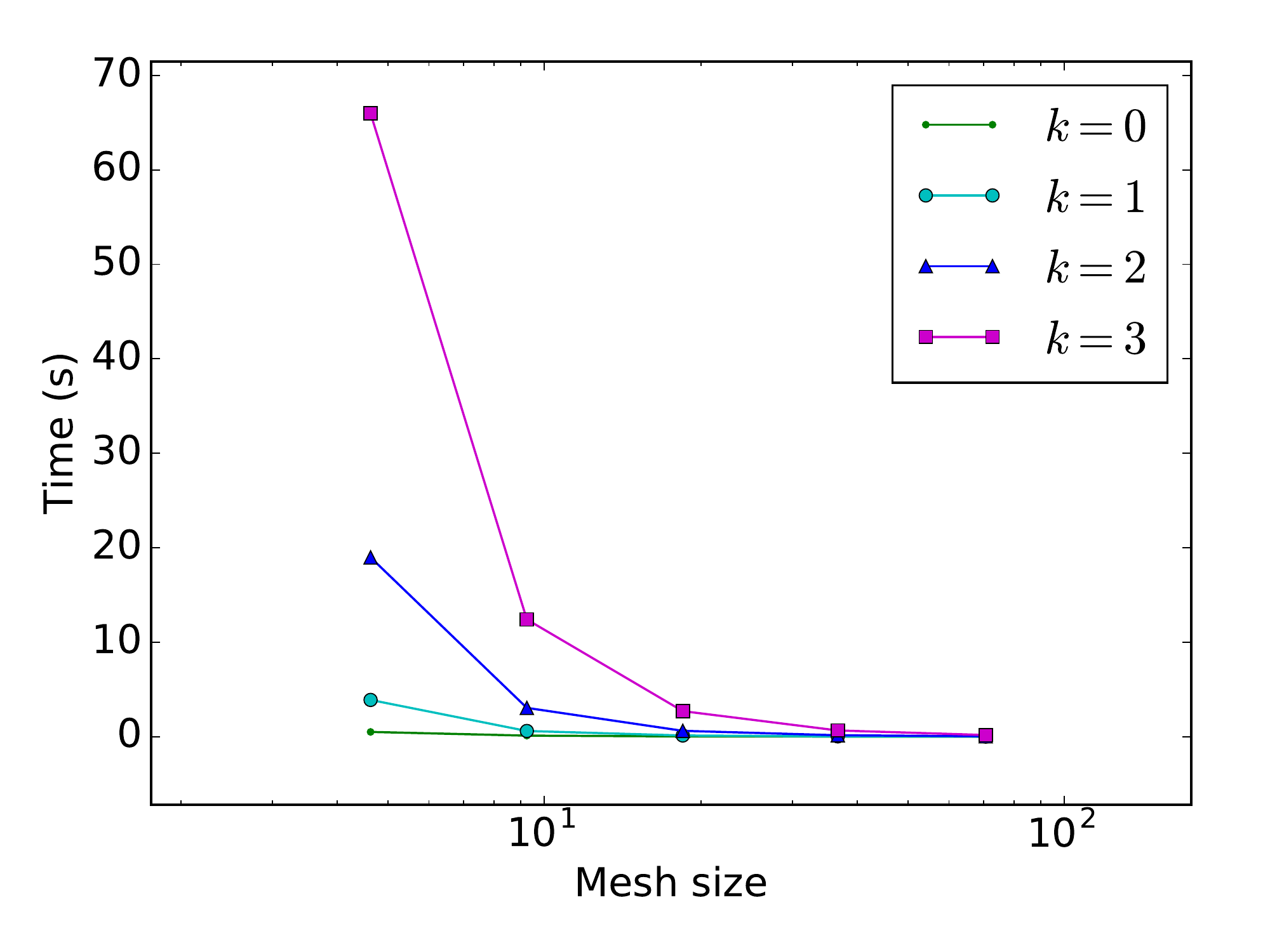}
		\caption{Execution time on hexagonal meshes}
	\end{subfigure}
	\caption{The average execution time for one time-step of Test \ref{test:peaceman2} on various meshes.}
	\label{fig:peaceman_computational_time}
\end{figure}

\section{Conclusion}\label{sec:conclusion}

We designed and implemented an arbitrary-order scheme for a miscible incompressible
flow model used in tertiary oil recovery. The scheme was based on the Hybrid High-Order method,
and is applicable on any kind of polygonal or polyhedral meshes.
To ensure the stability of the numerical approximation, the pressure equation has to be approximated
with an order twice as large as the order used for the concentration equation,
and special care must be taken in reconstructing advective Darcy fluxes from the approximate
pressure.

We produced several numerical tests on classical 2D test cases encountered in the literature. These
tests show that the best balanced of accuracy vs. computational cost is obtained by selecting a spatial order $k=1$ and
a Crank-Nicolson time stepping. The results are stable with respect to the geometry of the meshes, in the sense
that the total recovered oil is similar for all meshes and $k\ge 1$. Selecting an order $k=1$
fixes grid effects that are present with $k=0$ (as in low-order finite volume methods).

\section{Appendix: existence and stability of the solution to the scheme}\label{sec:appen}

Let $\scriptM_h = (\scriptT_h, \scriptF_h)$ be a mesh. As in \cite{di2011mathematical}, we take
a matching simplicial subdivision $\mathcal I_h$ of $\scriptM_h$ and $\varrho>0$
such that for any simplex $S\in\mathcal I_h$ of diameter $h_S$ and inradius $r_S$, $\varrho h_S\le r_S$, and
for all $T\in\scriptT_h$ and all $S\in\mathcal I_h$ such that $S\subset T$, $\varrho h_T \le h_S$.
In the following, $A\lesssim B$ means that $A\le CB$ for some $C$ depending only on $\Omega$ and $\varrho$
(not on $h$).



The following lemma states a stability result for the discrete elliptic bilinear form.
A proof can be found for $\Lambda={\rm Id}$ in \cite{di2014arbitrary}, and a sketch for
extending this to generic $\Lambda$ is given in \cite{di2015hybrid}.

\begin{lemma} \label{lemma:diffusive_bilinear_form_is_positive}
	Let $\bLambda$ be a bounded, symmetric, uniformly coercive tensor-valued function on $\Omega$. Then
 for any $\tf_h \in \dof^k_h$, the discrete diffusive bilinear form $\dbf_{h,\bLambda}$ \eqref{eqn:global_pressure_diffusion_bilinear_form} satisfies
	\begin{equation} \label{eqn:diffusive_bilinear_form_seminorm_inequality}
	\dbf_{h,\bLambda}(\tf_h, \tf_h) \gtrsim \sum_{T\in\scriptT_h} \| \bLambda^\frac{1}{2} \nabla \pf_T \|^2_{L^2(T)} + \sum_{T\in\scriptT_h} \sum_{F\in\scriptF_T} \frac{{\bLambda_{TF}}}{h_F} \|\pf_F - \pf_T \|^2_{L^2(F)}.
	\end{equation}
\end{lemma}

The following lemma is the key ingredient in proving the stability of the solution
to the scheme. It holds true because of the specific choice of reconstructed Darcy velocity,
chosen to be compatible with the discretisation of the concentration equation.

\begin{lemma} \label{lemma:advective_bilinear_form_is_positive}
	Let $R$ be given by \eqref{eqn:concentration_reaction_terms} and $((\RU_T)_{T\in\scriptT_h}, (\fv_{TF})_{T\in\scriptT_h, F\in\scriptF_T})$ be the reconstructed Darcy velocity \eqref{eqn:discrete_fluxes}--\eqref{eqn:darcy_velocity_volumetric_reconstruction} corresponding to a solution $\dpdof_h^\hs$ of the discrete pressure equation \eqref{eqn:discrete_pressure_equation}. Then for any $\tf_h \in \dof^k_h$, the discrete advection--reaction bilinear form \eqref{eqn:global_advection_bilinear_form} satisfies
	\begin{equation} \label{eqn:advective_bilinear_form_seminorm_inequality}
	\dbf_{h,R,\RU}(\tf_h, \tf_h) \geq \int_\Omega \frac{2\Phi}{\Delta t} \pf^2_h + \sum_{T\in\scriptT_h} \sum_{F\in\scriptF_T} \frac{1}{2}\int_F |\fv_{TF}| (\pf_T - \pf_F)^2.
	\end{equation}
\end{lemma}

\begin{proof}
{\scshape Proof}: 
	Equation \eqref{eqn:local_advection_bilinear_form} and \eqref{eqn:global_advection_bilinear_form} give
	\begin{equation}
	\dbf_{h,R, \RU}(\tf_h, \tf_h) = \sum_{T\in\scriptT_h} \left\{ -\lip{T}{\pf_T}{(\ad_{T,\RU}^k \tf_{T})} + \lip{T}{R\pf_T^2}{}  + \stab^-_{\RU,T}(\tf_{T}, \tf_{T}) \right\} .
	\end{equation}
	Expanding via the definitions of the advective derivative $\ad_{T,\RU}^k$ (Definition \ref{def:advective_derivative}), the advection stabilisation $\stab^-_{\RU,T}$ \eqref{eqn:advective_stabilisation_form} and the reaction terms $R$
\eqref{eqn:concentration_reaction_terms} and collecting cell and face terms, we find
	\begin{equation} \label{eqn:advection_stability_terms}
	\begin{split}
	\dbf_{h,R, \RU}&(\tf_h, \tf_h) =\sum_{T\in\scriptT_h} \left\{ \lip{T}{\frac{2\Phi}{\Delta t}}{\pf_T^2} -  \lip{T}{(\RU_T \cdot \nabla \pf_T)}{\pf_T} + \lip{T}{q^-(t^\hs,\cdot) \pf_T^2}{} \right\} \\& + \sum_{T\in\scriptT_h} \sum_{F \in \scriptF_T} \left\{ -\lip{F}{\fv_{TF} (\pf_F - \pf_T)}{\pf_T} +  \lip{F}{[\fv_{TF}]^-\left(\pf_F - \pf_T\right)^2}{} \right\}.
	\end{split}
	\end{equation}
	Considering the second cell term of \eqref{eqn:advection_stability_terms}, we notice that
	\begin{equation} \label{eqn:stability_gradient_factor}
	-(\RU_T \cdot \nabla \pf_T)\pf_T =  -\RU_T \cdot \nabla \left(\frac{1}{2}\pf_T^2 \right).
	\end{equation}
	Next we make use of the identity $(a-b)a = a^2 - ab = \frac{1}{2}(a^2 - b^2) + \frac{1}{2}(a-b)^2$
	to rewrite the first face term of \eqref{eqn:advection_stability_terms} as
	\begin{equation} \label{eqn:stability_flux_inequality}
	\begin{split}
	-\lip{F}{(\fv_{TF} (\pf_F - \pf_T))}{\pf_T} &= \frac{1}{2} \lip{F}{\fv_{TF} (\pf_T^2 - \pf_F^2)}{} + \frac{1}{2} \lip{F}{\fv_{TF} (\pf_T - \pf_F)^2}{}.
	\end{split}
	\end{equation}
	
	Inspired by \eqref{eqn:stability_gradient_factor} and \eqref{eqn:stability_flux_inequality}, we set for any cell $T\in\scriptT_h$,
	\begin{equation}\label{def:wtilde}
	\widetilde{\tf}_{T} = \frac{1}{2} \left( \pf_T^2, \left( \pf_F^2 \right)_{F\in\scriptF_T} \right) \in \dof^{2k}_T,
	\end{equation}
	and use the conservation of the fluxes (Theorem \ref{thm:conservation_of_the_discrete_fluxes}) with $\widetilde{\tf}_T$
instead of $\tf_T$ to write
	\begin{equation} \label{eqn:stability_conservation_term}
	\begin{split}
	 - \frac{1}{2}&\lipd{T}{\RU_T}{\nabla \left(\pf_T^2 \right)} + \frac{1}{2} \sum_{F\in\scriptF_T} \lip{F}{\fv_{TF} (\pf_T^2 - \pf_F^2)}{}  \\
	&= - \lipd{T}{\RU_T}{\nabla \widetilde{\pf}_T } + \sum_{F\in\scriptF_T} \lip{F}{\fv_{TF} (\widetilde{\pf}_T - \widetilde{\pf}_F)}{}  = \dbf_{T,\bkappa^\hs} (\dpdof_{T}^\hs, \widetilde{\tf}_{T}).
	\end{split}
	\end{equation}
	Summing over the cells and using the fact that $\dpdof_{T}^\hs$ solves the discrete pressure equation \eqref{eqn:discrete_pressure_equation}, we deduce that
	\begin{equation} \label{eqn:stability_conservation_term2}
	\begin{split}
	- \frac{1}{2} \sum_{T\in\scriptT_h} &\lipd{T}{\RU_T}{\nabla \left(\pf_T^2 \right)} + \frac{1}{2} \sum_{T\in\scriptT_h} \sum_{F\in\scriptF_T} \lip{F}{\fv_{TF} (\pf_T^2 - \pf_F^2)}{}  \\
	&= l_h^{p,\hs}(\widetilde{\tf}_h) 
	= \sum_{T\in\scriptT_h} \frac{1}{2} \lip{T}{(q^+(t^\hs, \cdot) - q^-(t^\hs, \cdot))}{\pf_T^2}.
	\end{split}
	\end{equation}
	
	Gathering the results from \eqref{eqn:stability_flux_inequality}--\eqref{eqn:stability_conservation_term2} and substituting into \eqref{eqn:advection_stability_terms}, we have
	\begin{align}
	&\dbf_{h,R, \RU}(\tf_h, \tf_h)  \nonumber\\
	&= \sum_{T\in\scriptT_h} \left\{ \lip{T}{\frac{2 \Phi}{\Delta t}}{\pf_T^2} + \frac{1}{2} \lip{T}{(q^+(t^\hs, \cdot) - q^-(t^\hs, \cdot))}{\pf_T^2} + \lip{T}{q^-(t^\hs,\cdot) \pf_T^2}{} \right\}  \nonumber \\
	&+ \sum_{T\in\scriptT_h} \sum_{F \in \scriptF_T} \left\{\frac{1}{2} \lip{F}{\fv_{TF} (\pf_T - \pf_F)^2}{} + \lip{F}{[\fv_{TF}]^-\left(\pf_F - \pf_T\right)^2}{} \right\}.
 \label{eqn:stability_after_conservation}
	\end{align}
	Combining the second and third cell terms of \eqref{eqn:stability_after_conservation}, and using $\frac{1}{2} \fv_{TF} + [\fv_{TF}]^- = \frac{1}{2} |\fv_{TF}|$ and the non-negativity of $q^+$ and $q^-$, we find
	\begin{equation*}
	\begin{split}
	\dbf_{h,R, \RU}(\tf_h, \tf_h) ={}& \sum_{T\in\scriptT_h} \left\{ \lip{T}{\frac{2 \Phi}{\Delta t}}{\pf_T^2} + \frac{1}{2} \lip{T}{(q^+(t^\hs, \cdot) + q^-(t^\hs, \cdot))}{\pf_T^2} \right\}   \\
	&+ \sum_{T\in\scriptT_h} \sum_{F \in \scriptF_T} \frac{1}{2} \lip{F}{|\fv_{TF}| (\pf_T - \pf_F)^2}{}\\
	\geq{}& \lip{\Omega}{\frac{2\Phi}{\Delta t}}{} \pf_h^2 + \sum_{T\in\scriptT_h} \sum_{F \in \scriptF_T} \frac{1}{2} \lip{F}{|\fv_{TF}| (\pf_T - \pf_F)^2}{}.
	\end{split}
	\end{equation*}
	The proof is complete.
\end{proof}

\begin{remark}[Order $2k$ on the pressure, and choice of the Darcy fluxes]\label{rem:2k.flux}
The reason for discretising the pressure equation with an HHO scheme of order $2k$, instead of $k$, 
is found in \eqref{eqn:stability_conservation_term}. Obtaining this relation requires
the usage of $\widetilde{\tf}_T$, defined by \eqref{def:wtilde} and belonging to $\dof^{2k}_T$,
into \eqref{eqn:conservation_of_the_discrete_fluxes}.

Equation \eqref{eqn:stability_conservation_term} is an essential component of the
stability proof, and it also justifies our choice of Darcy flux and volumetric velocity \eqref{eqn:discrete_fluxes}
and \eqref{eqn:darcy_velocity_volumetric_reconstruction}.
\end{remark}

We are now ready to prove the existence, uniqueness and stability of the solution to the scheme.

\begin{proof}[Proof of Theorem \ref{thm:stability_of_crank_nicolson}]
	Let us first assume that we have a solution $(\dpdof_h,\dc_h)$ to the scheme, and
let us prove the \emph{a priori} estimate \eqref{est:cN}.
	By \eqref{eqn:discrete_advection_diffusion_equation}, we have for all $\tf_h \in \dof^k_h$
	\begin{equation}
	\dbf_{h,\dD, R, \RU}(\dch_h, \tf_h) = l^{c,\hs}_h(\tf_h).
	\end{equation}
	Select $\tf_h = \dch_h$ as the test function and expand by the definitions of the discrete linear forms to write
	\begin{equation*}
	\dbf_{h,\dD}(\dch_h, \dch_h) + \dbf_{h,R, \RU}(\dch_h, \dch_h) = l^{c,\hs}_h(\dch_h).
	\end{equation*}
	Using Lemma \ref{lemma:diffusive_bilinear_form_is_positive} (with $\bLambda=\dD$)
	and Lemma \ref{lemma:advective_bilinear_form_is_positive},
	the definition \eqref{eqn:concentration_rhs_linear_form} of $l^{c,\hs}$ then yields
	\begin{equation}\label{est:for_existence}
	\begin{split}
	\lip{\Omega}{\frac{2\Phi (\pch_h)^2}{\Delta t}}{} 
	+\sum_{T\in\scriptT_h}\sum_{F\in\scriptF_h} \frac{\beta}{h_F}\|\pc_F^\hs-\pc_T^\hs \|_{L^2(F)}^2\\
	\leq \lip{\Omega}{\left(q^+(t^\hs, \cdot) \hat{c}(t^\hs, \cdot) + \frac{2\Phi}{\Delta t}\pc_h^n\right)}{\pch_h}
	\end{split}
	\end{equation}
	where $\beta>0$ is a coercivity constant of $\dD$ ($\beta$ depends on $\Phi$, $d_m$, $d_l$ and $d_t$).
	Gathering the time-stepping terms together and dropping the second term in the left-hand side, we then write
	\begin{equation} \label{eqn:stability_inequality}
	\lip{\Omega}{\frac{2\Phi (\pch_h - \pc^n_h)}{\Delta t}}{\pch_h} \leq  \lip{\Omega}{\left(q^+(t^\hs, \cdot) \hat{c}(t^\hs, \cdot)\right)}{\pch_h}.
	\end{equation}
	Recalling the definition of the half time-stepped concentration \eqref{eqn:half_time_concentration}, we easily deduce
	\[
	\lip{\Omega}{\frac{2\Phi (\pch_h - \pc^n_h)}{\Delta t}}{\pch_h}
	= \lip{\Omega}{\Phi \frac{(\pc^{n+1}_h)^2 - (\pc^n_h)^2}{2\Delta t}}{}.
	\]
	Hence, using the Cauchy--Schwarz and Young's inequalities in the right-hand side of \eqref{eqn:stability_inequality} yield,
	for any $\varepsilon>0$,
	\[
	\begin{split}
	\lip{\Omega}{\Phi \frac{(\pc^{n+1}_h)^2 - (\pc^n_h)^2}{2\Delta t}}{} \leq{}&  \| q^+(t^\hs, \cdot) \hat{c}(t^\hs, \cdot) \|_{L^2(\Omega)} \|\pch_h \|_{L^2(\Omega)}\\
	\leq{}&  \frac{1}{2\varepsilon}\| q^+(t^\hs, \cdot) \hat{c}(t^\hs, \cdot) \|_{L^2(\Omega)}^2 
+ \frac{\varepsilon}{2}\|\pch_h \|_{L^2(\Omega)}^2.
	\end{split}
	\]
	Summing over the time steps $n=0,\ldots,N-1$, the sum telescopes in the left-hand side. Using $|\widehat{c}|\le 1$
	and, by convexity of the square function,
	\[
	(\pch_h)^2=\left(\frac{\pc_h^n+\pc_h^{n+1}}{2}\right)^2\le \frac{(\pc_h^n)^2+(\pc_h^{n+1})^2}{2}
	\]
	we infer
	\[
	\begin{split}
	\lip{\Omega}{\Phi \frac{(\pc^{N}_h)^2 - (\pc^0_h)^2}{2\Delta t}}{}
	\le{}& \frac{N}{2\varepsilon}\| q^+\|_{L^\infty(0,t_f;L^2(\Omega))}^2 \\
	&	+ \frac{\varepsilon}{4}\sum_{n=0}^{N-1}\left(\|\pc_h^n \|_{L^2(\Omega)}^2+\|\pc_h^{n+1} \|_{L^2(\Omega)}^2\right)\\
	\le{}& \frac{N}{2\varepsilon}\| q^+\|_{L^\infty(0,t_f;L^2(\Omega))}^2
	+ \frac{\varepsilon}{2}\sum_{n=0}^{N}\|\pc_h^n \|_{L^2(\Omega)}^2.
	\end{split}
	\]
	Applying the boundedness of $\Phi$ \eqref{eqn:permeability_assumptions},
	\[
	\int_\Omega{\Phi_* \frac{(\pc^{N}_h)^2}{2\Delta t}} \leq \int_\Omega{\Phi_*^{-1} \frac{(\pc^0_h)^2}{2\Delta t}} + \frac{N}{2\varepsilon} \|q^+\|^2_{L^\infty(0,t_f;L^2(\Omega))} + \frac{\varepsilon}{2} \sum_{n=0}^{N}\int_\Omega (\pc^n_h)^2.
	\]
	Multiplying both sides by $2\Delta t/\Phi_*$ and recalling that $N \Delta t = t_f$ yields
	\begin{equation} \label{eqn:stability_bounds3}
	\int_\Omega (\pc^{N}_h)^2 \leq \int_\Omega \frac{(\pc^0_h)^2}{\Phi_*^2} + \frac{t_f}{\varepsilon\Phi_* } \|q^+\|^2_{L^\infty(0,t_f;L^2(\Omega))} + \frac{\varepsilon t_f }{\Phi_*N} \sum_{n=0}^{N} \int_\Omega (\pc^n_h)^2.
	\end{equation}
Take $\varepsilon=\frac{\Phi_*}{2t_f}$, so that $\frac{\varepsilon t_f}{\Phi_* N}=\frac{1}{2N}<1$.
Applying the Gronwall inequality of \cite[Lemma 5.1]{HR90} yields
\[
\| \pc^N_h\|_{L^2(\Omega)}^2\le \exp\left(\frac{1}{2N}\sum_{n=0}^N \frac{2N}{2N-1}\right)
\left(\frac{\|\pc^0_h\|_{L^2(\Omega)}^2}{\Phi_*^2} + \frac{t_f}{\varepsilon\Phi_* } \|q^+\|^2_{L^\infty(0,t_f;L^2(\Omega))}\right).
\]
The proof of \eqref{est:cN} is complete since $\frac{t_f}{\varepsilon\Phi_* }=\frac{2t_f^2}{\Phi_*^2}$
and $\frac{1}{2N}\sum_{n=0}^N \frac{2N}{2N-1}=\frac{N+1}{2N-1}\le 2$. The estimate was obtained
for $\pc_h^N$ but the same reasoning shows that it holds for $\pc_h^n$ for all $n=0,\ldots,N$.

\medskip

The existence and uniqueness of $(\dpdof_h,\dc_h)$ follows easily. At each iteration of Algorithm \ref{algo1},
$\dpdof_h^\hs$ is sought as a solution of the linear system \eqref{eqn:discrete_pressure_equation}--\eqref{eqn:discrete_pressure_equation:norm}.
If $l^{p,\hs}_h=0$, plugging $\tf_h=\dpdof_h^\hs$ in \eqref{eqn:discrete_pressure_equation} and
using Lemma \ref{lemma:diffusive_bilinear_form_is_positive} shows that the only possible solution to
this linear system is zero
(by \eqref{eqn:diffusive_bilinear_form_seminorm_inequality}, all cell unknowns must be constant and, working from
neighbour to neighbour, equal to all face unknowns and to all other cell unknowns; then
\eqref{eqn:discrete_pressure_equation:norm} fixes this constant uniform value to zero).
Hence, the matrix corresponding to \eqref{eqn:discrete_pressure_equation}--\eqref{eqn:discrete_pressure_equation:norm} is invertible, which means that this system has a unique solution $\dpdof_h^\hs$ at each time step.
After the pressure is fixed, $\dc_h^\hs$ is sought as a solution to the linear
equation \eqref{eqn:discrete_advection_diffusion_equation}. If $l^{c,\hs}_h=0$,
the right-hand side of \eqref{est:for_existence} vanishes, which shows that all cell and face
degrees of freedom are equal to 0.
Hence, the matrix of \eqref{eqn:discrete_advection_diffusion_equation} has a trivial
kernel, which shows the existence and uniqueness of $\dc_h^\hs$ solution to this equation.

\end{proof}

\section{Appendix: Implementation of the Scheme}\label{appen:implementation}

We present here algorithms for computing the local operators that define the numerical scheme. A fully functional implementation of the scheme in \texttt{C++} along with all of the tests present in Section \ref{sec:tests} can be found at 

\centerline{
\href{https://github.com/DanielLiamAnderson/hho-peaceman}{\texttt{https://github.com/DanielLiamAnderson/hho-peaceman}}.}

 The code for handling the mesh is a preliminary version of what later became the \texttt{DiSk++} library by Cicuttin et al. \cite{CICUTTIN2017}.

\subsection{A basis for the function spaces}

In order to realise the algorithms for solving the pressure and concentration equations, we first need to express our function spaces and test functions concretely. We recall the space $\Poly^m(K)$ of polynomials of degree $\le m$ over the domain $K$, and decompose its elements in terms of the following basis functions. Denote by $(\bar{x}_K, \bar{y}_K)$ the centre of mass of $K$, and by $h_K$ its diameter. The basis functions for 2D elements $T \in \scriptT_h$ are given by
\begin{equation} \label{eqn:polynomial_basis}
\varphi^T_{r,s} : T \to \R, \qquad \varphi^T_{r,s}(x,y) = \left( \frac{x - \bar{x}_T}{h_T} \right)^{r} \left( \frac{y - \bar{y}_T}{h_T} \right)^{s},
\end{equation}
for $r,s \geq 0$ and $r + s \leq m$. The basis functions for 1D elements $F \in \scriptF_h$ are given by
\begin{equation}
\varphi^F_r : F \to \R, \qquad \varphi^F_r(x,y) = \left( \frac{(x-\bar{x}_F)(x_0 - \bar{x}_F)}{h_F^2} + \frac{(y-\bar{y}_F)(y_0 - \bar{y}_F)}{h_F^2} \right)^r,
\end{equation}
for $0 \leq r \leq m$ where $(x_0,y_0)$ is one of the endpoints of $F$. We then take, for each cell $T \in \scriptT_h$ and face $F \in \scriptF_h$, the set of all basis functions covering the entire mesh, where each function is extended to $\Omega$ by defining $\varphi = 0$ outside its initial domain:
\begin{equation}
\scriptB^m = \left(\bigcup_{T \in \scriptT_h} \{\varphi^T_{r,s}\}_{r+s\leq m} \right) \bigcup \left( \bigcup_{F \in \scriptF_h} \{\varphi^F_{r}\}_{r \leq m} \right).
\end{equation}

The algorithms for solving the pressure and concentration equations then become square linear systems by evaluating the discrete equations \eqref{eqn:discrete_pressure_equation} and \eqref{eqn:discrete_advection_diffusion_equation} at all basis functions $\tf_h \in \scriptB^m$. Since each basis function is only non-zero in a single element of the meshed domain, the resulting linear systems will be sparse. The following notation are referred to throughout the implementation.
\begin{center}
	\begin{tabular}{ l l }
		$\scriptB^m$	& the set of all basis functions over mesh elements up to degree $m$, \\
		$\scriptB^{m}_K$	&	the basis functions on the cell or face $K$ of degree up to $m$,  \\
		$\scriptB^{m,1}_K$	&	the basis functions on the cell or face $K$ of degree at least $1$, up to $m$,  \\
		$\scriptB^{m}_{\underline{T}}$	&	the basis functions on the cell $T$ and all adjacent faces of degree up to $m$.
	\end{tabular}
\end{center}

\begin{remark}
	The space $\scriptB^{m,1}_K$ is useful when considering gradients of high-order basis functions, since gradients of degree zero basis functions are identically zero.
\end{remark}

\subsection{Numerical quadrature}

The assembly of the scheme matrices requires the numerical integration of products of arbitrary order polynomial basis functions. To ensure no loss of accuracy or stability, sufficiently accurate numerical quadrature rules must be used. Most of the integrals that we are required to evaluate consist of the product of two or three degree $m$ polynomials. Based on this observation, for polynomial degrees of freedom of order $m$, we employ numerical quadrature schemes that are exact for polynomials of degree up to $3m$. To integrate cell polynomials, we use the numerical quadrature schemes introduced by Dunavant in \cite{dunavant1985high}. The Dunavant quadrature rules provide exact integrals for fixed degree polynomial functions on triangular domains so we will split each cell of the mesh into triangular sub-elements such that each face of the cell corresponds to one sub-element. Quadrature for edge polynomials is performed using standard Gaussian quadrature rules for one-dimensional domains.

\begin{figure}[H]
	\centering
	\begin{tikzpicture}
	\draw (0,0) -- (1,2);
	\draw (1,2) -- (5,2);
	\draw (5,2) -- (6,-1);
	\draw (6,-1) -- (3,-2);
	\draw (3,-2) -- (0,0);
	
	\draw[fill] (2.9,0.2) circle [radius=0.1];
	\node [below left] at (2.9,0.2) {$\bar{x}_T$};
	
	\draw [decorate,decoration={brace,amplitude=10pt}]
	(1,2.1) -- (5,2.1) node [black,midway,yshift=0.6cm]
	{\footnotesize $F$};
	
	\draw (3, 1.2) node {$K_{TF}$};
	
	\draw [dashed] (0,0) --  (2.9,0.2);
	\draw [dashed] (1,2) --  (2.9,0.2);
	\draw [dashed] (5,2) --  (2.9,0.2);
	\draw [dashed] (6,-1) --  (2.9,0.2);
	\draw [dashed] (3,-2) --  (2.9,0.2);
	\end{tikzpicture}
	\caption{A cell split into triangular sub-elements to facilitate numerical quadrature.} \label{fig:split_cell_into_subelements}
\end{figure}
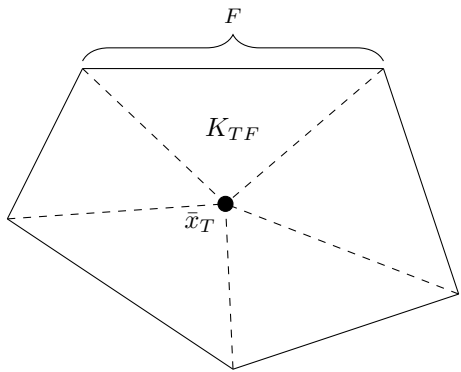

In the algorithms that follow, $K_{TF}$ stands for the triangular sub-element of the cell $T$ adjacent to the edge $F \in \scriptF_T$ as depicted in Figure \ref{fig:split_cell_into_subelements}.

\subsection{Assembly of the local operators for the pressure scheme}

The majority of the integrals in \eqref{eqn:discrete_pressure_equation} will be zero, so the local operators $\dbf_{T,\bLambda}$ can be efficiently assembled for each cell and then combined to form a sparse matrix for the scheme. These local constructions are outlined in the following algorithms. Algorithm \ref{algo:local_reconstruction} demonstrates the computation of local gradient reconstruction operator $\nabla \ro^{m+1}_{T, \bLambda}$ by inverting a high-order mass matrix of the basis gradient functions. The local gradient reconstruction is used in the assembly of the local diffusion operator $\dbf_{T,\bLambda}$ and in the reconstruction of the fluxes from the solution to the pressure equation. Algorithm \ref{algo:local_diffusion_bilinear_form} outlines the assembly of the local diffusion operator. The key step is the computation of the projections from the high-order correction, which is accomplished by computing a mass matrix of mixed high and low-order basis functions. Finally, Algorithm \ref{algo:local_load_vector} presents the implementation of the pressure source term, i.e.\ the right-hand side of the equation.

\subsection{Computation of the numerical fluxes}

We now present a simple algorithm to compute the numerical fluxes $\fv_{TF}$, that does not require us to
compute the whole right-hand side of \eqref{eqn:discrete_fluxes}. Let $T \in \scriptT_h$ and use the conservativity of the fluxes (Theorem \ref{thm:conservation_of_the_discrete_fluxes}) with a cell-absent test function $\tf_{T}=(0,\tf_F) \in \{0\}\times\dof^{2k}_{\partial T}$ to obtain
\begin{equation}\label{bal.eq.algo}
\dbf_{T,\bkappa}(\dpdof_{T}, \tf_{T}) = - \sum_{F\in\scriptF_T} (\fv_{TF}, \pf_F)_{L^2(F)}.
\end{equation}
Select an edge $F \in \scriptF_T$, then write the flux $\fv_{TF}$ in terms of the basis $\scriptB^{2k}_F$ like so
\begin{equation}\label{eqn:computed_fluxes}
\fv_{TF} = \sum_{i=0}^{|\scriptB^{2k}_F| - 1} \lambda^F_i \varphi^F_i.
\end{equation}
The balance equation \eqref{bal.eq.algo} now reads
\begin{equation} \label{eqn:balance_of_basis_functions}
\dbf_{T,\bkappa}(\dpdof_T, \tf_T) 
= - \sum_{i=0}^{|\scriptB^{2k}_F| - 1} \lambda_i^F (\varphi^F_i, \pf_F)_{L^2(F)},
\end{equation}
which results in a square linear system for $\{\lambda^F_i\}_i$ in terms of the basis functions ${\scriptB^{2k}_F}$. The construction of the numerical fluxes is shown in Algorithm \ref{algo:flux_reconstruction}.

\subsection{Assembly of the local operators for the concentration scheme}

The computation of the advective derivative is similar to that of the local gradient reconstruction. Algorithm \ref{algo:local_advective_derivative} shows how to efficiently compute $\ad_{T,\RU}^k$ for each cell $T \in \scriptT_h$ by solving a local problem involving the mass matrix of the cell. The local advection-reaction operator is computed as shown in Algorithm \ref{algo:local_advection_reaction_bilinear_form} by assembling together the advective derivatives of the local test functions combined with a mass-reaction matrix that accounts for the value of the reaction terms. Lastly, we present Algorithm \ref{algo:concentration_local_load_vector}, which computes the right-hand side of the concentration equation.

\vfill
\thanks{\textbf{Acknowledgement}: this research was supported by the Australian Government through the Australian Research Council's Discovery Projects funding scheme (pro\-ject number DP170100605).
	The authors would also like to thank Daniele Di Pietro and Matteo Cicuttin for giving us access
	to the {\sf hho} software platform\footnote{\emph{Agence pour la Protection des Programmes} deposit number IDDN.FR.001.220005.000.S.P.2016.000.10800}, which served as an invaluable starting point for our implementation
 and whose development was funded by Agence Nationale de la Recherche project HHOMM ANR-15-CE40-0005.
}

\begin{algorithm}[h] 
	\caption{Computation of the local gradient reconstruction matrix\label{algo:local_reconstruction}} 
	\begin{algorithmic}[1] 
		\State {Set $M_{\varphi_i, \varphi_j}=0$ for all $\varphi_i, \varphi_j \in \scriptB^{m+1,1}_{T}$}
		\State {Set $B_{\varphi_i, \varphi_j}=0$ for all $\varphi_i \in \scriptB^{m+1,1}_{T}, \varphi_j \in \scriptB^{m}_{\underline{T}}$}
		\For{each edge $F \in F_T$}
		\LineComment{Compute volumetric terms corresponding to the sub-element of $F$}
		\For{each quadrature node $(x, w_x)$ of $K_{TF}$}
		
		\For{each basis function $\varphi_i \in \scriptB^{m+1,1}_T$}
		
		\For{each basis function $\varphi_j \in \scriptB^{m+1,1}_T$}
		\State $M_{\varphi_i, \varphi_j} \leftarrow M_{\varphi_i, \varphi_j} + w_x (\bLambda(x) \nabla \varphi_i(x)) \cdot \nabla \varphi_j(x)$
		\EndFor
		
		\For{each basis function $\varphi_j \in \scriptB^{m}_T$}
		\State $B_{\varphi_i, \varphi_j} \leftarrow B_{\varphi_i, \varphi_j} + w_x (\bLambda(x) \nabla \varphi_i(x)) \cdot \nabla \varphi_j(x)$
		\EndFor
		
		\EndFor
		
		\EndFor
		
		\LineComment{Compute edge terms}
		\For{each quadrature node $(x, w_x)$ of $F$}
		\For{each basis function $\varphi_i \in \scriptB^{m+1,1}_T$}
		
		\For{each basis function $\varphi_j \in \scriptB^{m}_F$}
		\State $B_{\varphi_i, \varphi_j} \leftarrow B_{\varphi_i, \varphi_j} + w_x \nabla \varphi_i(x) \cdot (\bLambda(x) \nv_{TF}) \varphi_j(x)$
		\EndFor
		\For{each basis function $\varphi_j \in \scriptB^{m}_T$}
		\State $B_{\varphi_i, \varphi_j} \leftarrow B_{\varphi_i, \varphi_j} - w_x \nabla \varphi_i(x) \cdot (\bLambda(x) \nv_{TF}) \varphi_j(x)$
		\EndFor
		
		\EndFor
		\EndFor
		
		\EndFor
		\State {Set $G = M^{-1} B$} 		\Comment{Solve for the gradient reconstruction}
	\end{algorithmic} 
\end{algorithm}

\begin{algorithm}[h] 
	\caption{Assembly of the local diffusion matrix\label{algo:local_diffusion_bilinear_form}} 
	\begin{algorithmic}[1] 
		\LineComment{Compute mass matrices}
		\State {Set $M^{TT}_{\varphi_i,\varphi_j} = 0$ for all $\varphi_i, \varphi_j \in \scriptB^{m+1}_T$}
		\State {Set $M^{TF}_{\varphi_i,\varphi_j} = 0$ for all $\varphi_i \in \scriptB^{m}_F,\ \varphi_j \in \scriptB^{m+1}_T$ for all edges $F \in \scriptF_T$}
		\State {Set $M^{FF}_{\varphi_i,\varphi_j} = 0$ for all $\varphi_i, \varphi_j \in \scriptB^{m}_F$ for all edges $F \in \scriptF_T$}
		\For{each edge $F \in F_T$}
		
		\LineComment{Compute cell-on-cell mass matrix}
		\For{each quadrature node $(x, w_x)$ of $K_{TF}$}
		\For{each basis function $\varphi_i \in \scriptB^{m+1}_T$}
		\For{each basis function $\varphi_j \in \scriptB^{m+1}_T$}
		\State $M^{TT}_{\varphi_i, \varphi_j} \leftarrow M^{TT}_{\varphi_i, \varphi_j} + w_x \varphi_i(x) \varphi_j(x)$
		\EndFor
		\EndFor
		\EndFor
		
		\LineComment{Compute cell-on-edge and edge-on-edge mass matrices}
		\For{each quadrature node $(x, w_x)$ of $F$}
		\For{each basis function $\varphi_i \in \scriptB^{m}_F$}
		\For{each basis function $\varphi_j \in \scriptB^{m+1}_T$}
		\State $M^{TF}_{\varphi_i, \varphi_j} \leftarrow M^{TF}_{\varphi_i, \varphi_j} + w_x \varphi_i(x) \varphi_j(x)$
		\EndFor
		\EndFor
		
		\For{each basis function $\varphi_i \in \scriptB^{m}_F$}
		\For{each basis function $\varphi_j \in \scriptB^{m}_F$}
		\State $M^{FF}_{\varphi_i, \varphi_j} \leftarrow M^{FF}_{\varphi_i, \varphi_j} + w_x \varphi_i(x) \varphi_j(x)$
		\EndFor
		\EndFor
		\EndFor
		\EndFor
		
		\LineComment{Compute the volumetric term}
		\State {Compute the gradient reconstruction $G$ (Algorithm \ref{algo:local_reconstruction})}
		\State {Set $A = B^{tr} G$} \Comment{$B^{tr}$ is the transpose of $B$ from Algorithm \ref{algo:local_reconstruction}}
		
		\LineComment{Compute the local reconstruction cell projection matrix}
		\State {Set $M^{TT, m}_{\varphi_i, \varphi_j} = M^{TT}_{\varphi_i, \varphi_j}$ for all $\varphi_i, \varphi_j \in \scriptB^{m}_T$} 
		\State {Set $M^{TT, m + 1}_{\varphi_i, \varphi_j} = M^{TT}_{\varphi_i, \varphi_j}$ for all $\varphi_i \in \scriptB^{m}_T, \varphi_j \in \scriptB^{m+1,1}_T$} 
		
		\State {Solve $M^{TT, m} P_T = M^{TT, m+1}$ for $P_T$}
		
		\LineComment{Compute the edge terms}
		\For{each edge $F \in F_T$}
		\LineComment{Compute the edge projection matrix}
		\State {Set $P_F = (M^{FF})^{-1}$}
		\State {Set $M^{TF,m}_{\varphi_i, \varphi_j} = M^{TF}_{\varphi_i, \varphi_j}$ for all $\varphi_i \in \scriptB^{m}_F, \varphi_j \in \scriptB^{m}_T$} 
		\State {Set $M^{TF,m+1}_{\varphi_i, \varphi_j} = M^{TF}_{\varphi_i, \varphi_j}$ for all $\varphi_i \in \scriptB^{m}_F, \varphi_j \in \scriptB^{m+1,1}_T$} 
		\LineComment{Compute projections}
		\State {Set $B_F = P_F M^{TF,m+1} G - I_F$} \Comment{$I_F=$ identity matrix on the edge terms}
		\State {Set $B_T = P_F M^{TF,m} (I_T - P_T)$} \Comment{$I_T=$ identity matrix on the cell terms}
		\State {Set $B_{RF} = B_F + B_T$}
		\LineComment{Compute local stabilisation terms}
		\State $A \leftarrow A + \frac{\bLambda_{TF}}{h_F} B_{RF}^{tr} M^{FF} B_{RF}$
		\EndFor 
	\end{algorithmic} 
\end{algorithm}

\begin{algorithm}[h] 
	\caption{Assembly of the pressure source vector\label{algo:local_load_vector}} 
	\begin{algorithmic}[1] 
		\State Set $b_{\varphi_i} = 0$ for all $\varphi_i \in \scriptB^{2k}_{\underline{T}}$
		\For{each edge $F \in F_T$}
		\For{each quadrature node $(x, w_x)$ of $K_{TF}$}
		\For{each basis function $\varphi_i \in \scriptB^{2k}_T$}
		\State $b_{\varphi_i} \leftarrow b_{\varphi_i} + w_x \varphi_i(x) (q^+(t^\hs, x) - q^-(t^\hs, x))$
		\EndFor
		\EndFor
		\EndFor
	\end{algorithmic} 
\end{algorithm}

\begin{algorithm}[h] 
	\caption{Computation of the local fluxes\label{algo:flux_reconstruction}} 
	\begin{algorithmic}[1]
		\State {Compute the local diffusion operator matrix $A$ with $m = 2k$ (Algorithm \ref{algo:local_diffusion_bilinear_form})}
		\For{each edge $F \in \scriptF_T$}
		\State Set $\alpha_{\varphi_i} = 0$ for all $\varphi_i \in \scriptB^{2k}_F$
		
		\For{each basis function $\varphi_i \in \scriptB^{2k}_F$}
		\State $\alpha_{\varphi_i} \leftarrow - (\dpdof_T \cdot A_{:, \varphi_i})$	\Comment{$A_{:,\varphi_i}$ is the column of $A$ corresponding to $\varphi_i$}
		\EndFor
		
		\LineComment{Build the Gram matrix}
		\State {Set $G_{\varphi_i, \varphi_j} = 0$ for all $\varphi_i, \varphi_j \in \scriptB^{2k}_F$}
		\For{each quadrature node $(x, w_x)$ on $F$}
		\For{each basis function $\varphi_i \in \scriptB^{2k}_F$}
		\For{each basis function $\varphi_j \in \scriptB^{2k}_F$}
		\State $G_{\varphi_i, \varphi_j} \leftarrow G_{\varphi_i, \varphi_j} + w_x \varphi_i(x) \varphi_j(x)$
		\EndFor
		\EndFor
		\EndFor
		
		\State {Set $\lambda^F = G^{-1} \alpha$}
		\EndFor
		\State {Compute $\fv_{TF}$ in terms of $\lambda^F$ as in \eqref{eqn:computed_fluxes}}
	\end{algorithmic} 
\end{algorithm}

\begin{algorithm}[h] 
	\caption{Assembly of the local advective-reactive derivative matrix\label{algo:local_advective_derivative}} 
	\begin{algorithmic}[1] 
		\State {Set $M_{\varphi_i, \varphi_j} = 0$ for all $\varphi_i, \varphi_j \in \scriptB^k_T$}
		\State {Set $B_{\varphi_i, \varphi_j} = 0$ for all $\varphi_i \in \scriptB^k_T, \varphi_j \in \scriptB^k_{\underline{T}}$}
		
		\For{each edge $F \in \scriptF_T$}
		\LineComment{Compute the volumetric terms}
		\For{each quadrature node $(x, w_x)$ in $K_{TF}$}
		\For{each basis function $\varphi_i \in \scriptB^k_T$}
		\For{each basis function $\varphi_j \in \scriptB^k_T$}
		\State $M_{\varphi_i, \varphi_j} \leftarrow M_{\varphi_i, \varphi_j} + w_x \varphi_i(x) \varphi_j(x)$
		\State $B_{\varphi_i, \varphi_j} \leftarrow B_{\varphi_i, \varphi_j} + w_x (\fv_{TF}(x) \cdot \nabla \varphi_i(x)) \varphi_j(x)$
		\EndFor
		\EndFor
		\EndFor
		
		\LineComment{Compute the edge terms}
		\For{each quadrature node $(x, w_x)$ on $F$}
		\For{each basis function $\varphi_i \in \scriptB^k_T$}
		\For{each basis function $\varphi_j \in \scriptB^k_F$}
		\State $B_{\varphi_i, \varphi_j} \leftarrow B_{\varphi_i, \varphi_j} + w_x (\fv_{TF}(x) \cdot \nv_{TF}) \varphi_i(x) \varphi_j(x)$
		\EndFor
		\For{each basis function $\varphi_j \in \scriptB^k_T$}
		\State $B_{\varphi_i, \varphi_j} \leftarrow B_{\varphi_i, \varphi_j} - w_x (\fv_{TF}(x) \cdot \nv_{TF}) \varphi_i(x) \varphi_j(x)$
		\EndFor
		\EndFor
		\EndFor
		\EndFor
		\LineComment{Solve for the advective-reaction derivative}
		\State {Set $G = M^{-1}B$}
	\end{algorithmic} 
\end{algorithm}

\begin{algorithm}[h] 
	\caption{Assembly of the local advection-reaction matrix\label{algo:local_advection_reaction_bilinear_form}} 
	\begin{algorithmic}[1] 
		\State {Set $A_{\varphi_i,\varphi_j} = 0$ for all $\varphi_i,\varphi_j \in \scriptB^k_{\underline{T}}$}
		\State {Compute the advective-reactive derivative $G$ (Algorithm \ref{algo:local_advective_derivative})}
		\For{each edge $F \in \scriptF_T$}
		\LineComment{Compute the volumetric terms}
		\For{each quadrature node $(x, w_x)$ in $K_{TF}$}
		\For{each basis function $\varphi_i \in \scriptB^k_T$}
		\For{each basis function $\varphi_j \in \scriptB^k_T$}
		\State $A_{\varphi_i, \varphi_j} \leftarrow A_{\varphi_i, \varphi_j} + w_x \varphi_i(x) \varphi_j(x) \mu(x)$
		\EndFor
		\EndFor
		\For{each basis function $\varphi_i \in \scriptB^k_{\underline{T}}$}
		\LineComment{Compute $\ad_{T,\RU}^k \tf_{\underline{T}}$}
		\State {Set $v = 0$}
		\For{each basis function $\varphi_j \in \scriptB^k_T$}
		\State $v \leftarrow v + \varphi_j(x) G_{\varphi_j, \varphi_i}$
		\EndFor
		\For{each basis function $\varphi_j \in \scriptB^k_T$}
		\State $A_{\varphi_i, \varphi_j} \leftarrow A_{\varphi_i, \varphi_j} - w_x \varphi_j(x) v$
		\EndFor
		\EndFor
		\EndFor
		
		\LineComment{Compute the edge terms}
		\For{each quadrature node $(x, w_x)$ on $F$}
		\For{each basis function $\varphi_i \in \scriptB^k_T$}
		\For{each basis function $\varphi_j \in \scriptB^k_F$}
		\State $A_{\varphi_i, \varphi_j} \leftarrow A_{\varphi_i, \varphi_j} - w_x [\fv_{TF}(x) \cdot \nv_{TF}]^- \varphi_i(x) \varphi_j(x)$
		\State $A_{\varphi_j, \varphi_i} \leftarrow A_{\varphi_j, \varphi_i} - w_x [\fv_{TF}(x) \cdot \nv_{TF}]^- \varphi_i(x) \varphi_j(x)$
		\EndFor
		\For{each basis function $\varphi_j \in \scriptB^k_T$}
		\State $A_{\varphi_i, \varphi_j} \leftarrow A_{\varphi_i, \varphi_j} + w_x [\fv_{TF}(x) \cdot \nv_{TF}]^- \varphi_i(x) \varphi_j(x)$
		\EndFor
		\EndFor
		\For{each basis function $\varphi_i \in \scriptB^k_F$}
		\For{each basis function $\varphi_j \in \scriptB^k_F$}
		\State $A_{\varphi_i, \varphi_j} \leftarrow A_{\varphi_i, \varphi_j} + w_x [\fv_{TF}(x) \cdot \nv_{TF}]^- \varphi_i(x) \varphi_j(x)$
		\EndFor
		\EndFor
		\EndFor
		\EndFor
	\end{algorithmic} 
\end{algorithm}

\begin{algorithm}[h] 
	\caption{Assembly of the concentration source vector\label{algo:concentration_local_load_vector}} 
	\begin{algorithmic}[1] 
		\State Set $\bb_{\varphi_i} = 0$ for all $\varphi_i \in \scriptB^{k}_{\underline{T}}$
		\For{each edge $F \in F_T$}
		\For{each quadrature node $(x, w_x)$ of $K_{TF}$}
		\For{each basis function $\varphi_i \in \scriptB^{k}_T$}
		\State $b_{\varphi_i} \leftarrow b_{\varphi_i} + w_x \varphi_i(x) (q^+(t^\hs, x)\hat{c}(t^\hs, x) + \frac{2\Phi}{\Delta t} \pc^n_h(x))$
		\EndFor
		\EndFor
		\EndFor
	\end{algorithmic} 
\end{algorithm}

\clearpage
\bibliographystyle{abbrv}
\bibliography{ref}

\end{document}